\documentclass{article}
\usepackage{amssymb,amsfonts,amsmath,amsopn,color,graphicx,amsthm}
\usepackage{wrapfig}
\usepackage{wasysym}
\usepackage{comment}
\usepackage[colorlinks=true, pdfstartview=FitV, linkcolor=blue, citecolor=blue, urlcolor=blue]{hyperref}

\usepackage{filecontents}

\newtheorem{theorem}{Theorem}[section] 
\newtheorem{prop}[theorem]{Proposition}
\newtheorem{lemma}[theorem]{Lemma}
\newtheorem{cor}[theorem]{Corollary}
\newtheorem{defn}[theorem]{Definition}

\theoremstyle{definition}

\newtheorem*{remark}{Remark}

\def\C{{\mathbb C}}
\def\D{{\mathbb D}}
\def\eqref#1{(\ref{#1})}
\def\&{&\hspace{-20pt}}
\def\R{{\mathbb R}}
\def\d{{\rm d}}
\def\e{{\rm e}}
\def\i{{\rm i}}
\def\a{{\alpha}}

\def\ii{{\rm i}}
\def\res{\mathop{{\rm res}}}

\newcommand{\be}{\begin{eqnarray}}
\newcommand{\ee}{\end{eqnarray}}
\newcommand{\bes}{\begin{eqnarray*}}
\newcommand{\ees}{\end{eqnarray*}}
\newcommand{\ds}{\displaystyle}
\renewcommand{\c}[1]{\overline{#1}}
\newcommand{\order}[1]{\ensuremath{{\mathcal O}\left(#1\right)}}

\newcommand{\atopfrac}[2]{\genfrac{}{}{0pt}{}{#1}{#2}}

\definecolor{light-blue}{rgb}{0.8,0.85,1}
\definecolor{blue}{rgb}{0,0,1}
\definecolor{red}{rgb}{1,0,0}

\def\pmtwo#1#2#3#4{\left[\begin{array}{cc}\ds #1&\ds #2\\\ds #3&\ds #4\end{array}\right]}

\renewcommand{\theequation}{\arabic{section}.\arabic{equation}}
\makeatletter
\@addtoreset{equation}{section}
\makeatother

\title{Strong asymptotics of the orthogonal polynomials with respect to a measure supported on the plane}
\author{Ferenc Balogh, Marco Bertola, Seung Yeop Lee, Kenneth D. T-R McLaughlin}

\begin{document}
\maketitle
\begin{abstract}
We consider the orthogonal polynomials $\{P_{n}(z)\}$ with respect to the measure
\begin{equation*}
|z-a|^{2N c} \e^{-N |z|^2} \,\d A(z)
\end{equation*} over the whole complex plane.
We obtain the strong asymptotic of the orthogonal polynomials in the complex plane and the location of their zeros in a scaling limit where $n$ grows to infinity with $N$.   

The asymptotics are described in terms of three (probability) measures associated with the problem. 
The first measure is the limit of the counting measure of zeros of the polynomials, which is captured by the $g$-function much in the spirit of ordinary orthogonal polynomials on the real line.
The second measure is the equilibrium measure that minimizes a certain logarithmic potential energy, supported on a region $K$ of the complex plane. 
The third measure is the harmonic measure of $K^c$ with a pole at $\infty$.  This appears as the limit of the probability measure given (up to the normalization constant) by the squared modulus of the $n$-th orthogonal polynomial times the orthogonality measure, i.e. $|P_n(z)|^2 |z-a|^{2N c} \e^{-N |z|^2} \,\d A(z)$.

The compact region $K$ that is the support of the second measure undergoes a topological transition under the variation of the parameter $t=n/N$; in a double scaling limit near the critical point given by $t_c=a(a+2\sqrt c)$ we observe the Hastings-McLeod solution to Painlev\'e\ II in the asymptotics of the orthogonal polynomials.

\end{abstract}

\tableofcontents

\section{Introduction}\label{sec:intro}

In this paper, we study orthogonal polynomials with respect to a weight function that is supported on the whole complex plane. 
Let $P_n(z)=P_{n,N}(z)$ be the monic orthogonal polynomial such that 
\begin{equation}\label{2d_ortho}
\int_{\C}P_{n,N}(z)\,\c{P_{m,N}(z)} \,e^{-NQ(z)}\d A(z) = h_{n,N}\delta_{nm} \quad (n,m = 0,1, \dots),
\end{equation}
where $N$ is a positive real parameter, $Q:\C\to\R$ is called the {\em external potential}, $\d A(z)$ is the standard area measure on $\C$, and $h_n=h_{n,N}$ is the (positive) norming constant.

Our principal motivation is the connection between these orthogonal polynomials and so-called
``random normal matrices", see \cite{mehta_book} for a complete description of this
connection. One considers the space of $n\times n$ normal matrices (those that
commute with their adjoint) equipped with a probability measure of the form
\begin{equation}\label{eq:gibbs}
\frac{1}{Z_{N,n}} e^{ - N \sum_{j=1}^n Q(z_j) } \d M \ ,
\end{equation}
where the $z_j$'s are the $n$ complex eigenvalues of a normal matrix $M$ and
$\d M$ represents the measure induced on normal matrices by the Euclidean metric on
the space ${\C}^{n^2}$ of all $n \times n$ complex matrices. 
One of the
fundamental quests within random matrix theory is to understand the statistical behavior
of the eigenvalues in the limit when $n$ (and $N$) tend to infinity. 
Most natural statistical
quantities concerning eigenvalues can be expressed in terms of the associated orthogonal
polynomials.
In particular, the {\it averaged density of eigenvalues} is 
\begin{equation}\label{density}
\sigma_{n,N}(z):=\frac{1}{n} e^{- N Q(z)} \sum_{j=0}^{n-1}\frac{\left| P_{j,N}(z) \right|^{2}}{h_{j,N}}.
\end{equation}
Using this fact, once one has complete asymptotic control of the orthogonal polynomials, it is an interesting challenge to compute many different statistical properties of the eigenvalues in a concrete fashion as $n$ and $N$ tend to $\infty$.

There is only a handful of external potentials $Q$ for which these polynomials can be
explicitly computed, and the average density can be analyzed asymptotically, as $n$ and $N$ tend to $\infty$. 
Taking the simplest case, $Q(z) = |z|^{2}$, one obtains the $h_{j,N}$'s and the averaged density of eigenvalues $\sigma_{n,N}$ in terms of the Gamma function. For this choice of $Q$, the associated matrix model is called the Ginibre ensemble, and has been studied to great effect via the simple orthogonal polynomials $z^{n}$ (see for instance \cite{mehta_book,forrester2010log} and the appendix of \cite{wz_2d_dyson_gas}).
In fact, for any radially symmetric potential $Q(z) = Q(|z|)$, the orthogonal polynomials are the monomials $z^{n}$.

A harmonic deformation of the Gaussian case, i.e. $Q(z)=|z|^2+(\text{harmonic})$, is particularly interesting (because of its relation to a growth model called the Hele-Shaw flow \cite{wz_normal_matrix_growth}).
The simplest such deformation of the Gaussian normal matrix model uses the potential 
$
Q = |z|^{2} + A \left( z^{2} + \overline{z}^{2} \right).
$
The associated orthogonal polynomials are now Hermite polynomials \cite{difrancesco_laughlin}, for which the asymptotics are well known.   The case with {\em cubic} deformation was recently studied in \cite{BK2012}.

Concerning the limiting density of eigenvalues, there are other approaches \cite{elbau_felder,AHM}, which connect to a classical energy problem from potential theory.\footnote{In random Hermitian matrices, these are analogous to the saddle point approach from the physics community (see, for example \cite{BIPZ}), which is put on a solid mathematical footing originally by Johansson \cite{johansson_fluctuations}.}

Let us scale the real parameter $N$ with $n$ so that $\lim_{n\to\infty} n/N=t$ for some fixed positive number $t$.
The general fact \cite{elbau_felder,AHM} is that, under some basic assumptions on the potential $Q$, the (normalized) counting measure of eigenvalues, averaged over the probability measure \eqref{eq:gibbs}, converges as $n\to\infty$ (in the sense of measures) to the unique probability measure $\mu^*$ in the plane which minimizes the functional ${\cal L}(\mu)$:
\begin{equation}\label{functional}
{\cal L}(\mu)=t\int_\C Q(z)\,\d \mu(z) +t^2\iint_{\C\times\C}\log\frac{1}{|z-w|}\d \mu(z)\,\d\mu(w).
\end{equation}
This functional is the well-known Coulomb energy functional in $\mathbb{R}^{2}$, and has a well-developed classical variational theory \cite{saff_totik_book}. 
The measure $\mu^*$ is equivalently characterized in terms of the Euler-Lagrange variational conditions, 
which can be
summarized as follows: 
$\mu^*$ is said to satisfy the Euler-Lagrange conditions if there is a constant $\ell_\text{2D}$ so that 
\begin{equation}\label{ELeq}
Q(z)-2t\int_\C \log|z-w|\, \d \mu^*(w)+\ell_\text{2D}\geq 0,
\end{equation}
for all $z$ in $\C$ (up to a set of capacity zero), with equality on the support of the measure $\mu^*$ (again up to a set of capacity zero).   The left hand side of the above is denoted by ${\cal U}_\text{2D}$, the {\em effective potential of the Coulomb system}, and of which the smooth extension, ${\cal U}$, (defined at \eqref{calU}) will show up in Corollary \ref{cor-harm} in a conspicuous way.

It is a general fact \cite{saff_totik_book} that under very mild assumptions on the growth of $Q(z)$ the minimizing equilibrium measure $\mu^*$ is supported on a compact set. Moreover, for $Q(z)=|z|^2+(\text{harmonic function})$, it is also known that the {\em density} of the equilibrium measure is uniform over the support.

 In this paper we study one of the simplest cases,  given by
\begin{equation}\label{ourQ}
Q(z) := |z|^2 +2c\log{\frac{1}{|z-a|}}, \ \ c>0,\ a>0.
\end{equation}
The case of $a\in\C\setminus\{0\}$ can be reduced to the above by a rotation of the coordinate.  We will assume $a\neq 0$ (for $a=0$, one has $P_{n,N}(z)= z^{n}$).

This case has been studied in \cite{Balogh,wz_normal_matrix_growth} where the corresponding equilibrium measure is described.

In the following paragraphs we shall define the region $K\subset \C$ in terms of $a,c$ and $t$; these regions will be proved to be the support of the equilibrium measure in Proposition \ref{thm-postK} and Proposition \ref{thm-pre}.

\paragraph{Definition of $K$ and the related objects:}
 
The region $K$ that we are going to define will have area $\pi t$ and depend on the parameters $a,c,t$ of our problem.
As Figure \ref{fig-1} indicates, $K$ undergoes a topological transition at the critical time $t_c$ given by 
\begin{equation}
t_c:= a(a+2\sqrt c).
\end{equation}
The next ingredient is a smooth curve which will be the attractor set for the zeros of the orthogonal polynomials.  We will denote it by ${\cal B}$.  As we will define below, ${\cal B}$ also undergoes similar topological transition as the set $K$.

We separate our analysis into the following three cases:
\begin{itemize}
\item the {\bf pre-critical case}, $0<t<t_c$:  The set $K$ is defined to be the interior of the real analytic Jordan curve given by the image of the unit circle $\{v\,|\,|v|=1\}$ with respect to the rational map,
\begin{equation}\label{fv}
f(v)= \rho\, v -\frac{\kappa}{v-\alpha}-\frac{\kappa}{\alpha},
\end{equation}
whose (positive) parameters $\rho$, $\kappa$ and $\alpha$ are given in terms of $a$, $c$ and $t$ below. 
\begin{equation}\label{rhokappa}
\rho = \frac {t+ a^2 \alpha^2}{2 a\,\alpha} ,\qquad \kappa  = \frac {(1-\alpha^2)(t- a^2 \alpha^2)}{2 a\, \alpha}.
\end{equation}
The parameter $\alpha$ is given by the {\em unique} solution of $P(\alpha^2)=0$ where
\begin{equation}\label{000}
P(X):=X^3-\left(\frac{a^2+4c+2t}{2a^2}\right)X^2+\frac{t^2}{2a^4},
\end{equation}
such that $0<\alpha<1$ and $\kappa>0$.
The existence of $\alpha$ satisfying this definition is stated in Lemma \ref{lem-preA} and the comment right below the lemma.   The proof that $K$ is the support of the equilibrium measure is provided in Proposition \ref{thm-pre}.

We denote the {\em critical values} and {\em critical points} of $f$ by
\begin{equation}\label{betapre}
\beta:= f\left(v_\beta\right)=\alpha\rho-\frac{\kappa}{\alpha}+2i\sqrt{\kappa\rho} \text{ ~~ where ~~} v_\beta:=\alpha +i\sqrt{\frac{\kappa}{\rho}}.
\end{equation}
and their complex conjugates, $\overline\beta$ and $\overline{v_\beta}$. 

We define the smooth arc ${\cal B}$ by the following three conditions:
\begin{itemize}
\item ${\cal B}$ has the endpoints at $\beta$ and $\overline\beta$.
\item ${\cal B}$ intersects the negative real axis.
\item On the curve ${\cal B}$ the following equation is satisfied.
\begin{equation}\label{eq:qd}
\frac{a^2(z-b)^2(z-\beta)(z-\c\beta)}{(z-a)^2z^2}\,{\rm d} z^2<0 ,\quad b:=\frac{\alpha}{\rho},
\end{equation}
where ${\rm d}z$ is the standard differential.
\end{itemize}
The uniqueness and the existence of ${\cal B}$ satisfying these three conditions is shown in Lemma \ref{lem:calB} (and the remark below the lemma).  

We define $y(z)\,{\rm d}z$ by the square root of the quadratic differential in \eqref{eq:qd}, i.e.
\begin{equation}\label{eq:y}
	y(z):=\frac{a(z-b)\sqrt{(z-\beta)(z-\c\beta)}}{(z-a)z},\quad z\notin{\cal B},
\end{equation}
by choosing ${\cal B}$ as the branch cut, and taking the analytic branch that satisfies $\lim_{z\to\infty}y(z)=a$.

\item the {\bf post-critical case}, $t>t_c$:  We define $K$ to be the region between two circles, i.e.
\begin{equation}
	K=\overline{D(0,\sqrt{t+c})}\setminus D(a,\sqrt c)
\end{equation}
 where $D(A,B)$ stands for the disk of radius $B$ centered at $A$.
The proof that $K$ is the support of the equilibrium measure is provided in Proposition \ref{thm-postK}.

We define the two points $\beta<b$ on the real axis by 
\begin{equation}\label{bbetapost}
b=\frac{a^2+t+\sqrt{(t-a^2)^2-4a^2c}}{2a}\ ,
\qquad
\beta=\frac{a^2+t-\sqrt{(t-a^2)^2-4a^2c}}{2a}\ .
\end{equation}
We define the simple closed curve ${\cal B}$ by the following three conditions:
\begin{itemize}
\item ${\cal B}$ is a simple closed curve with $\beta\in{\cal B}$.
\item ${\cal B}$ has $0$ and $a$ in its interior.
\item On the curve ${\cal B}$ the following equation is satisfied.
\begin{equation}\label{eq:qd1}
a^2\frac{(z-b)^2(z-\beta)^2}{z^2(z-a)^2}{\rm d}z^2<0 
\end{equation}
\end{itemize}
The unique existence of the curve satisfying the three conditions is argued in the 2nd paragraph of Section \ref{sec-ypost}.

Using ${\cal B}$ we define the function $y$ by 
\begin{equation}\label{ypost}
y(z):=\pm a\frac{(z-b)(z-\beta)}{z(z-a)},\quad \begin{cases} +: z\in{\rm Ext}({\cal B});
\\-:  z\in{\rm Int}({\cal B}).
\end{cases}
\end{equation}
\item the {\bf critical case}, $t\sim t_c$: We consider a {\em scaling limit} to the criticality such that
\begin{equation}\label{eq:ttc}
t-t_c={\cal O}(N^{-2/3}),
\end{equation}
that is, we look at the scaling limit such that $\lim_{N\to\infty}N^{2/3}(t-t_c)$ converges.

We define $b$ and $\beta$ by \eqref{bbetapost} as in the post-critical case.
When $t>t_c$ these definitions give {\em complex} $b$ and $\beta$ such that $b=\overline\beta$; we choose $b$ and $\beta$ such that ${\rm Im}\,\beta>0>{\rm Im}\,b$.  Using these $b$ and $\beta$ the rest of the definitions, for ${\cal B}$ and the function $y$, are exactly same as in the post-critical case.  The uniqueness and the existence of ${\cal B}$ in this case is stated after \eqref{betabasymp}. 

\end{itemize}
\begin{figure}[ht]
\begin{center}
\includegraphics[width=0.9\textwidth]{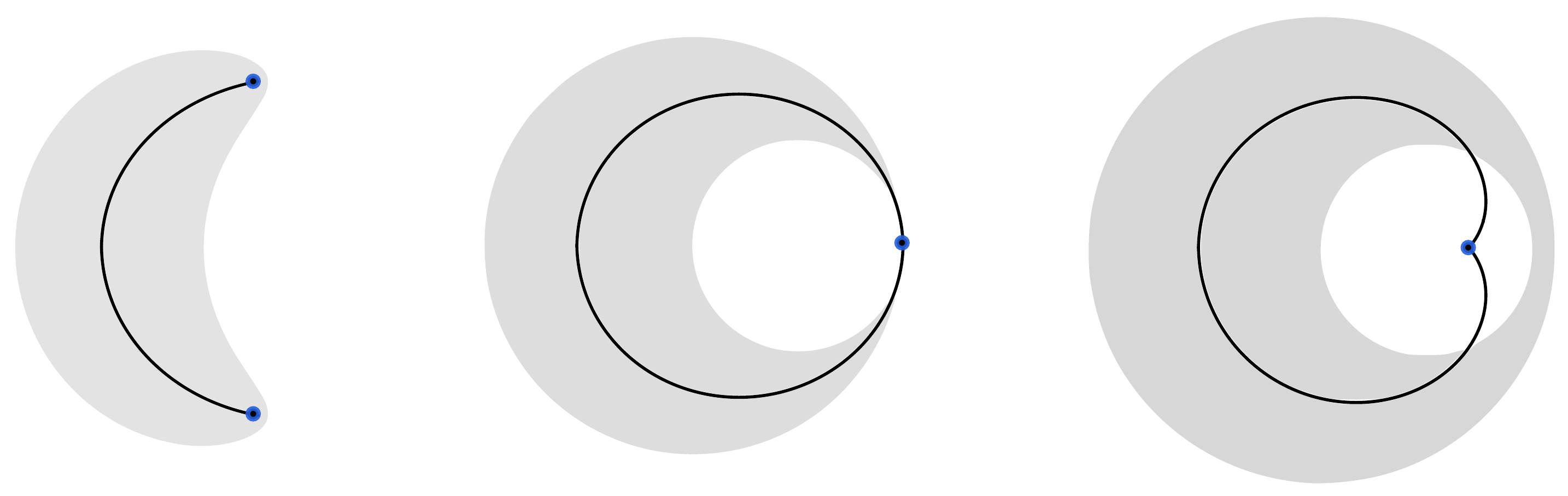}
\caption{\label{fig-1}$K$ (shaded regions), ${\cal B}$ (lines) and $\beta$ and/or $\c\beta$ (dots).  The pre-critical (left), critical (middle), post-critical (right) cases.  
}
\end{center}
\end{figure}

In the next two definitions we will introduce two functions.  The first function, $\phi$, will be the main object for the steepest descent analysis, namely we take the steepest descent with respect to $\Re\phi$.  The latter is also the effective potential of the limiting distribution of the zeros.  The second function, $g$-function, is the function whose real part is the logarithmic potential of the limiting distribution of the zeros.  Outside $K$, the latter coincides with the logarithmic potential of the uniform measure supported on $K$.

\begin{defn}\label{def-phi0} For all $t$, we define the {\em multi-valued} $\phi(z)$ by
\begin{equation}
\phi(z):=\int_\beta^z y(s)\,\d s
\end{equation}
where the integration contour lies in $\C\setminus{\cal B}$.
\end{defn}
Regarding the multi-valuedness, see the remark below Definition \ref{def-phi}.
\begin{defn}\label{def-g}
For all $t$, the $g$--function is defined by
\begin{equation}
\label{ydef}
2  g(z) := \frac{1}{t}V(z)-\frac{1}{t} \phi(z) + \ell .
\end{equation}
The constant $\ell\in\R$ is defined such that $\lim_{z\to\infty} [g(z)-\log z]= 0$.
\end{defn}

For the critical and the post-critical cases, the $g$--function turns out to be simple:
\begin{equation}
\label{gdough}
g(z) = \begin{cases}\displaystyle \ln z + \frac c t \ln \left(\frac {z}{z-a} \right), &  z\in {\rm Ext}({\cal B});
\\\displaystyle 
 \frac a t z + \ell, &z\in{\rm Int}({\cal B});
\end{cases}
\end{equation}
where the Robin's constant $\ell$ is chosen so that $\Re g(z)$ is continuous at $\beta \in {\cal B}$ and thus given by:
\begin{equation}\label{Robindough}
\ell=\frac{1}{t}\Re\big[(t+c)\log\beta-c\log(\beta-a)-\beta a\big].
\end{equation}
We should mention that the curve $\mathcal B$ is a part of the level set $\Re g(z)=0$ because of the conditions \eqref{eq:qd} and \eqref{eq:qd1}.  We also note that $\Re$ is needed when $\beta$ is complex for the critical case (i.e. in the scaling regime), see the next paragraph below \eqref{eq:ttc}.

\paragraph{Summary of results:} We consider the following version of the scaling limit, given by
\begin{equation}\label{eq:Nnrt}
N=(n-r)/t,\quad n \to \infty, \quad \text{for fixed constants, $t\in\R^+$ and $r\in {\mathbb Z}$}.
\end{equation}
Note that $N$ is {\em defined} by (1.23) in terms of $n, r, t$ and thus is not an integer.
We express our asymptotic results in terms of {\em geometric quantities related to $K$} (we believe that, in such form, the asymptotic expression holds for more general $Q$). 
Let us define ${\bf P}(K)$ to be the polynomial convex hull of $K$ (which in the present cases may be defined as the set whose complement is precisely the connected component of $K^c$ containing $\infty$).  Since $K$ is either simply connected ($t<t_c$), or multiply connected ($t\geq t_c$) as shown in Figure \ref{fig-1}, we have
\begin{equation}
{\bf P}(K)=\begin{cases} K , &\text{ pre-critical case}.
\\ \{ z:|z|<\sqrt{t+c}\}, &\text{ critical and post-critical cases}.
\end{cases}
\end{equation}
Below, we describe some quantities naturally associated with the geometry of $K$. 
\begin{itemize}
\item[--] Let $F$ be the conformal map $F:\widehat\C\setminus {\bf P}(K)\to \widehat\C\setminus\overline \D$ (where $\D$ is the unit disk centered at the origin and $\widehat \C=\C\cup\{\infty\}$) such that $F(z)=z/\rho + {\cal O}(1)$ for some $\rho>0$ as $z\to\infty$.  The constant $\rho$ is called the {\em logarithmic capacity of ${\bf P}(K)$}. 
\item[--] 
We extend the domain of $F$.  In the pre-critical case, we will show that $F$ extends analytically to $\C\setminus{\cal B}$ (See Lemma \ref{lem-F}). In the post-critical case, as well as the critical case, $F(z)=z/\sqrt{t+c}$ extends analytically to the whole ${\rm Ext}({\cal B})$, where $\mathcal B$ is a closed Jordan curve as defined below \eqref{bbetapost}.  
\item[--]  For $z\notin K$, $g(z)$ is exactly the {\em complex logarithmic potential of the measure uniformly supported on $K$} (see Lemma \ref{lem-g}). 
For the critical and the post-critical cases, $g(z)$ is given explicitly by \eqref{gdough}.  
\end{itemize}
To state our results, we define the regions $\Omega_\pm$ to be on each side of ${\cal B}$ as in Figure \ref{fig_31} and Figure \ref{default}. (These regions are only used to describe the region around ${\cal B}$ and, therefore, we do not need their exact extent away from ${\cal B}$.) 

We further define a few constants:
\begin{equation}\label{gamma1}
\gamma_1:=\sqrt{\frac{\beta(\beta-a)}{a(b-\beta)}},\qquad \gamma_c:=2 \frac{b_c^{1/3} c^{1/6}}{a^{1/3}}, \qquad b_c:=a+\sqrt c.
\end{equation}

\begin{theorem}[Pre-critical]\label{thm1}
For any fixed $t<t_c$, we have
\begin{equation}\label{Pnm}
P_{n,N}(z)=\left(1+{\cal O}\left(N^{-1}\right)\right)\sqrt{\rho F'(z)}\, (\rho F(z))^r \e^{ N t \,g(z)}, 
\end{equation}
uniformly over $z$ in compact subsets of $\widehat\C\setminus {\cal B}$.  For $z$ near ${\cal B}$ but away from $\beta$ and $\c\beta$, we have
\begin{equation}\label{onBpre}
P_{n,N}(z)={\rm e}^{tN g(z) } \sqrt{\rho F'(z)}
\left(  (\rho F(z))^r  \pm   {\rm e}^{  N \phi(z)  } \left(\frac{\alpha\,z}{F(z)}\right)^{r} \frac{\sqrt{\kappa\rho}}{\rho F(z)-\alpha\rho}+ \mathcal O(N^{-1})\right),
\end{equation}
uniformly over $z$ in compact subsets of $\overline{\Omega_\pm}\setminus\{\beta,\c\beta\}$ (respectively for $\pm$).
The real parameters, $\kappa$ and $\alpha$, are uniquely determined from the conformal map $F$ by
\begin{equation}
F(z)=\frac{z}{\rho}+\frac{\kappa}{\alpha\rho}+\frac{\kappa}{z}+{\cal O}\left(\frac{1}{z^2}\right),\quad z\to\infty.
\end{equation}
Lastly, in order to describe $P_{n,N}$ near $\c\beta$, define $z(\zeta):=\overline\beta+ (C_{\c\beta}N)^{-2/3}\zeta$ with
\vspace{-0.2cm}
\begin{equation}\label{Cbeta}
C_{\c\beta}:= \frac{a(\c\beta-b)\sqrt{\c\beta-\beta}}{2\c\beta(\c\beta-a)}.
\end{equation}
where $\beta$ and $b$ are defined at \eqref{betapre} and \eqref{eq:qd}.   We have
\begin{equation}\label{PAiry}
P_{n,N}\left(z(\zeta)\right)=\left(\frac{\sqrt{\kappa\rho}}{4\ii(z(\zeta)-\c\beta)}\right)^{1/4} (\c v_\beta\rho)^r \sqrt{\pi}\, \zeta^{1/4} \left({\rm Ai}(\zeta)+{\cal O}\left(N^{-1/3}\right)\right) \,\e^{\frac 23\zeta^{3/2}}\e^{Ntg(z(\zeta))},
\end{equation}
uniformly over $\zeta$ in compact subsets of $\C$.
The correct branch of the $4$th-root is determined by knowing that the factor comes from the leading behavior of $\sqrt{\rho F'(z)}$ when $z\to\c\beta$.
The case for $z\to\beta$ is obtained by the replacement: $\beta\leftrightarrow\c\beta$. 
\end{theorem}
In \eqref{PAiry} we write the asymptotic formula such that one can easily see that it matches the formula \eqref{Pnm} as $|\zeta(z)|\to\infty$.

\begin{theorem}[Post-critical]\label{thm2}  For $t>t_c$, we have
\begin{equation}\label{Pnmpost}
P_{n,N}(z)=  \begin{cases}\displaystyle
\left(1+{\cal O}\left(N^{-3/2}\right)\right)z^{n}\left(\frac {z}{z-a}\right)^{Nc}, &z\in{\rm Ext}({\cal B});
\vspace{0.1cm}\\\displaystyle
\left(\frac{\gamma_1\beta^r}{\sqrt{2\pi N}}+{\cal O}\left(N^{-3/2}\right)\right) \frac{\e^{Naz +  tN \ell}}{\beta-z},
&z\in{\rm Int}({\cal B});
\end{cases} 
\end{equation}
uniformly over $z$ in compact sets of $\widehat\C\setminus{\cal B}$.  The constant $\ell$ is defined in \eqref{Robindough}.  For $z$ near ${\cal B}$ but away from $\beta$, we have
\begin{equation}\label{onBpost}
P_{n,N}(z)=\begin{cases}\displaystyle
z^{n}\left(\frac {z}{z-a}\right)^{Nc}\left(1-\frac{\gamma_1\beta^r}{\sqrt{2\pi N}}\frac{{\rm e}^{N\phi(z)}}{z^r (z-\beta)}+{\cal O}\left(N^{-3/2}\right)\right)  ,\quad &z\in \overline{\Omega_-}\setminus\{\beta\};
\vspace{0.1cm}\\ \displaystyle
\e^{Naz +  tN \ell}\left(z^r{\rm e}^{N\phi(z)}-\frac{\gamma_1\beta^r}{\sqrt{2\pi N}}\frac{1}{\beta-z}+{\cal O}\left(N^{-3/2}\right)\right), & z\in \overline{\Omega_+}\setminus\{\beta\};
\end{cases}
\end{equation}
uniformly over $z$ in compact subsets of $\overline{\Omega_\pm}\setminus\{\beta\}$.

To describe $P_{n,N}$ near $\beta$, we define $z(\zeta):=\beta+\gamma_1 N^{-1/2}\zeta$ with $\gamma_1$ defined at \eqref{gamma1}.  We have
\begin{equation}
P_{n,N}(z(\zeta))= \e^{Naz +  tN \ell}\beta^r \left( {\cal F}(\zeta) +{\cal O}\left(N^{-1/2}\right)\right), 
\end{equation}
uniformly over $\zeta$ in compact subsets of $\C$, where ${\cal F}$ is the entire function given by
\begin{equation}\label{calF}
{\cal F}(\zeta):= \displaystyle\frac{1}{2i\pi} \int_{-i\infty}^{i\infty} \frac{{\rm e}^{s^2/2}}{(s-\zeta)}\, \d s +\begin{cases} 0, &\Re\zeta<0,\\ \e^{\zeta^2/2}, &\Re\zeta>0.\end{cases}
\end{equation}
(This function is plotted in Figure \ref{WeberFig}.) 
\end{theorem}
\begin{theorem}[Critical]\label{thm-crit} In the critical case, i.e. $t-t_c={\cal O}(N^{-2/3})$, we have
\begin{equation}\label{Pnmcrit}
P_{n,N}(z)=
\begin{cases}
\displaystyle \left(1 - \frac {\gamma_c u(s) }{2N^{1/3}(z-b_c)} +{\cal O}\left(N^{-2/3}\right)\right)z^{n}\left(\frac {z}{z-a}\right)^{Nc}, 
 &z\in {\rm Ext}({\cal B});
\vspace{0.1cm}\\\displaystyle
\left(\frac{\gamma_c \,q(s) (b_c)^r}{2N^{1/3}}\e^{\frac{2}{3}|s|\sqrt s}+{\cal O}\left(N^{-2/3}\right)\right) \frac{\e^{Naz +  tN \ell}}{\beta-z} , &z\in {\rm Int}({\cal B});
\end{cases}
\end{equation}
uniformly over $z$ in compact subsets of $\widehat\C\setminus{\cal B}$. 
Here 
\begin{equation}
 	s=c^{1/6}a^{-1/3}(b_c)^{-2/3}N^{2/3}(t-t_c)~\text{ and }~ u(s):=q'(s)^2-s\,q(s)^2-q(s)^4
 \end{equation} 
with $q(s)$ being the Hastings-McLeod solution of Painlev\'e II equation.

For $z$ near ${\cal B}$ but away from $\beta$, we have
\begin{equation}\label{onBcrit}
P_{n,N}(z) =\begin{cases}\displaystyle
z^{n}\left(\frac {z}{z-a}\right)^{Nc}\left(1 - \frac {\gamma_{c}  u(s)}{2 N^{1/3}(z-b_c)} + \frac { \gamma_{c}  q(s)\,(b_c)^{r}{\rm e}^{\frac{2}{3}|s|\sqrt s}}{2N^{1/3}z^r (z-b_c)} e^{N\phi(z)} +{\cal O}\left(N^{-2/3}\right)\right) ,
\vspace{0.1cm}\\\displaystyle
\e^{Naz +  tN \ell}\left(\left(1- \frac {\gamma_{c}  u(s)}{2 N^{1/3}(z-b_c)} \right) z^r\e^{N\phi(z)}+ \frac { \gamma_{c}  q(s)\,(b_c)^{r}{\rm e}^{\frac{2}{3}|s|\sqrt s }}{2N^{1/3} (z-b_c)} +{\cal O}\left(N^{-2/3}\right)\right),
\end{cases}\!\!\!\!
\end{equation}
uniformly over $z$ in compact subsets, respectively, of $\overline{\Omega_-}\setminus\{\beta\}$ and of $\overline{\Omega_+}\setminus\{\beta\}$.
\end{theorem}

\begin{remark}
We note that the leading asymptotic behavior of $P_{n}$ in \eqref{Pnm} is conveniently expressed in terms of the geometric quantities: $\rho$, $F$ and $g$.  The asymptotic formula in ${\rm Ext}({\cal B})$ for the post-critical and critical cases shown in \eqref{Pnmpost} and \eqref{Pnmcrit} can also be expressed in terms of the geometric quantities if one uses that in the post-critical case the conformal map of the exterior of $\mathbf P(K)$ satisfies the simple identity: $\rho F(z)=z$, and \eqref{gdough}.  Then the leading terms take precisely the same form as the leading term in \eqref{Pnm}, i.e. $z^n \left(z/(z-a)\right)^{Nc}=(\rho F(z))^r\sqrt{\rho F'(z)} \,\e^{N t g(z)}$. It has been conjectured \cite{PRL02} that the asymptotic behavior in $K^c$ should always take the latter form.\end{remark} 


\begin{remark} One can confirm that the strong asymptotics in the critical case match those in the post-critical case in the limit $s\to\infty$, using 
\begin{equation}
q(s)\,{\rm e}^{i\frac {16}3\zeta_\beta^3 }\sim \frac{1}{2\sqrt\pi s^{1/4}} \text{ as } s\to\infty, \quad\text{ and }~~ \gamma_1\sim \frac{\gamma_{c}N^{1/6} }{2\sqrt2 \,s^{1/4}},
\end{equation}
where the former is obtained from the asymptotic behavior of the Hastings-McLeod solution \eqref{HMsol}, and the latter is by writing $\gamma$ \eqref{gamma1} using \eqref{betabasymp}.
An argument for the appearance of the Painlev\'e II equation in the critical case was explained in \cite{bubble}.
\end{remark}

\begin{remark}
In the above theorems, though not immediately noticeable, the leading asymptotics \eqref{onBpre}\eqref{onBpost}\eqref{onBcrit} near ${\cal B}$ are analytic on ${\cal B}\setminus\{\beta\cup\c\beta\}$ and they are exactly the sum of the two contributions from both sides of ${\cal B}$ by analytic continuations of the asymptotic formula on each side.   We also note that the two leading terms for $z\in\Omega_\pm$ are simply related by $[tg(z)]_\pm=[tg(z)+\phi(z)]_\mp$.
\end{remark}

Equipped with the above asymptotic results concerning the orthogonal polynomials, it is straightforward but tedious
to prove the following proposition.
In this paper we do not give a complete proof of this proposition.
Rather, we provide (at the end of Section \ref{sec-post}) a sketch of the important steps involved in the proof.
(See \cite{Kuijlaars-McLaughlin-04} for similar argument regarding the location of zeros.)
\begin{prop}\label{prop-zero}
The support of the counting measure of the zeros of the polynomials $P_{n,N}$ outside of arbitrary small disk centered at the turning points $z =\beta, \overline \beta$ (pre-critical)  and $z=\beta$ (critical/post-critical), tends uniformly to the branchcut $\mathcal B$ for all three cases.
Moreover, in the post-critical case, the zeros are within $\mathcal O(N^{-1})$ away from the curve defined implicitly by \eqref{zeroas} that tends to $\mathcal B$ at the rate of $\mathcal O(\ln N/N)$ (see Figures \ref{fig10} and \ref{fig11}).  The normalized counting measure of the zeros of $P_{n,N}$ converges to the probability measure (supported on ${\cal B}$) whose logarithmic potential is given by $\Re g$.
\end{prop}

It is possible to deduce from the proof of Theorem \ref{thm-crit} an asymptotic description for $P_{n,N}$ which is valid near $\beta$ for the critical case. The formulae, however, are cumbersome in that they involve the solution of the Painlev\'e II Riemann-Hilbert problem itself. 
Rather than highlighting this with a separate theorem, we provide instead the following proposition, establishing a uniform bound on $P_{n,N}$ in the entire plane.
\begin{prop}\label{prop-O1} We have the following uniform bounds:
\begin{equation}\label{uniform}
P_{n,N}(z)=\e^{N t g(z)}\times 
\begin{cases}{\cal O}(N^{-1/2}),&\text{over compact subsets of ${\rm Int}({\cal B})$ for post-critical};
\\{\cal O}(N^{-1/3}),&\text{over compact subsets of ${\rm Int}({\cal B})$ for critical};
\\{\cal O}(N^{1/6}),&\text{on a neighborhood of $\{\beta,\c\beta\}$ for pre-critical};
\\ {\cal O}(1) ,&\text{otherwise}.
\end{cases} 
\end{equation}
\end{prop}

The following theorem, that we prove in Section \ref{sec-norm}, gives the asymptotics of $h_n$.
\begin{theorem}\label{propnorm} 
 For all $t\neq t_c$, we have
 \begin{equation}\label{hnuniform}
 h_n = \left(1+ \mathcal O\left(N^{-1}\right)\right) 2\pi \sqrt{\frac{\pi}{2N}}\,\rho^{2r+1} e^{N\ell_\text{2D}},
 \end{equation}
 where $\ell_\text{2D}$ is the real number given by $\ell_\text{2D}=2 t \Re g(z)-Q(z)$ for $z\in\partial {\bf P}(K)$.  More explicitly, we have (see Lemma \ref{lem-ell2d})
\begin{equation}\label{eq:ell2D}
	\ell_\text{2D}= \begin{cases}\displaystyle 
	2(c+t)\log\rho-2c\log\alpha +\left(\alpha ^2-1\right) \left(a^2-\frac{2 a \rho }{\alpha }\right)-\frac{\rho ^2}{\alpha ^2} 
	&\text{ for } t<t_c;
	\\(c+t)\log(c+t)-(c+t) &\text{ for } t\geq t_c.
	\end{cases}
\end{equation}
See \eqref{fv} and the paragraph that follows for the definitions of $f$, $\rho$, $\kappa$ and $\alpha$.  In fact, $\ell_{2D}$ is continuous at $t=t_c$ and \eqref{hnuniform} holds with the error bound replaced by ${\cal O}(N^{-1/3})$.
\end{theorem}

\begin{figure}[ht]\begin{center}
\includegraphics[width=0.9\textwidth]{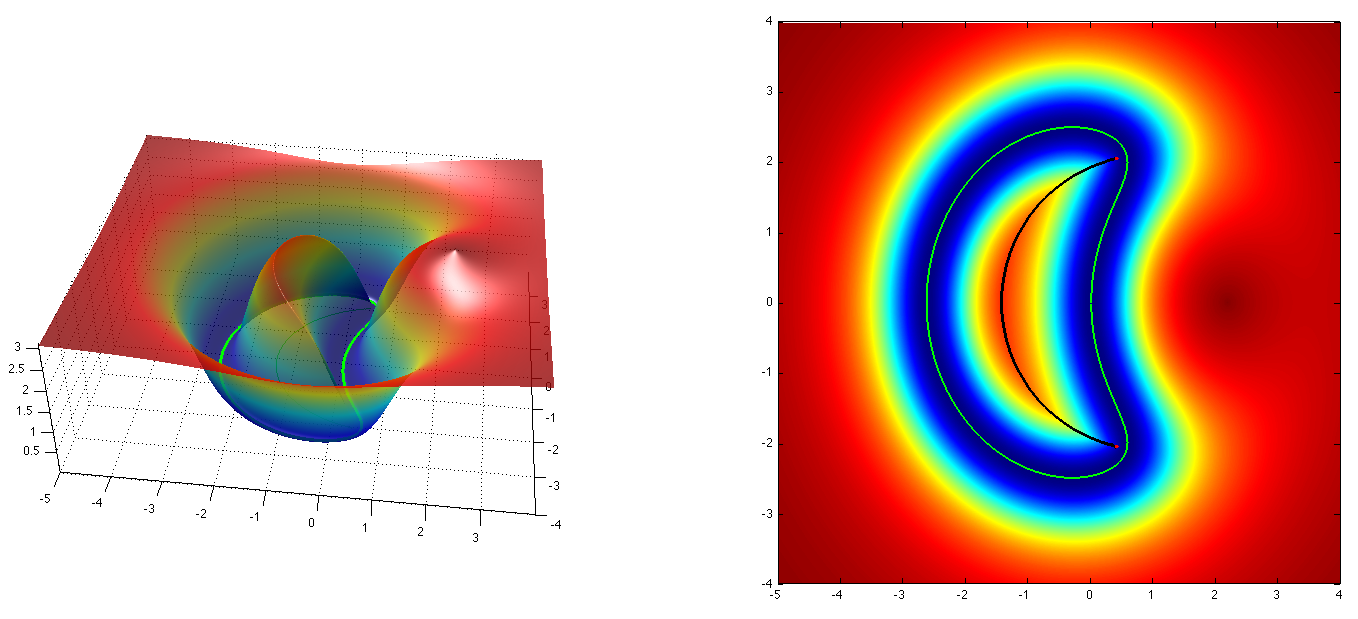}
\caption{The graph of $2\arctan({\cal U}(z))$ for the pre-critical case; the right side is a view from the top. The parameters are $a = 3.7619$, $c=6.9168$, $t=4.0557$ ($\rho = 2.1;  \alpha = 0.4;  \kappa = 1/2$). 
Indicated are the boundary of $K$ (thin green line) and the branch cut ${\cal B}$.
\label{fig-U}
}\end{center}
\end{figure}

\begin{figure}[ht]\begin{center}
\includegraphics[width=0.5\textwidth]{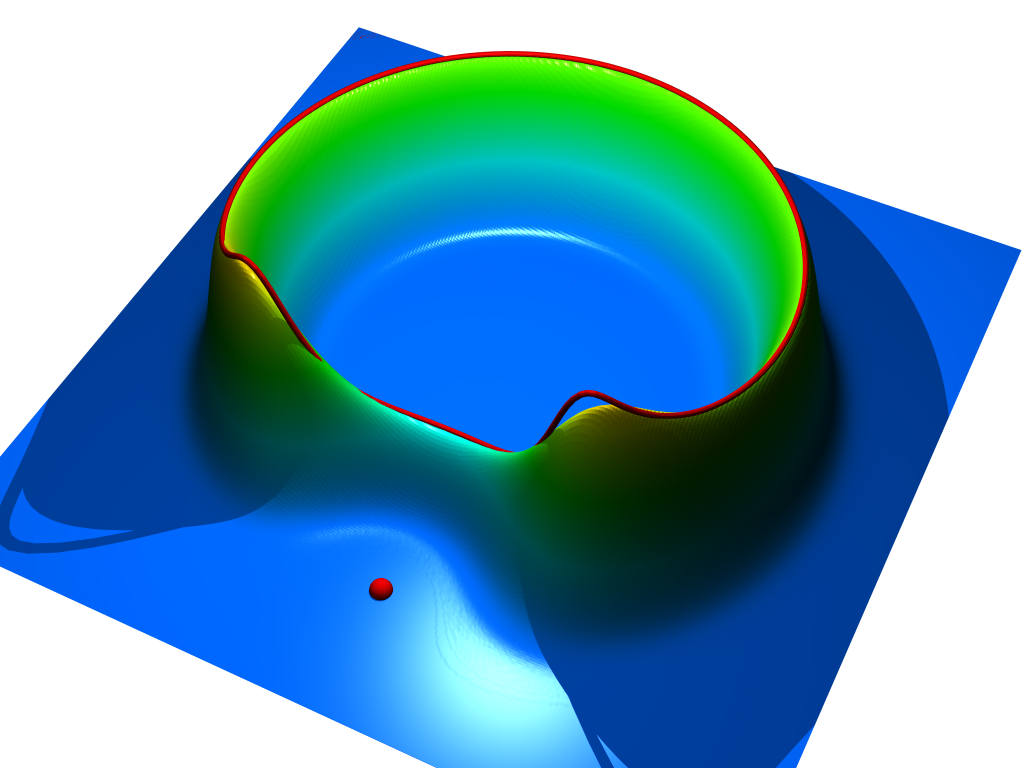}
\caption{Plot of $2\pi\sqrt{\pi/(2N)} \,\rho_n(z)$ for $c=1/6,\ a=1,\ N=30,\ n=25,\ t=0.8\c 3,\ t_c \approx 1.81649,\ N t_c = 54\sim 55$.  The red dot is the location of $z=a$; The red line shows the height $|F'(z)|$ for $z\in{\partial K}$. \label{fig-harm} }
\end{center}
\end{figure}

\paragraph{Harmonic measure from $P_{n,N}$:} 

The harmonic measure $\omega$ of a domain $\Omega\subset{\widehat\C}$ with a pole at $\infty$ (assume $\infty\in \Omega$)  is the probability measure supported on $\partial\Omega$ that satisfies the identity 
$
\widehat f(\infty)=\oint_{\partial\Omega} f(z)\,\d\omega(z)
$ 
for any continuous function $f$ on $\partial \Omega$ where $\widehat f:\overline\Omega\to \C$ is the unique bounded harmonic extension of $f$. 

If $\partial\Omega$ consists of piecewise smooth Jordan curves, then we have \cite{book-harmonic-measure}
\begin{equation}
\d \omega(z) 
=\frac{1}{2\pi}|F'(z)|\,\d \ell(z)
\end{equation}
where $F:\Omega\to \overline{\D}^c$ is a conformal map that preserves $\infty$, and $\d\ell$ is the length measure on $\partial\Omega$.

We define the following probability density function:
\begin{equation}\label{rhon}
\rho_n(z):=\frac{1}{h_{n}}|P_{n}(z)|^2 \e^{-N Q(z)}.
\end{equation}
This is but a single term in the average density \eqref{density} of the eigenvalues, representing ``the density of an additional eigenvalue on top of $n$ eigenvalues".  In the scaling limit, an additional eigenvalue gives an infinitesimal increase to $t$ by an amount of 1/N (recall that $t \sim n/N$).  
By considering the perturbation of $t$ in the functional ${\cal L}(\mu)$ (or, equivalently, by taking the $t$-derivative of \eqref{ELeq}), one can deduce that the incremental growth of the measure, $\d(t \mu^*)/\d t$, can be expressed in terms of the harmonic measure of $\widehat\C\setminus {\bf P}(K)$ with a pole at $\infty$.    This suggests that the measure $\rho_n\d A$ converges to the harmonic measure.  This fact was conjectured in the physics literature \cite{PRL02} and proved in Theorem 7.7.3 of \cite{AHM}.

Below we obtain global uniform asymptotics for $\rho_{n}$, which not only shows that $\rho_n\d A$ converges to the harmonic measure, but also shows {\em in a very precise way} how $\rho_n$ converges pointwise.    We mention that the leading term in \eqref{rhon} was reported in \cite{PRL02}.
\begin{cor}\label{cor-harm}
We have $\rho_n(z)= {\cal O}(1) \,\e^{-N {\cal U}(z)}$ uniformly in $z\in\widehat\C$ and, furthermore,
\begin{equation}\label{rhonresult}
\rho_n(z)= 
\left(1+{\cal O}\left(N^{-1/3}\right)\right)\frac{|F(z)|^{2r}}{2\pi}\sqrt{\frac{2N}{\pi}}|F'(z)|\,\e^{-N {\cal U}(z)} 
\end{equation}
uniformly over $z$ in compact subsets of $\widehat\C\setminus{\cal B}$  for the pre-critical case and of ${\rm Ext}({\cal B})$ for the critical and post-critical cases.   Over compact subsets in ${\rm Int}({\cal B})$, we have $\rho_n(z)={\cal O}(N^{-1})\,\e^{-N {\cal U}(z)}$ for the post-critical case, and $\rho_n(z)={\cal O}(N^{-2/3})\,\e^{-N {\cal U}(z)}$ for the critical case.   For the pre-critical case, over a neighborhood of ${\cal B}$, we have $\rho_n(z)={\cal O}(N^{1/3})\,\e^{-N {\cal U}(z)}$. \end{cor}
\begin{proof}
The proof follows from Proposition \ref{prop-O1}, the asymptotic formula \eqref{Pnm} and Theorem \ref{propnorm}.  Let us only check it for $z$ in compact subset of $\C\setminus{\cal B}$ and of ${\rm Ext}({\cal B})$ (the other regions can be similarly considered).  Using the remark below Theorem \ref{thm-crit} the leading term \eqref{Pnm} applies to all values of $t$.  Substituting the asymptotics of $P_n$ \eqref{Pnm} and $h_n$ \eqref{hnuniform} into $\rho_n$ \eqref{rhon}
we obtain
\begin{equation}
\rho_n(z)=\left(1+{\cal O}\left(N^{-1/3}\right)\right)\frac{|\rho F'(z)||\rho F(z)|^{2r} \e^{2N t \Re g(z)} }{ 2\pi \sqrt{\pi/(2N)}\,\rho^{2r+1} e^{N\ell_\text{2D}}} \e^{-N Q(z)}.
\end{equation}
The error bound comes from the worst case: \eqref{Pnmcrit} in the critical case.
The result follows from the identity ${\cal U}(z)=Q(z)-2 t \Re g(z)+ \ell_\text{2D}$ which can be seen from \eqref{Phi2D} and Lemma \ref{lem-g}.
\end{proof}

Corollary \ref{cor-harm} says that, for the pre- and the post-critical cases, the probability density function $\rho_n$ is exponentially suppressed away from $\partial{\bf P}(K)$ where ${\cal U}$ is minimized (Lemma \ref{globalUpre}, Lemma \ref{globalUpost}).   For the critical case, ${\cal U}$ is minimized on $\partial K$ by Lemma \ref{globalUcrit} so $\rho_n$ may be peaked on $\partial K\cap {\rm Int}({\cal B})$ also.    However, $\rho_n
 =\mathcal O(N^{-2/3} ){\rm e}^{-N \mathcal U(z)}$ in ${\rm Int}({\cal B})$, and this means only ${\cal O}(N^{-2/3})$ amount of mass from the probability measure $\rho_n\,\d A$ can be within ${\rm Int}({\cal B})$.   For the pre-critical case, no mass can be accumulated at $\beta$ or $\c\beta$ in the scaling limit because ${\cal O}(N^{1/3})\,\e^{-N {\cal U}(z)}$ is exponentially suppressed there.
To see how we obtain the approximation of the harmonic measure, let $p$ be a point on $\partial{\bf P}(K)$ and ${\bf n}$ be a unimodular complex number whose direction in the complex plane is perpendicular to $\partial{\bf P}(K)$.   Then 
we see from \eqref{Uboundary}, the fact that ${\cal U}$ is the smooth extension of ${\cal U}_\text{2D}$ to the interior, together with \eqref{rhonresult} that 
\be
\rho_n(p+ x\,{\bf n})\sim  \frac {|F'(p + x{\bf n})|}{2\pi} |F(p + x{\bf n})|^{2r} \sqrt{\frac {2N}\pi}   \e^{-2Nx^2}
\ee 
where we see the normal distribution approximation of the Dirac delta function of the boundary ($x=0$) in the trailing term.
  This means that the density $\rho_n(p+ x{\bf n}) $ is localized around $\partial{\bf P}(K)$ with the typical width $\sim\sqrt{1/N}$.  
Since $\frac{1}{2\pi}|F'(p)|\,\d \ell(p)$  (with $\d\ell$ being the length measure on $\partial{\bf P}(K)$) gives the harmonic measure and $|F(p)|=1$, we conclude that $\rho_n\d A$ converges to the harmonic measure and moreover $\rho_n$ converges pointwise to the density of the harmonic measure (with respect to the area measure).

\paragraph{Strategy of the proofs:}
In Section \ref{sec-2}, we begin by defining the region $K$ in terms of explicit formul\ae.     We introduce the Schwarz function $S(z)$ in terms of which we define $y(z)$ that facilitates the construction of the $g$-functions. 

Section \ref{sec-3} contains the crucial step in our analysis, showing that $P_{n,N}$ is an orthogonal polynomial with respect to an analytic weight supported on a contour in the complex plane.   This leads us to a Riemann-Hilbert formulation of $P_{n,N}$ and the asymptotic analysis is performed in the subsequent sections via the steepest descent analysis  \cite{DKMVZ} on the corresponding Riemann-Hilbert problem. 

The proofs of Theorem \ref{thm1}, Theorem \ref{thm2}, and Theorem \ref{thm-crit} are respectively in Section \ref{sec-pre}, Section \ref{sec-post}, and Section \ref{sec-crit}.

\paragraph{Acknowledgements} M. B.  acknowledges a support by Natural Sciences and Engineering Research Council of Canada.  
S. Y. Lee acknowledges the support by Sherman Fairchild Foundation.
K. D. T-R McLaughlin is supported in part by the National Science Foundation under grants
DMS-0200749, and DMS-0800979.  The authors wish to thank the faculty and staff of the Centre de Recherches Math\'ematiques for their support.  The authors also thank AIM workshop ``Vector equilibrium problems and their applications to random matrix models" and useful discussion in that research-conducing environment.

\section{2D equilibrium measure problem}\label{sec-2}

We introduce a compact region $K$ as the support of the equilibrium measure $\mu^*$ that minimizes ${\cal L}(\mu)$ defined in \eqref{functional}.  
From \eqref{ELeq}, it is well known that under certain regularity assumptions on the potential $Q$, the measure $\mu^*$ is known explicitly: $\d\mu^*=\frac{1}{\pi t}\c\partial\partial Q \cdot {\bf 1}_K$. (We denote $\partial=\frac{1}{2}(\partial_x-\ii \partial_y)$ and $\overline\partial=\frac{1}{2}(\partial_x+\ii \partial_y)$; except when $\partial$ is followed by a set; then it means ``boundary".)
The central problem then becomes the determination of the support set $K$.

For our specific potential $Q$ given in \eqref{ourQ} (which is of the type $Q(z)=|z|^2-(\text{harmonic})$)
$\frac{1}{\pi t}{\bf 1}_K \d A$ is the {\em unique} equilibrium measure, if the following holds.
\begin{equation}\label{equil}
\begin{split}
&\text{{\bf (\!A\!)} There exists $\ell_\text{2D}\in\R$ such that ${\cal U}_\text{2D} =0$ on $K$;}\qquad
\\
&\text{{\bf (\!B\!)} For the $\ell_\text{2D}$ chosen above, ${\cal U}_\text{2D} > 0$ on $K^c$};
\end{split}
\end{equation}
where
\begin{equation}\label{Phi2D}
{\cal U}_\text{2D}(z):=Q(z)-\frac{2}{\pi} \int_K\log|z-w|\d A(w)+\ell_\text{2D}.
\end{equation}

The following definition is specific\footnote{The standard notion of the Schwarz functions can be found in \cite{davis_book}.} to $Q$ given by \eqref{ourQ}.
\begin{defn}\label{def-S}
Given a compact set $K\subset\C$ of area $\pi t$, if there exists a function $S:K^c\to\widehat\C$ such that the following three conditions are satisfied 
\begin{itemize}
\item[i)] $S(z)$ has continuous boundary value on $\partial K$ and satisfies $S(z)=\overline z$ on $\partial K$,
\item[ii)] $S(z)-c/(z-a)$ is analytic in $K^c$,
\item[iii)] $S(z)= (c+t)/z+{\cal O}(z^{-2})$ for $z\to\infty$,
\end{itemize}  
then we say that $S$ is the {\bf Schwarz function of $K$} or {\bf $K$ has the Schwarz function} $S$.
\end{defn}
If $K$ has the Schwarz function $S$, the following identities make a useful connection between the Schwarz function and ${\cal U}_\text{2D}$.
\begin{equation}\label{SA1}
\frac{1}{\pi}\int_K\frac{\d A(w)}{z-w}=\begin{cases}\displaystyle
\c z+\frac{1}{2 i\pi} \oint_{\partial K} \frac{\overline w\,\d w}{z-w}=\overline z+\frac{1}{2\pi\ii}\oint_{\partial K} \frac{S(w)\,\d w}{z-w}=\overline z-\frac{c}{z-a},\quad &z\in K,\vspace{0.2cm}
\\ \displaystyle
\frac{1}{2 i\pi} \oint_{\partial K} \frac{\overline w\,\d w}{z-w}=\frac{1}{2\pi\ii}\oint_{\partial K} \frac{S(w)\,\d w}{z-w}=S(z)-\frac{c}{z-a},\quad &z\notin K.
\end{cases}
\end{equation}
Here the contour $\partial K$ is positively oriented with respect to $K$. 
This identity gives
\begin{equation}\label{PhiS}
\partial {\cal U}_\text{2D}(z)= \begin{cases}\displaystyle
0,\quad &z\in K,
\\ \displaystyle
\overline z-S(z),\quad &z\notin K.
\end{cases}
\end{equation}
From the equation \eqref{PhiS} we obtain, for each connected component of $K^c$,
\begin{equation}\label{PhifromS}
{\cal U}_\text{2D}(z)=|z|^2-|z_0|^2 -2\Re\left[\int_{z_0}^z S(w)\,\d w\right],  \quad z\in\text{(a component of $K^c$)},
\end{equation}
where $z_0\in\partial K$ is any point on the boundary of the same component.

\begin{lemma}  Suppose $K$ is a compact set such that each connected component of $K^c$ is bounded by a finite union of closed real analytic curves (without any singularity) and $K$ has a Schwarz function as in Definition \ref{def-S}.   Then $K$ is the strict local minima of ${\cal U}_\text{2D}$ in the sense that, for $\epsilon>0$ small enough  (i.e. near the boundary $\partial K$),
\begin{equation}\label{Uboundary}
{\cal U}_\text{2D}(z_0+ \epsilon (n_x+in_y)) = 2\epsilon^2 +{\cal O}(\epsilon^3),\quad \forall z_0\in\partial K, 
\end{equation}
where $(n_x,n_y)$ is the unit vector that is outward normal to $\partial K$ at $z_0$.
\end{lemma}
\begin{proof} We use the fact that the Schwarz function $S(z)$ can be analytically continued in a neighborhood of a real analytic curve $\partial K$.   This then gives a smooth extension, say ${\cal U}$ (later at \eqref{calU} we define this quantity as the maximal extension), of ${\cal U}_\text{2D}$ in a neighborhood of $\partial K$ by using \eqref{PhifromS}. We have ${\cal U}={\cal U}_\text{2D}$ on $\overline{K^c}$. 

On $\partial K$, $\partial{\cal U}(z)=\c\partial{\cal U}(z)=0$.    Because $\Delta{\cal U}=4$, the Hessian has trace $4$.   Because $|S'(z)|=1$ for $z\in\partial K$ (since $S(z)=\overline z$ there), the determinant of the Hessian of ${\cal U}$, which is given by
\begin{equation}\label{hess}
4\left((\partial\c\partial \,{\cal U}(z))^2-(\partial^2 {\cal U}(z))(\c\partial^2 {\cal U}(z))\right)=4\left(1-|S'(z)|^2\right),
\end{equation} vanishes.
So the Hessian of ${\cal U}$ has eigenvalues 0 and $4$.  Since the flat direction is tangent to $\partial K$ the direction normal to $\partial K$ gives the steepest ascent of ${\cal U}$ as described in the lemma. 
\end{proof}
\begin{lemma}\label{lem-S2}  Suppose $K$ is a compact set of area $\pi t$ such that each connected component of $K^c$ is bounded by a finite union of closed real analytic curves and $K$ has a Schwarz function as in Definition \ref{def-S}.  Then using this set $K$ in \eqref{Phi2D}, {\bf (\!A\!)} in \eqref{equil} is satisfied.
If, furthermore, ${\cal U}_\text{2D}$ has no more than one critical point in $K^c$ then {\bf (\!B\!)} is also satisfied.
\end{lemma}
\begin{proof}
One shows {\bf (\!A\!)} of \eqref{equil}  by direct calculation using Stokes' theorem and $\partial{\cal U}_\text{2D}=\c\partial{\cal U}_\text{2D}=0$ in $K$ by \eqref{PhiS}.

Suppose there is a local minimum of ${\cal U}_\text{2D}$ at $K^c$.   Then the Mountain pass theorem  says that there must be a critical point on a path connecting $K$ and the local minimum.   To briefly explain, one crosses a critical point $z_\text{crit}$ (i.e. a mountain pass) by following a path for which the maximal elevation is minimal, i.e. ${\cal U}_\text{2D}(z_\text{crit})=\inf_{(\text{paths})}[\max_{z\in(\text{path})} {\cal U}_\text{2D}(z)]$.)  This requires at least two critical points in $K^c$.  Therefore, if there is no more than one critical point in $K^c$ then there is no local minimum in $K^c$, and the lemma follows.
\end{proof}

\paragraph{Plan of the proofs:}
In the rest of this section, we define the region $K$ explicitly.   Then we show that i) there exists the Schwarz function of $K$ and ii) ${\cal U}_\text{2D}$ has no more than one critical point in $K^c$.  By Lemma \ref{lem-S2} this is enough to conclude that $K$ is the global minima of ${\cal U}_\text{2D}$. 

\begin{prop}\label{thm-postK}
For $t\geq t_c$, the global minimum of ${\cal U}_\text{2D}$ is exactly at
the doubly connected region $K=\overline{D(0,\sqrt{t+c})}\setminus D(a,\sqrt c)$ where $D(a,r)$ stands for the disk of radius $r$ centered at $a$. 
\end{prop}
\begin{proof}
Our proof relies on the Schwarz function of $K$:
\begin{equation}\label{Spost}
S(z)=\begin{cases} (t+c)/z ,\quad |z|>\sqrt{t+c},
\\ a+c/(z-a),\quad |z-a|<\sqrt c.
\end{cases}
\end{equation}
To show that there is no critical point of ${\cal U}_\text{2D}$ at $K^c$ note that any point that satisfies $S(z)=\overline z$ is on $\partial K$.   The proof is complete by Lemma \ref{lem-S2}.
\end{proof}

For the pre-critical case, the region $K$ is defined by the following lemma.
\begin{lemma}\label{lem-preA} 
For $0<t<t_c$,
there exists a unique simply connected compact region $K\subset\C$ and a conformal mapping $f:\C\setminus\c{\mathbb D}\to\C\setminus K$ given by
\begin{equation}
f(v)= \rho\, v -\frac{\kappa}{v-\alpha}-\frac{\kappa}{\alpha},\quad \rho\in\R_+,\quad \kappa\in \R_+,\quad \alpha\in(0,1),
\end{equation}
with its inverse conformal map denoted by $F:\C\setminus K\to \C\setminus \c{\mathbb D}$,  
such that the function $S(z)$ given by $S(z)=f(1/F(z))$, is the Schwarz function of $K$, see Definition \ref{def-S}.  The three parameters satisfy
\begin{equation}\label{conformal-eq}
1-\alpha^2\geq \frac{\kappa}{\rho}.
\end{equation}
\end{lemma}
The proof of this lemma is given in Appendix \ref{sec-proof-lem-preA}.   
We note that the unique $K$ in the above lemma is exactly the one defined at \eqref{fv} through the equations \eqref{000} and \eqref{rhokappa}; the latter two equations are derived in the proof.

Note that the lemma asserts the existence of functions $\rho(t), \alpha(t), \kappa(t)$; it is a direct consequence of the proof that these are analytic in $t$ for $0<t<t_c$.

In what follows we will need the deck transformation, defined via
\begin{equation}\label{deck}
{\rm deck}(v):=\alpha-\frac{\kappa}{\rho}\frac{1}{v-\alpha}.
\end{equation}
This function is characterized as the unique nontrivial function satisfying $f(v) = f({\rm deck}(v))$. 

\begin{prop}\label{thm-pre}
For the pre-critical case, $K$ defined by the conformal mapping in Lemma \ref{lem-preA} satisfies {\bf (\!A\!)} and {\bf (\!B\!)} of \eqref{equil}, i.e. the global minima of ${\cal U}_\text{2D}$ are exactly at $K$.
\end{prop}

The existence of the Schwarz function follows from Lemma \ref{lem-preA}.  We will complete the proof of Proposition \ref{thm-pre} in Section \ref{sec-proof-thm-pre}, by showing that there is only one critical point of ${\cal U}_\text{2D}$.

\subsection{Definitions of \texorpdfstring{$y(z)$}{yz} and \texorpdfstring{$F(z)$}{Fz}; Proof of Proposition \ref{thm-pre} (pre-critical)}\label{sec-proof-thm-pre}

There are two critical points of $f:\C\to \C$ at
$\alpha\pm  i\sqrt{\kappa/\rho}$.
By \eqref{conformal-eq}, both of them are inside $\D$.  Let us denote by $v_{\beta}$ the critical point in the upper half plane, and let us denote by $\beta$ the corresponding branch point (critical value) in the upper half plane:
\begin{equation}
\beta:= f\left(\alpha +i\sqrt{\frac{\kappa}{\rho}}\right)=\alpha\rho-\frac{\kappa}{\alpha}+2i\sqrt{\kappa\rho} \text{ ~~ and ~~} v_\beta:=\alpha +i\sqrt{\frac{\kappa}{\rho}}.
\end{equation}
The complex conjugates, $\overline\beta$ and $\overline{v_\beta}$, give the other critical value and the branch point.

\begin{lemma}\label{lem-beta} The branch points $\beta$ and ${\c\beta}$ are in the interior of $K$.
\end{lemma}
\begin{proof}  Since $v_\beta\in \D$ is a critical point of the 2-to-1 branched covering map, $v_\beta$ is the only point that maps to $\beta$ under $f:\C\to\C$.  If $\beta$ is not in the interior of $K$, then by the surjectivity of $f:\D^c\to \c{K^c}$ one of the preimages of $\beta$ must be outside the unit disk, which is a contradiction.  
\end{proof}
Because $(v-\alpha)(f(v)-\beta)$ is a quadratic polynomial in $v$ with double roots at $v_\beta$ we have
$(v-\alpha)(f(v)-\beta)=\rho(v-v_\beta)^2$ and, similarly, $(v-\alpha)(f(v)-\c\beta)=\rho(v-\c{v_\beta})^2$.
We also have $(v-\alpha)(v-{\rm deck}(v))=(v-v_\beta)(v-\c{v_\beta})$ which, combined with the above, gives
\begin{equation}\label{diff}
(v-{\rm deck}(v))^2=\rho^{-2}(f(v)-\beta)(f(v)-\c\beta).
\end{equation}
We also obtain by a direct calculation that
\begin{equation}\label{sum}
v+{\rm deck}(v)=\frac{1}{\rho}f(v)+\frac{1}{\alpha}\left(\alpha^2+\frac{\kappa}{\rho}\right)=\frac{1}{\rho}(f(v)+|\beta|),
\end{equation}
where we use 
\begin{equation}
|\beta|=\alpha\rho+\frac{\kappa}{\alpha}. 
\end{equation}
that is obtained by taking $v=0$ in \eqref{diff}. 
From \eqref{diff} and \eqref{sum}, for any $v\notin\{v_\beta,\c{v_\beta}\}$, there is exactly one choice of sign such that
\begin{equation}\label{vbyfv}
\frac{1}{2\rho}\left[f(v)+ |\beta|\pm \sqrt{(f(v)-\beta)(f(v)-\c\beta)}\right]=v
\end{equation}
is satisfied; The other sign gives ${\rm deck}(v)$.
Now let us select our branch cut to be {\em inside} $K$ and to intersect $\R$ only at the negative real axis.   This is possible because $K\cap\R_-\neq\emptyset$ which is simply confirmed by $f(-1)=-\rho-\kappa/(\alpha(1+\alpha))<0$.  Having determined the branch cut, we select the branch of the square root by $[(z-\beta)(z-\c\beta)]^{1/2}\approx z+{\cal O}(1)$ as $z\to\infty$.    Then \eqref{vbyfv} gives the following lemma.
\begin{lemma}\label{lem-F}
Take any simple smooth curve ${\cal B}$ in $K$ that starts at $\beta$ and ends at $\c\beta$ while intersecting $\R$ only once at the negative real axis.
The map $F$ (the inverse of $f$; see Lemma \ref{lem-preA}) has the analytic continuation over $\C$ with the branch cut at ${\cal B}$ via the formula
\begin{equation}\label{Fz}
F(z)= \frac{1}{2\rho}\left[z+ |\beta|+ \sqrt{(z-\beta)(z-\c\beta)}\right],\quad {\rm deck}\circ F(z)= \frac{1}{2\rho}\left[z+ |\beta|- \sqrt{(z-\beta)(z-\c\beta)}\right].
\end{equation}
\end{lemma}
\begin{proof} It is enough to determine the sign in front of $\sqrt{...}$ by considering $F(z)\approx z/\rho$ as $z\to\infty$.  
\end{proof}
Using this $F$, the analytic continuation of the Schwarz function is given by $S(z)=f(1/F(z))$, of which the explicit form is announced in the next lemma.

\begin{lemma}[Definition of $y(z)$]\label{lem-sy}
$S(z)$ has analytic continuation to the entire complex $z$ plane with the branch cut at ${\cal B}$ (that is defined in Lemma \ref{lem-F}) via the formula
\begin{equation}\label{y}
S(z)=\frac{a}{2}+\frac{c}{2}\frac{1}{z-a}+\frac{c+t}{2z}-\frac{1}{2} y(z),\quad y(z):=a\frac{(z-b)\sqrt{(z-\beta)(z-\c\beta)}}{(z-a)z},
\end{equation}
where $y(z)$ satisfies three conditions respectively near $a,0$ and $\infty$:
\begin{equation}\label{ysing}
y(z)\sim \frac{-c}{z-a},\quad y(z)\sim \frac{c+t}{z},\quad y(z)\sim a-\frac{t}{z}+{\cal O}(1/z^2).
\end{equation}
\end{lemma}
We also note that the middle condition in \eqref{ysing} and the last condition in \eqref{valS} gives
\begin{equation}\label{bpre}
b=|\beta|/(t+c)=\alpha/\rho.
\end{equation}
It is useful to define a function $S_\text{back}$ (the analytic continuation of $S$ through the cut ${\cal B}$) by
\begin{equation}\label{Sback}
S_\text{back}(z):=f(1/{\rm deck}(F(z)))=S(z)+y(z).
\end{equation}

The following Lemma, combined with Lemma \ref{lem-S2}, proves Proposition \ref{thm-pre}.
\begin{lemma}\label{lem-double} The only point that satisfies $S(z)=\c z$ and $F(z)\notin \partial\D$ is at $z= b\in (a,\infty)$.  
\end{lemma}
\begin{proof}   
If $z$ satisfies $z=\c{S(z)}$ then $v=F(z)$ satisfies $f(v)=\c{f(1/v)}=f(1/\c v)$. This implies ${\rm deck}(v)=1/\c v$ if $v\notin \partial\D$, and we have 
\begin{equation}
S(z)=f(1/v)=\c{f(v)}=f(\overline v)=f(1/(1/\c v))=f(1/{\rm deck}(v))=S_\text{back}(z).
\end{equation}
Therefore, $z=b$ because $v_\beta$ is fixed under the deck transformation.  

To show that $b\in(a,\infty)$, one notes that $K\cap\R$ is connected (otherwise $K$ is not simply-connected) and lies to the left of $a$ (which can be checked from the conformal map).   
Because ${\cal U}_\text{2D}$ grows to $\infty$ at $a\in\R$ and $+\infty$, ${\cal U}_\text{2D}:\R\to \R$ must have at least one local minimum at $(a,\infty)$, which must be a critical point of ${\cal U}_\text{2D}$, where $S(z)=\c z$ holds.
\end{proof}
\noindent The last lemma tells that ${\cal U}_{2D}$ has at most one critical point in $K^c$, which proves Proposition \ref{thm-pre}.

\subsection{Branch cut of \texorpdfstring{$y(z)$}{yz1}: pre-critical case}

Having defined $S(z)$, $y(z)$, $F(z)$, as the functions on the {\em whole complex plane} with the branch cut on ${\cal B}$, we can define the following objects 
\begin{align}\label{Uop}
&{\cal U}_\text{OP}(z):=\Re\left[\int_\beta^z  y(w)\,\d w\right],\\
&{\cal U}(z):= |z|^2 -|z_0|^2- 2 \Re\left[\int_{z_0}^z  S(w)\,\d w\right], \quad \text{ for a point } z_0 \text{ s.t. }  F(z_0)\in\partial\D,\label{calU}
\end{align}
where the integration contour is in $\C\setminus{\cal B}$.  From \eqref{PhifromS} we note that ${\cal U}$ is the smooth extension of ${\cal U}_\text{2D}$ from $\C\setminus K$.

In the previous subsection, we have not specified the exact location of the branch cut ${\cal B}$ {\em except} that it is inside $K$ and intersects the negative real axis.  In this section we determine the precise location of the branch cut in such a way that will be useful to the subsequent Riemann-Hilbert analysis of the orthogonal polynomial.
\begin{defn}\label{def:level}
The {\bf level line(s) from a point}, say $p$, is defined to be the {\em connected} component of the set $\{z:\Re[\int_{p}^{z}y(w)\d w]=0\}$ which contains $p$, and, since the integral defines a continuous function up to sign,  the level line is given independently of the branch cut.  
\end{defn}
We first prove that the level line from $\beta$ does not intersect the level line from $b$.  

\begin{lemma}\label{lem-UOP2d}
We have
\begin{equation}
2{\cal U}_\text{OP}(z)={\cal U}(z)+{\cal U}(S_\text{back}(z))-|S_\text{back}(z)-\overline z|^2.
\end{equation}
\end{lemma}
The proof uses the anti-holomorphic involution structure, and can be found in Appendix \ref{app-lem-UOP2d}.
We obtain the following corollary that tells that the level lines from $\beta$ and $b$ do not intersect.
\begin{cor}\label{lem-betab}
We have
$\Re[\int_\beta^b y(z)\,\d z]={\cal U}_\text{2D}(b)>0$. The contour of integration is in $\C\setminus{\cal B}$.
\end{cor}
\begin{proof}  By Lemma \ref{lem-double}, we immediately obtain the first equality.  The last inequality is from Proposition \ref{thm-pre}.  Note that, by Lemma \ref{lem-double}, $b$ is not in $K$ so ${\cal U}(b)={\cal U}_\text{2D}(b)$.  
\end{proof}

\begin{figure}[htbp]
\begin{center}
\includegraphics[width=0.9\textwidth]{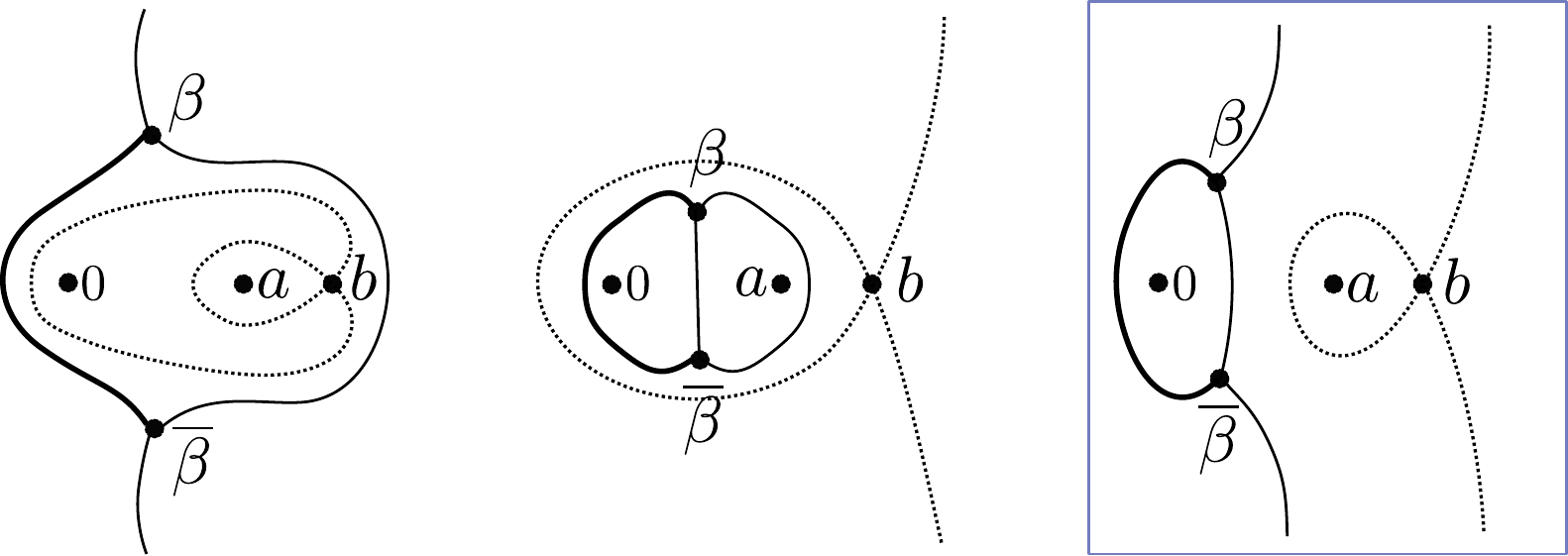}
\caption{Possible topological configurations of level lines from $\beta$ and $b$. The last (boxed) configuration satisfies the inequality in Lemma \ref{lem-betab}.  The thick lines are the (potential) branch cuts.}\label{fig_three} 
\end{center}
\end{figure}

In Figure \ref{fig_three} we show all the three possibilities that the level lines from $\beta$ and $b$ can place themselves in topologically distinguished ways.  We explain this below.  

The rules for the level lines from $b$ are as follows.
\begin{equation}\label{ruleb}
\begin{cases}
\text{The lines are symmetric with respect to the real axis;}
\\\text{A closed loop of level line must include pole(s) at $0$ and/or $a$;} 
\\\text{Four level lines emanate from $b$;}
\\\text{Out of four, at most two level lines can go to infinity;}
\end{cases}\qquad\qquad\qquad
\end{equation}
Following these rules, we first draw all the possible level lines from $b$ (the dotted lines).  
If none of the level from $b$ escapes to $\infty$ then the level lines from $b$ follows the dotted lines in the left figure;  If one pair of the level lines escape to $\infty$ then there are two possibilities, the middle and the right figures.

With the dotted lines being drawn, the level lines from $\beta$ are now uniquely determined.  The rules for them are exactly same except the last two items must change from ``four" to ``three".  

\begin{figure}[ht]
\begin{center}
\includegraphics[width=0.35\textwidth]{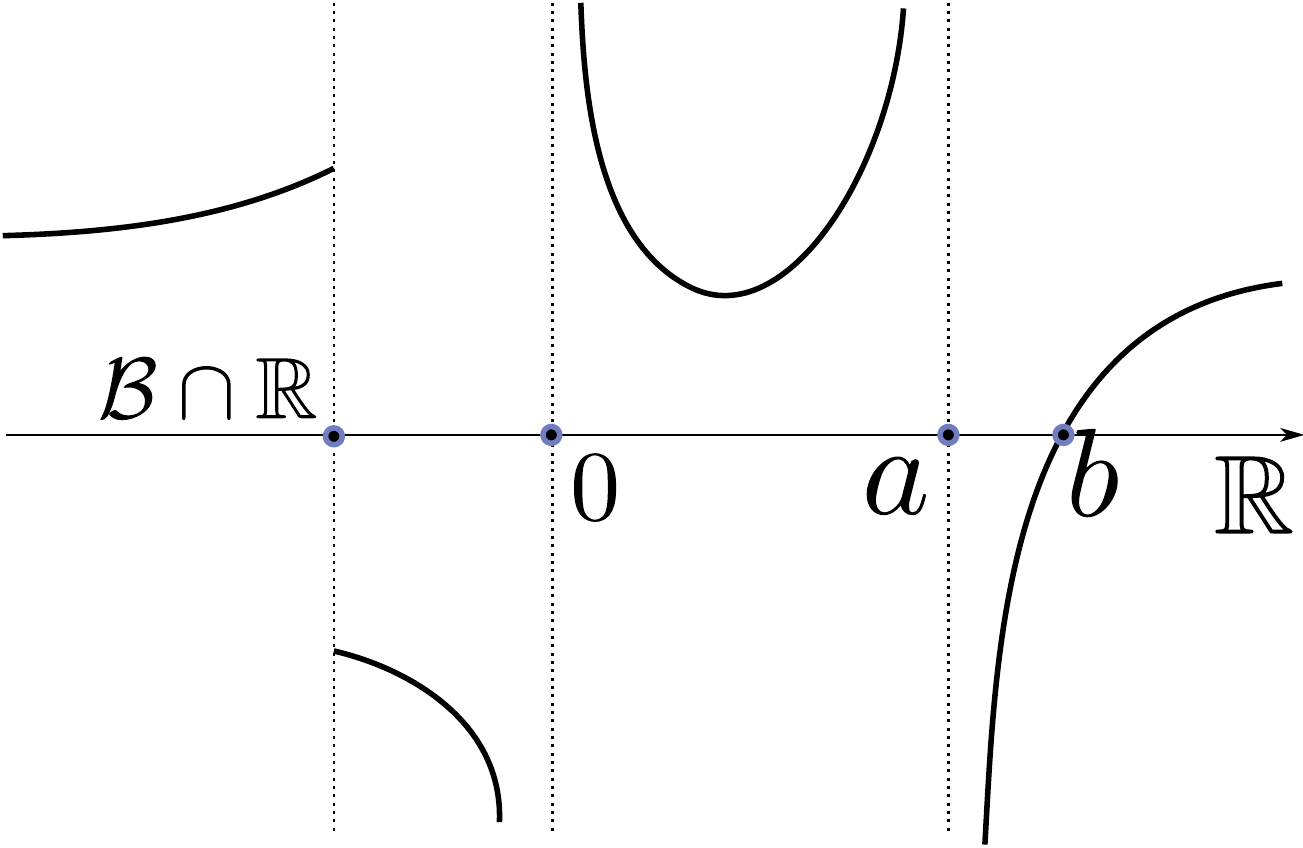}
\caption{$y(x)$ on $\R$: schematic view.}\label{fig-yR} 
\end{center}
\end{figure}
Out of the three possibilities, we prove that only the last one satisfies the inequality in Lemma \ref{lem-betab}.  Note that $y(+\infty)=a>0$.  Therefore $y(x)>0$ if $x>b$ and $y(x)<0$ if $a<x<b$, because there is no branch cut intersecting on $[0,+\infty)$. (See Figure \ref{fig-yR} for the behavior of $y(x)$ on $\R$.) Therefore, when we restrict $x\in(a,+\infty)$, $\int_b^x y(s)\,\d s$ reaches the minimum at $b$.  This means that, for the left and the middle cases in Figure \ref{fig_three}, the dotted lines are ``lower" than the regular lines and $\Re[\int_\beta^b y(z)\,\d z]<0$, which contradicts Lemma \ref{lem-betab}.  Therefore we have proven the following lemma.

\begin{lemma}\label{lem:calB} For all $0 < t < t_{c}$, the level lines from $b$ as well as those from $\beta$ and $\overline{\beta}$ are as shown in the last (boxed) configuration of Figure \ref{fig_three}. More precisely, the level lines from $b$ consist of three analytic arcs, one which encircles the point $a$ and two which diverge to $\infty$.  Two components of the level lines from $\beta$ are analytic arcs which terminates in $\overline{\beta}$, and together form a closed curve encircling $0$. 
\end{lemma}

We note that this proves the unique existence of ${\cal B}$ as we mention near \eqref{eq:qd}.  Note that the the equation \eqref{eq:qd} means that ${\cal B}$ is a {\em level line} according to Definition \ref{def:level}.

We want to choose the branch cut of $y(z)$ from one of the level lines from $\beta$.  
For this to be consistent with the previous definition of ${\cal B}$ (that appeared in Lemma \ref{lem-F}), we only need that the chosen level line be inside $K$.    Below we prove exactly this.    

The logic requires to use lots of hats ($~\widehat{}~$) temporarily.
We define $\widehat{\cal B}$ by the {\em unique} level line from $\beta$ to $\c\beta$ that intersects the negative real axis.  
Then we define the hatted version of the functions: $\widehat S(z)$, $\widehat y(z)$, $\widehat F(z)$, by the formula \eqref{Fz} and \eqref{y} {\em with the branch cut at $\widehat{\cal B}$}.   
After showing that $\widehat{\cal B}$ is inside $K$, we can replace all the hatted objects by the un-hatted ones.

\begin{defn}\label{def-B}
We define $\widehat{\cal B}$ by the {\em unique} level line from $\beta$ to $\c\beta$ that intersects the negative real axis.  We place the branch cut of $\widehat S(z)$ and $\widehat y(z)$ exactly at $\widehat{\cal B}$ using the formula \eqref{y}.  We also use \eqref{Uop} and \eqref{calU} accordingly to define $\widehat{\cal U}_\text{OP}$ and $\widehat{\cal U}$.\footnote{The existence of a point $z_0$ follows from $\{\beta,\c\beta\}\subset K$: starting from ${\cal B}$ that is completely inside $K$ (so that all $z\in\partial K$ satisfies $S(z)=\c z$) one cannot make $\widehat{\cal B}$ cross the whole of $\partial K$ by any continuous deformation ${\cal B}\to\widehat{\cal B}$.}
\end{defn}
\begin{lemma}\label{+measure} For any $p\in\widehat{\cal B}\setminus\{\beta,\c\beta\}$ let us define the normal derivative $\partial_{\bf n}:=n_x\partial_x+n_y \partial_y$ at $p$ along either of the two unit vectors ${\bf n}:=(n_x,n_y)$ that is perpendicular to $\widehat{\cal B}$ at $p$.  Choose $+$ and $-$ sides of the cut such that ${\bf n}$ points toward the $+$ side.  Then we have $\partial_{\bf n}\widehat{\cal U}_\text{OP}(p)\big|_+=-\partial_{\bf n}\widehat{\cal U}_\text{OP}(p)\big|_-=-|\widehat y(p)|$.
\end{lemma}
\begin{proof} The absolute value of $\partial_{\bf n}\widehat{\cal U}_\text{OP}(p)\big|_+$ is $|\widehat y(p)|$ because $\partial_{\bf n}\widehat{\cal U}_\text{OP}(p)\big|_+$ is purely real.  Defining $n=n_x+i n_y$ we have (the other $-$ side works similarly) 
\begin{equation}
\partial_{\bf n}\widehat{\cal U}_\text{OP}(p)\big|_+=(n\partial+\c n \c\partial)\widehat{\cal U}_\text{OP}(p)\big|_+=\Re\left( n\cdot \widehat y(p)\big|_+\right).
\end{equation}
By the definition of $\widehat{\cal B}$, $n\cdot \widehat y(p)$ is purely real and either of $\pm|\widehat y(p)|$.  The correct ($-$) sign can be determined at the point $\widehat{\cal B}\cap\R$ from Figure \ref{fig-yR}.
\end{proof}
\begin{lemma}\label{lem-nolocal}
There is no local minima of $\widehat{\cal U}(z)$ on $\widehat{\cal B}$.
\end{lemma} 
\begin{proof} Taking the same notations as in Lemma \ref{+measure} we have, for $p\in\widehat{\cal B}\setminus\{\beta,\c\beta\}$, 
\begin{equation}
\big[\partial_{\bf n}\widehat{\cal U}(p)\big]_+-\big[\partial_{\bf n}\widehat{\cal U}(p)\big]_-
=\big[\partial_{\bf n}\widehat{\cal U}_\text{OP}(p)\big]_+-\big[\partial_{\bf n}\widehat{\cal U}_\text{OP}(p)\big]_-=-2|\widehat y(p)|.
\end{equation}
This means that at least one of $\big[\partial_{\bf n}\widehat{\cal U}(z)\big]_+$ and $-\big[\partial_{\bf n}\widehat{\cal U}(z)\big]_-$ is negative.  Therefore,  $\widehat{\cal U}(z)$ decreases in one of the directions perpendicular to $\widehat{\cal B}$.

If $p=\beta$ we consider the expansion around $\beta$:
\begin{equation}
\begin{split}
\widehat{\cal U}(z)&=\text{(real analytic in ${\mathbb R}^2$)}+ \widehat{\cal U}_\text{OP}(z)
\\ &=\widehat{\cal U}(\beta)+c_1(x-\Re\beta)+c_2(y-\Im\beta)+\Re(c_0(z-\beta)^{3/2})+{\cal O}(|z-\beta|^2)
\end{split}
\end{equation}
where $c_1,c_2\in {\mathbb R}$ and $c_0\in {\mathbb C}$ is determined to match $\widehat y(z) = \frac{3}{2} c_0\sqrt{z-\beta}+{\cal O}((z-\beta)^{3/2})$ and $c_0$ is nonzero.  If either of $c_1$ or $c_2$ are non-vanishing, then there exists a point $z$ near $\beta$ such that $\widehat{\cal U}(z)<\widehat{\cal U}(\beta)$.   If both $c_1$ and $c_2$ vanish, there still exists such a point that lowers the value of $\widehat{\cal U}$ by the non-vanishing term $\Re(c_0(z-\beta)^{3/2})$.      This proves Lemma \ref{lem-nolocal}.
\end{proof}

\begin{lemma}\label{globalUpre} For the pre-critical case, we have $\widehat{\cal B}\subset K$, and, by taking ${\cal B}=\widehat{\cal B}$, ${\cal U}(z)$ has the global minima exactly at $\partial K$. 
\end{lemma}
\begin{proof}
If $\widehat{\cal B}$ intersects $\partial K$ at $p$ then any neighborhood of $p$ contains a region where ${\cal U}\equiv\widehat{\cal U}$, and $\widehat{\cal U}(p)={\cal U}(p)$ by continuity ($\widehat{\cal U}$ is continuous by definition).  By the previous lemma, $p$ cannot be the global minimum of $\widehat{\cal U}$ and there must be a point $z\notin \widehat{\cal B}$ where $\widehat{\cal U}$ reaches the global minimum.   This point satisfies $|\widehat F(z)|\neq 1$ because, for $|\widehat F(z)|=1$, one can show that $\widehat{\cal U}(z)=\widehat{\cal U}(p)$.   Then $z$ must satisfy $\partial \widehat{\cal U}(z)=\c z-\widehat S(z)=0$ and, by Lemma \ref{lem-double}, $z$ must be $b$.
Both ${\cal B}$ and $\widehat{\cal B}$ do not intersect $[0,+\infty)$ and we have $\widehat{\cal U}(b)={\cal U}(b)$, and $\widehat{\cal U}(b)-\widehat{\cal U}(p)={\cal U}(b)-{\cal U}(p)>0$ by Corollary \ref{lem-betab}.   Therefore $\widehat{\cal B}$ does not intersect $\partial K$, and $\partial K$ is the global minima of $\widehat{\cal U}$.
\end{proof}

\subsection{Branch cut \texorpdfstring{${\cal B}$}{calB}: The post-critical and the critical case}\label{sec-ypost}
\paragraph{Post-critical case:} 
We define $y$ by
\begin{equation}\label{ypost1}
y(z):=\pm \left(a+\frac{c}{z-a}-\frac{t+c}{z}\right),\quad \begin{cases} +: z\in{\rm Ext}({\cal B}),
\\-:  z\in{\rm Int}({\cal B}),
\end{cases}
\end{equation}
for some closed Jordan curve ${\cal B}$ that we determine below. In fact, this formula is consistent with the one \eqref{y} for the pre-critical case, if one uses the formulae in \eqref{Spost} respectively for ${\rm Ext}({\cal B})$ and ${\rm Int}({\cal B})$.
Note that the previous definition of $y$ in \eqref{ypost} with the definitions of $b$ and $\beta$ at \eqref{bbetapost} is equivalent to the above definition of $y$.
\begin{figure}[ht]
\begin{center}
\includegraphics[width=0.2\textwidth]{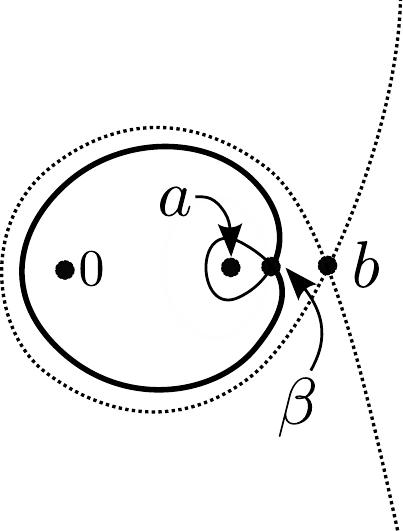}
\caption{Topological configuration of level lines from $\beta$ and $b$ for the post-critical case.}\label{fig-last} 
\end{center}
\end{figure}

One can check that $a$ is less than $\beta$.  With such configuration, there is topologically only one way to draw the level lines from $\beta$, which is illustrated in Figure \ref{fig-last}.  The proof is similar but simpler than that for the pre-critical case. Consider all the (three) possible ways of drawing level lines from $b$ using the rules in \eqref{ruleb}; there exists only one case where the level lines from $\beta$ can be drawn consistently with the level lines from $b$.    

\paragraph{Critical case:}
For $t-t_c={\cal O}(N^{-2/3})$ we choose $y(z)$ to be given by the equation \eqref{ypost} except that we extend the definition for $t<t_c$.  As a result, the two critical points $y(z)$, $b$ and $\beta$, are exactly given by \eqref{bbetapost} (taking the principal branch of the roots).  Their asymptotic behavior is given by
\begin{equation}\label{betabasymp}
b=b_c+\frac{c^{1/4}}{\sqrt a}\sqrt{t-t_c}+{\cal O}(t-t_c), \quad \beta=b_c-\frac{c^{1/4}}{\sqrt a}\sqrt{t-t_c}+{\cal O}(t-t_c),
\end{equation}
with $b_c:=a+\sqrt c$.
For $t>t_c$ the level lines from $\beta$ are exactly as in the post-critical case.
For $t<t_c$, one can verify, similarly as in the post-critical case, that the level lines from $\beta$ are arranged as shown at the left side in Figure \ref{default}.  

\begin{defn}[${\cal B}$, $y(z)$, ${\cal U}$, ${\cal U}_\text{OP}$]\label{def-post}
For the critical and the post-critical case, we define ${\cal B}$ by the unique level lines of $\beta$ that is a Jordan curve enclosing both $0$ and $a$.  We define ${\rm Ext}({\cal B})$ and ${\rm Int}({\cal B})$ respectively to be the exterior (including $\infty$) and the interior of the Jordan curve ${\cal B}$.   We define $y(z)$ by \eqref{ypost} and then $S(z)$ is accordingly defined over the extended domain $\C\setminus{\cal B}$ by 
\begin{equation}
S(z):=\frac{a}{2}+\frac{c}{2}\frac{1}{z-a}+\frac{c+t}{2z}-\frac{1}{2} y(z).
\end{equation}
Using the above defined $y(z)$ and $S(z)$, we define ${\cal U}_\text{OP}$ and ${\cal U}$ by \eqref{Uop} and by \eqref{calU} (where we can choose explicitly $z_0=\sqrt{t+c}$).
\end{defn}
We note that, for the post-critical case ${\cal B}$ is not inside $K$, and $S(z)$ does {\em not} match $S(z)$ at \eqref{Spost} in some part of $\{z:|z-a|<\sqrt c\}$. 

\begin{lemma}\label{globalUpost} For the post-critical case, ${\cal B}$ is contained in ${\bf P}(K)=\{z:|z|\leq\sqrt{t+c}\}$. Also ${\cal U}(z)$ has the global minima exactly at $\partial {\bf P}(K)=\{z:|z|=\sqrt{t+c}\}$. 
\end{lemma}
\begin{proof}
Since the proof is similar to the pre-critical case, some details are omitted.  It can be shown that the global minima of ${\cal U}$ cannot occur on ${\cal B}\setminus\{\beta\}$ as in the proof of Lemma \ref{lem-nolocal}.   One can show $\beta\in{\cal B}$ is in $\{z:|z-a|\leq\sqrt{c}\}$ by using $\sqrt{(t-a^2)^2-4 a^2 c}\geq t-a^2-2a\sqrt c$ with the formula for $\beta$ at \eqref{bbetapost}.  
The point $\beta$ is ruled out from the global minimum by explicitly comparing $\cal U(\beta)$ with ${\cal U}(\sqrt{t+c})$. Therefore, the global minima can only occur in $\{z:|z|=\sqrt{t+c}\}\cup\{z:|z-a|=\sqrt c\}$ (those are the only values of $z$ where $\partial{\cal U}(z)=\overline z-S(z)=0$ can possibly occur), and ${\cal B}\setminus\{\beta\}$ should not intersect both circles simultaneously.   This means that ${\cal B}$ is in the interior of the larger circle and (whether ${\cal B}$ intersects the smaller circle or not) the global minima of ${\cal U}$ is on $\{z:|z|=\sqrt{t+c}\}$. 
\end{proof}

\begin{lemma}\label{globalUcrit} For $t=t_c$, ${\cal B}\setminus\{\beta\}$ is contained in $K=\{z:|z|<\sqrt{t+c}\text{ and } |z-a|> \sqrt c\}$. Also ${\cal U}(z)$ has the global minima exactly at $\partial K=\{z:|z|=\sqrt{t+c}\}\cup\{z:|z-a|=\sqrt c\}$. 
\end{lemma}
\begin{proof} 
By the same logic as in the pre-critical case, there is no global minima on ${\cal B}\setminus \{\beta\}$.
Then the only possible global minima is where $S(z) = \overline z$, which is the big circle and/or the small circle.
Both are indeed the global minima because they are the limit of $\partial K$ as $t \nearrow t_c$ where $\partial K$ is the global minima for all $t < t_c$ by Lemma \ref{globalUpre}.
As a consequence, the cut ${\cal B}\setminus \{\beta\}$ is strictly inside K.
 \end{proof} 

\section{Orthogonal polynomial on a contour}\label{sec-3}

As the first step toward the asymptotic analysis of $P_{n,N}$ we show that $P_{n,N}$ is an orthogonal polynomial with respect to a non-hermitian, $n$-dependent orthogonality measure on certain contours in the complex plane.   This observation hinges on the following lemma which relates the area integral to a contour integral.
\begin{lemma}
\label{lemma:moments}
The following identity holds
\begin{equation}\label{area=contour}
\iint_\C z^j \c {(z-a)}^k \left| z-a\right|^{2Nc}{\rm e}^{-N z \c z} \d A( z) = \frac {\Gamma(Nc + k + 1)}{2i \, N^{Nc+k+1}} \oint_\Gamma z^{j} w_{k+1,N}(z)\, \d z,
\end{equation}
where
\begin{equation}
w_{n,N}(z) := \frac{(z-a)^{Nc} e^{-Naz}}{z^{Nc+n}}\ , \qquad z \in \C\setminus [0,a]\ .
\end{equation}
The contour of integration $\Gamma$ is a closed loop enclosing both $0$ and $a$ counterclockwise.
\end{lemma}
\begin{proof}
The integrand in the left hand side of \eqref{area=contour} can be re-expressed using the function $\chi$, defined as follows:
\begin{equation}\label{Fid}
{\chi}(z):=(z-a)^{Nc} \int_{a}^{\c{z}}(s-a)^{Nc+k} e^{-Nzs}\d s, \quad z\in\C\setminus(-\infty,a),
\end{equation}
where the contour of integration is chosen so as not to intersect the half-line $(-\infty, a)$.
Clearly $\chi$ satisfies 
\begin{equation}
	\c\partial {\chi}(z) =( \c{z-a})^k |z-a|^{2Nc} {\rm e}^{-N |z|^2}.
\end{equation}
The degree to which $\chi$ is smooth can best be seen from the alternative representation, easily obtainable by tricks of calculus:
\begin{equation}
	\chi(z)=|z-a|^{2Nc}(\overline z-a)^{k+1}{\rm e}^{-N a z}\int_0^1 x^{Nc+k}e^{-Nz(\overline z-a)x}dx.
\end{equation}
If $Nc$ is a positive integer, then clearly $\chi$ is infinitely differentiable.  Otherwise $\chi$ has $\lfloor Nc\rfloor$ continuous partial derivatives.  In our case, $Nc$ tends to infinity and thus $\chi$ is at least twice differentiable (which will be enough to apply Stokes' theorem below).

The integral in ${\chi}(z)$ can be divided into two terms by 
\begin{equation}\label{gammaerror}
\begin{split}
\int_{a}^{\c{z}}(s-a)^{Nc+k} e^{-Nzs}\d s&= \int_{a}^{\c{z}\times\infty}(s-a)^{Nc+k} e^{-Nzs}\d s-\int_{\c z}^{\c{z}\times\infty}(s-a)^{Nc+k} e^{-Nzs}\d s
\\&=\frac{\Gamma(Nc + k + 1)e^{-Naz}}{(Nz)^{Nc+k+1}}+{\cal O}\left(e^{-N|z|^2/2}\right),
\end{split}
\end{equation}
where the branch cut of $z^{Nc+k+1}$ lies on $(-\infty,0]$. 

We derive the last equality below.
For fixed $N, c,k$, there exist $R>0$ such that the above integrand satisfies 
$
\left|(s-a)^{Nc+k} e^{-Nzs}\right|\leq \left|e^{N z s/2}e^{-Nzs}\right| 
$ for all $s\in [\c z,\c z \times \infty)$ and $|z|>R$.  Therefore the last term in the first line in \eqref{gammaerror} is bounded by
\begin{equation}
\left|\int_{\c z}^{\c{z}\times\infty}(s-a)^{Nc+k} e^{-Nzs}\d s\right|\leq \int_{R}^{\infty}e^{-Nzs/2}\d s =2\frac{\e^{-N|z|^2/2}}{N|z|}, \quad |z|>R.
\end{equation}
The first term after the first equality in \eqref{gammaerror} becomes, by the change of variables: $X=z(s-a)$,
\begin{equation}
\int_{a}^{\c{z}\times\infty}(s-a)^{Nc+k} e^{-Nzs}\d s=\frac{e^{-Nza}}{z^{Nc+k+1}}\int_0^{+\infty}X^{Nc+k} e^{-NX}\d X
=\frac{\Gamma(Nc + k + 1)e^{-Naz}}{(Nz)^{Nc+k+1}}.
\end{equation}
Since $\chi$ is twice differentiable, we can apply Stokes' theorem to get
\begin{align}
(\text{LHS of \eqref{area=contour}}) &=\lim_{r\to\infty}\iint_{|z|<r} z^j \c {(z-a)}^k \left| z-a\right|^{2Nc}{\rm e}^{-N z \c z} \d A( z)
\\ \label{stokes}&=\lim_{r\to\infty}\iint_\C z^j \,\c\partial {\chi}(z) \d A(z)=\lim_{r\to\infty}\frac{1}{2i}\oint_{|z|=r} z^j \, {\chi}(z) \d z
\\&=\lim_{r\to\infty}\frac{1}{2i}\oint_{|z|=r} z^j (z-a)^{Nc}\left(\frac{\Gamma(Nc + k + 1)e^{-Naz}}{(Nz)^{Nc+k+1}}+{\cal O}\left(e^{-N|z|^2/2}\right)\right) \d z.
\end{align}
The first term is independent of $r$ as long as the circle $|z|=r$ encloses $0$ and $a$; and the last term vanishes in the limit $r\to\infty$.  
\end{proof}

This lemma allows us to view $P_{n,N}$ as an ordinary orthogonal polynomial: indeed the system (\ref{2d_ortho}) is equivalent to
\begin{equation}\label{2d_ortho1}
\oint_{\Gamma}P_{n,N}(z)\frac{(z-a)^{Nc} e^{-Naz}}{z^{Nc+k+1}}dz =\begin{cases} 0\ , & k =0,1,\cdots, n-1,
\\\displaystyle 2i\frac{N^{Nc+n+1}}{\Gamma(Nc+n+1)} h_{n,N}\ , &k=n.
\end{cases}
\end{equation}
By relabeling $k\to n-1-j$, we get a standard system of non-hermitian orthogonality relations
\begin{equation}
\oint_{\Gamma}P_{n,N}(z)z^j\frac{(z-a)^{Nc} e^{-Naz}}{z^{Nc+n}}dz = 0\qquad  (j =0,1,\dots, n-1).
\label{contour_ortho}
\end{equation}

\begin{remark}
The proof of Lemma \ref{lemma:moments} can be extended to treat weights of the form 
\begin{equation*}
W(z,\c z) = 
{\rm e}^{-N |z|^2 +\Re \left(\alpha z^2 + \beta z\right)}\prod_{j=1}^K |z-a_j|^{Nc_j}
\end{equation*}
and allows to relate the bimoments $\int_\C z^j \c z^k W(z,\c z)\d^2 z$ with the moments of an iterated contour integral; this realizes the orthogonal polynomials as the biorthogonal polynomials appearing in \cite{bertobisemiclass} for a specific semiclassical bilinear moment functional. In particular one could use the formalism in loc. cit. to produce a corresponding Riemann--Hilbert problem (whose dimension depend on the integer $K$, however). We do not pursue this issue here.
\end{remark}

\subsection{The Riemann-Hilbert Problem}

It is well--known that the polynomial $P_{n,N}(z)$ can be characterized by a $2\times 2$ Riemann-Hilbert problem \cite{fokas_its_kitaev} which we recall in our setting.
Let us define 
\begin{equation}\label{nu}
\nu_k:= \oint_{\Gamma} z^k w_{n,N}(z)\d z.
\end{equation}
Introduce the polynomial  
\begin{equation}
Q_{n-1,N}(z):=\frac 1{\det\left[\nu_{i+j}\right]_{0\leq i,j\leq n-1}} \det \left[
\begin{array}{ccccc}
\nu_0 & \nu_1 & \dots & \nu_{n-1}\\
\nu_1 & \nu_2 & \dots & \nu_{n}\\
\vdots &&&\vdots \\
\nu_{n-2}&\dots&&\nu_{2n-3}\\
1 & z & \dots & z^{n-1}
\end{array}
\right]
\label{Qn}
\end{equation}
Note that $Q_{n-1,N}$ is not necessarily monic and its degree may be less than $n-1$: its existence relies uniquely on the nonvanishing of the determinant in the denominator, which will be shown below in Proposition \ref{nonzero}. It is also orthogonal to the powers $1, z,\dots z^{n-2}$. Define the matrix 
\begin{equation}
Y(z) = \pmtwo
{P_{n,N}(z)}{\frac{1}{2 \pi i} \int_{\Gamma} \frac{P_{n,N}(z')}{z' - z} w_{n,N}(z') dz'}{-{2 \pi i } Q_{n-1,N}(z)}{- \int_{\Gamma} \frac{Q_{n-1,N}(z')}{z' - z} w_{n,N}(z') dz'} \ .
\label{Ymatrix}
\end{equation}
Then $Y$ satisfies the standard \cite{fokas_its_kitaev} Riemann-Hilbert problem:
\begin{equation}
\label{YRHP}
\begin{cases}
Y(z)\mbox{ is holomorphic in } \C \setminus \Gamma;\\
Y_{+} (z)= Y_{-}(z) \ \pmtwo
{1}{w_{n,N}(z)}
{0}{1}\ ,\quad z\in\Gamma ;
\vspace{0.1cm}\\\displaystyle
Y(z) = \left(I +\order{\frac{1}{z}}\right)z^{n\sigma_3}\ , \quad z \to \infty;
\end{cases}
\end{equation}
{\em if and only if} $\det\left[\nu_{i+j}\right]_{0\leq i,j\leq n-1} \neq 0$: this  is seen by noticing that if $Q_{n,N}$ exists as given by (\ref{Qn}) then 
\begin{equation}
- \int_{\Gamma} \frac{Q_{n-1,N}(z')}{z' - z} w_{n,N}(z') dz' \sim z^{-n} (1 + \mathcal O(z^{-1}))
\end{equation}
We thus need to show the nonvanishing of the determinant in the following proposition. 
\begin{prop}
\label{nonzero}
The determinant $\det [\nu_{i+j}]_{0\leq i,j\leq n-1}$ does not vanish.
\end{prop}
The proof can be found in Appendix \ref{app-nonzero}.

\subsection{The \texorpdfstring{$g$}{g}-function and the 1st transform}

Given the Riemann--Hilbert characterization in eqs. (\ref{YRHP}), the asymptotic study can be addressed using the Deift--Zhou \cite{DKMVZ} steepest descent method by means of an appropriate $g$--function. 

In our context we cannot rely upon the minimization of an energy functional as in the case of ordinary orthogonal polynomials on the real line.   Therefore we will resort to constructing the $g$-function in terms of $y(z)$ that is defined at \eqref{y}\eqref{ypost} and in Definition \ref{def-post} for all three cases.

The weight function can be expressed in terms of a potential $V$ (of the semiclassical type \cite{marcellan}) by
\begin{equation}
w_{n,N}(z) = z^{-r}e^{-NV(z)},\quad V(z) := az - c \ln{(z - a)} + ( c + t) \ln{z}.
\end{equation}
We choose the cut of the logarithms at $[0,+\infty)$.

\begin{lemma}
[Definition of $\Gamma_{b\beta}$ and $\Gamma_{b\c\beta}$]\label{lem-steepest}
The steepest ``ascent" path (the maximal elevation) of ${\cal U}_\text{OP}$ from $\beta$ reaches $b$.  Equivalently, the steepest descent paths from $b$ reach $\beta$ and $\c\beta$; Let us denote the former (that reaches $\beta$) by $\Gamma_{b\beta}$ and the latter by $\Gamma_{b\c\beta}$.
\end{lemma}
\begin{proof} In the last part of Appendix \ref{app-lem-UOP2d} we show that $\int_\beta^b y(s)\,\d s$ is purely real.    This means that $b$ is on the steepest ascent path from $\beta$.  (See Figure \ref{fig_31} for a numerical confirmation.)
\end{proof}

\begin{figure}[ht]
\begin{center}
\includegraphics[width=0.5\textwidth]{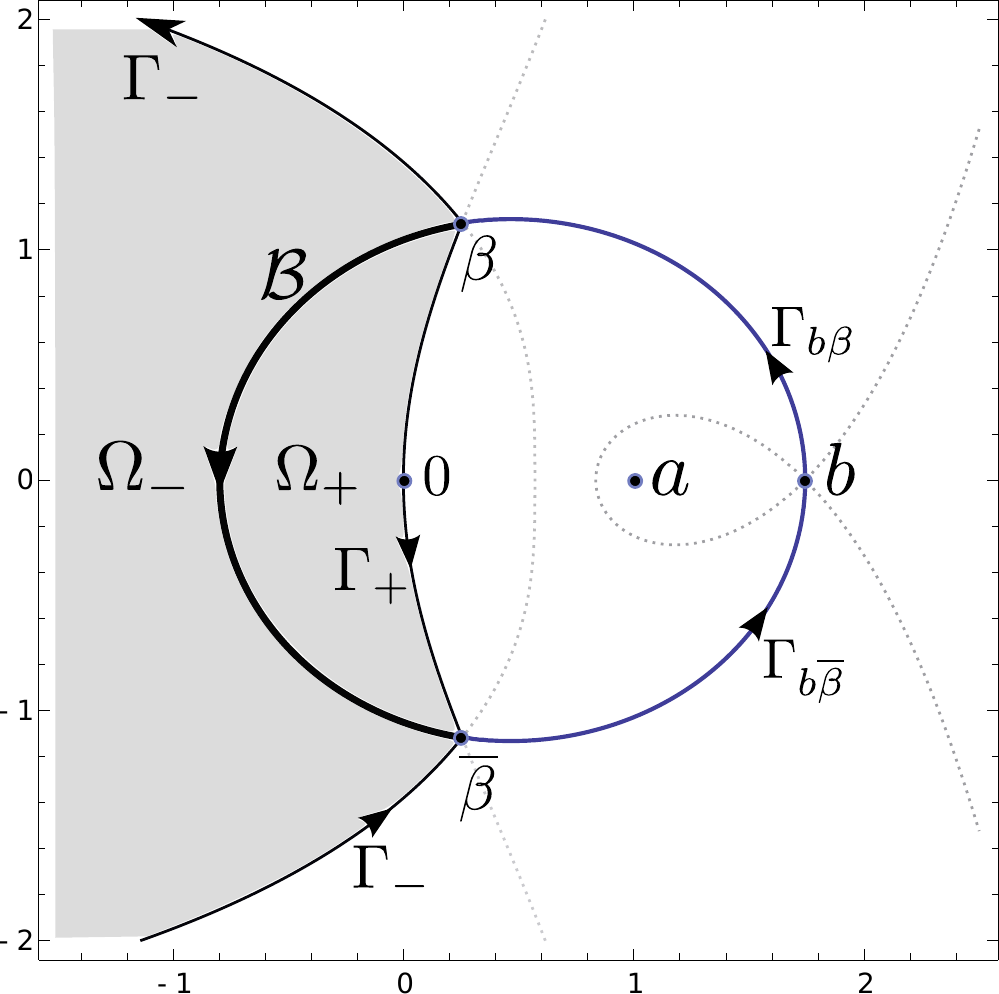}
\caption{${\cal B}$ (thick line); level lines (dotted lines); steepest descent/ascent paths (the other lines).  All the lines are with respect to the function ${\cal U}_\text{OP}$ and emanate from $\beta$ or $b$.    We close the contour $\Gamma_-$ at some region $\Re(z)<0$ (by deviating from the steepest descent paths) such that $\Omega_-$ is a bounded region. (All the lines are numerically plotted for $t=c=a=1$.) }\label{fig_31} 
\end{center}
\end{figure}

\begin{figure}[ht]
\begin{center}
\includegraphics[width=0.4\textwidth]{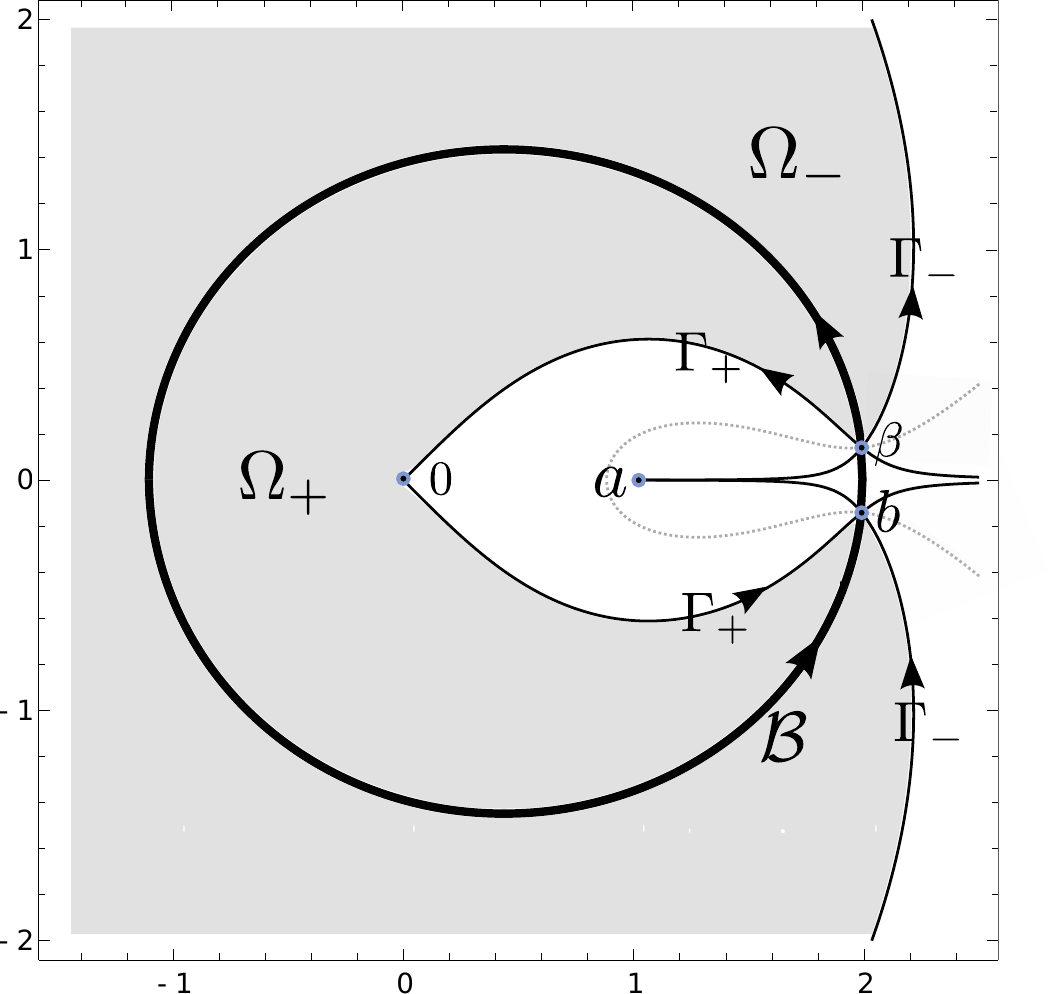}
\includegraphics[width=0.38\textwidth]{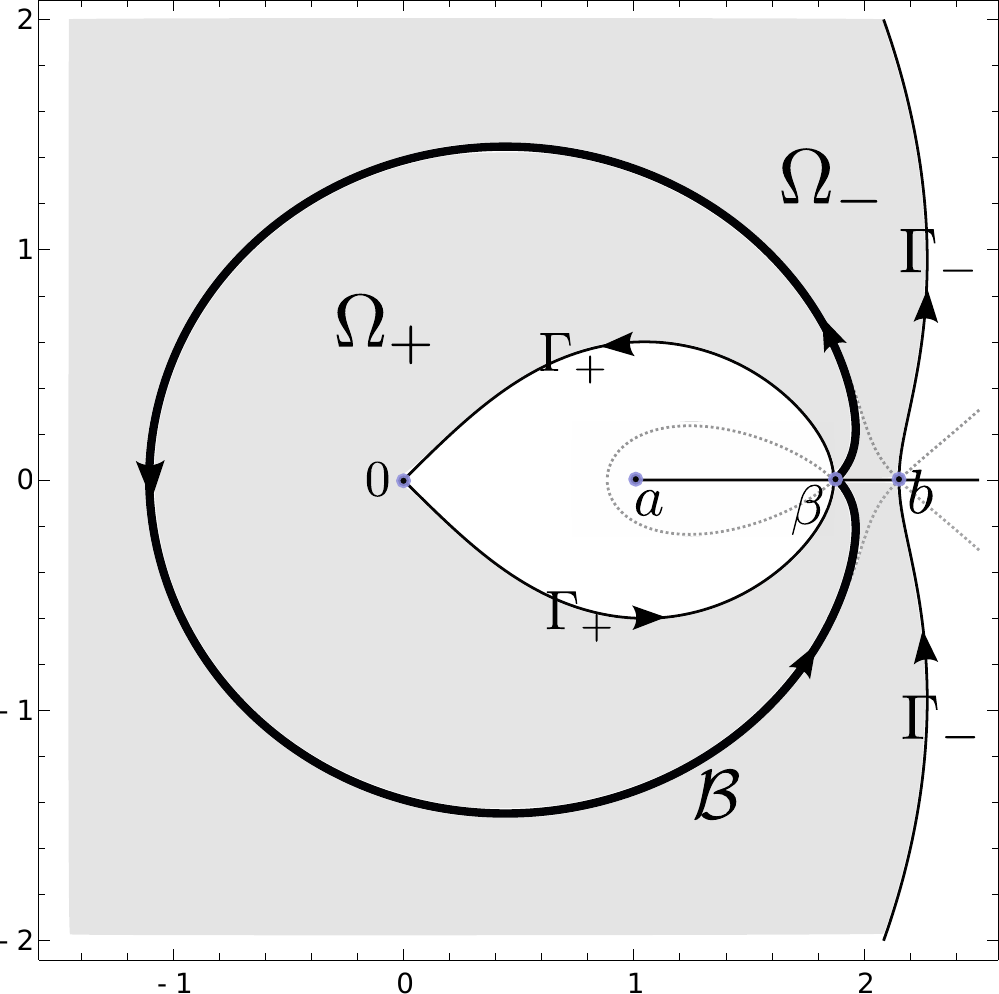}
\caption{The level curves and contours for the critical case with $t<t_c$ on the left; and $t>t_c$ (for both the critical and the post-critical case) on the right. The branch cuts ${\cal B}$ (thick lines) are chosen by one of the level lines from $\beta$; the other level lines from $\beta$ and $b$ are shown with dotted lines; the other lines are steepest descent/ascent paths. (All the lines are numerically plotted for $c=a=1$ and $t=t_c-0.02$ (left) and $t=t_c+0.02$ (right) where $t_c=3$.)}
\label{default}
\end{center}
\end{figure}

\begin{defn}\label{def-phi} For all $t$, we define $\phi(z)$ by
\begin{equation}
\phi(z):=\int_\beta^z y(s)\,\d s
\end{equation}
where the integration contour lies in $\C\setminus({\cal B}\cup \Gamma_{b\beta}\cup[b,+\infty))$ for the pre-critical case and in $\C\setminus({\cal B}\cup[\beta,+\infty))$ for the other cases.  We have ${\cal U}_\text{OP}(z)=\Re \phi(z)$ by \eqref{Uop}.
\end{defn}

\noindent {\bf Remark.}
By this, we replace the definition of $\phi$ given at Definition \ref{def-phi0}. Note that this definition is exactly the same as the Definition \ref{def-phi0} except the additional ``cuts'' to remove the multi-valuedness of $\phi$.  This difference is irrelevant when we state our main theorems in Section \ref{sec:intro}, because $\phi$ always appears in the exponentiated form, ${\rm e}^{N \phi(z)}$, which is (as the reader may verify) single-valued because the multivaluedness of $\phi$ is additive by integer multiples of $4i\pi t$ and $N = (n-r)/t$ \eqref{eq:Nnrt} with $n, r$ being integers.

Below, we recall that the $g$-function is defined in Definition \ref{def-g} through $\phi$ that is re-defined {\em above}.

\begin{lemma}\label{lem+}
For all the cases except for the critical case, we have
\begin{equation}
g(z)=\frac{1}{2\pi t}\int_{\cal B} \log(z-s)\,|y(s)|\,|\d s|
\end{equation}
where we take the branch cut of the logarithm in ${\cal B}\cup\Gamma_{b\beta}\cup[b,+\infty)$ for the pre-critical case and ${\cal B}\cup[\beta,+\infty)$ for the post-critical case.
\end{lemma}
\begin{proof}
One can check that $g'(z)$ is analytic away from ${\cal B}$; the singularities of $y(z)$ as described at \eqref{ysing} exactly cancel those of $V'(z)$.  Let us only consider the pre-critical case as the post-critical case is similar.  To check the jump discontinuity of $g$ on $z\in{\cal B}$ we evaluate
\begin{equation}
\big[g(z)\big]^+_-=-\frac{1}{2t}\int_{z_-}^{z_+} y(s)\,\d s=-\frac{1}{2t}\left[\int_{z_-}^{\c\beta} y(s)_-\,\d s-\int_{z_+}^{\c\beta} y(s)_+\,\d s\right]=\frac{i}{t}\int_z^{\c\beta} |y(s)|\,|\d s|
\end{equation}
where $\pm$ sides are with respect to the orientation of ${\cal B}$ directed from ${\beta}$ to $\c\beta$ for the pre-critical case (and counter-clockwise for the post-critical case).  The integration contour is taken away from the cuts of both $V$ and $\phi$.  In the last equality, we use that $y(s)\,\d s = \pm i |y(s)|\,|\d s|$ for the infinitesimal line segment $\d s$ tangential to ${\cal B}$.  The sign can be determined at the intersection of ${\cal B}$ with the real axis using Figure \ref{fig-yR}; as in the proof of Lemma \ref{+measure}.  The above calculation gives exactly the jump of the proposed formula for the $g$--function.
\end{proof}

\begin{lemma}\label{lem-g} On the exterior of ${\bf P}(K)$, the $g$--function is the complex logarithmic potential from the probability measure that is uniformly supported on $K$, i.e.
\begin{equation}
\Re g(z)=\frac{1}{\pi t}\int_K \log|z-w|\,\d A(w),\quad z\notin {\bf P}(K).
\end{equation}  
Above, the exponential is applied to disregard the (unimportant and arbitrary) branch cut. 
\end{lemma}
\begin{proof} It is enough to show that
\begin{equation}
g'(z)=\frac{1}{2t}\left(V'(z)- y(z)\right)=\frac{1}{t}\left(S(z)-\frac{c}{z-a}\right)=\frac{1}{t}\partial\left(Q(z)-{\cal U}(z)\right)=\frac{1}{\pi t} \int_K \frac{\d A(w)}{z-w},
\end{equation}
where each equality is from \eqref{ydef},\eqref{y},\eqref{calU}, and \eqref{Phi2D}, in the order of appearances.
\end{proof}

From (the proof of) Lemma \ref{lem+} the sign of $\Re\phi$ is {\em negative} on both sides of the cut $\cal B$.

We may choose the contour $\Gamma$ (of orthogonality) by any continuous deformation within $\C\setminus[0,a]$. We choose 
\begin{equation}\label{Gamma}
\Gamma:= \begin{cases}{\cal B}\cup \Gamma_{b\beta}\cup\Gamma_{b\c\beta} \quad &\text{pre-critical},
\\{\cal B} & \text{critical and post-critical}.\end{cases}
\end{equation}
We write some useful relations that follow from ${\cal U}_\text{OP}=\Re \phi$.
\begin{equation}\begin{split}
w_{n,N}(z) e^{ Nt \left(
g_{+}(z) + g_{-}(z) - \ell
\right)}  = \frac 1 {z^r}e^{-N(\phi_{+}(z) +\phi_{-}(z))/2}=\begin{cases}1/{z^r},\quad& z\in{\cal B};
\\ e^{-N{\cal U}_\text{OP}}/z^r, \quad &z\in \Gamma\setminus{\cal B}.
\end{cases}\end{split}
\end{equation}

\begin{defn}[$\Gamma_\pm$ and $\Omega_\pm$]\label{def-Omega}
We choose contours $\Gamma_+$ and $\Gamma_-$ following the steepest descent paths (as shown in Figure \ref{fig_31} and Figure \ref{default}) from $\beta$ (and $\c\beta$ for the pre-critical case) such that $\Re\phi(z)={\cal U}_\text{OP}(z)<0$ on those contours.   We choose $\Gamma_-$ to be the unbounded ones.  We then deform $\Gamma_-$, away from the steepest path, to make it a {\em single bounded curve} such that ${\cal U}_\text{OP}(z)<0$ on the deformed curve.   Then the domains $\Omega_\pm$ are defined by the open sets enclosed by ${\cal B}$ and $\Gamma_\pm$ respectively.   See Figure \ref{fig_31} for the pre-critical case and the right side in Figure \ref{default} for the post-critical case. 
\end{defn}

On $\Gamma$ \eqref{Gamma}, we have
\begin{equation}
Y_+(z)=Y_-(z)
\left[
\begin{array}{cc}
1 & w_{n,N}(z)\\
0 & 1
\end{array}
\right].
\end{equation}
Correspondingly we introduce the matrix $A(z)$ as follows:
\begin{equation}\label{Adef}
A(z):=e^{ -\frac{tN \ell}{2} \sigma_{3}} Y(z)
\left\{
\begin{array}{cc}
I  & z\in \C\setminus \Omega_+\cup \Omega_-\\
&\\
\left[
\begin{array}{cc}
1 & 0\\
-1 /w_{n,N} & 1
\end{array}
\right]  & z\in \Omega_+\\
&\\
\left[
\begin{array}{cc}
1 & 0\\
1/w_{n,N} & 1
\end{array}
\right]  & z\in \Omega_-\\
\end{array}
\right\} 
e^{ -tN \left( g(z) - \frac{ \ell}{2} \right) \sigma_{3}}\ .
\end{equation}
The RHP for $A$ is
\begin{equation}\label{RHPA}
\left\{
\begin{array}{cl}
A_+=A_-\pmtwo{1}{0}{z^re^{N\phi}}{1} & \mbox{on } \Gamma_\pm;\\
&\\
A_+=A_-\pmtwo{0}{ z^{-r} }{-z^r }{0} & \mbox{on } {\cal B};\\
&\\
A_+=A_-\pmtwo{1}{z^{-r} e^{-N{\cal U}_\text{OP}}}{0}{1} & \mbox{on } \Gamma\setminus {{\cal B}};
\\&
\\A(z)=\displaystyle\left(I+{\cal O}\left(\frac1z\right)\right)\pmtwo{z^r }{0}{0}{z^{-r}} & z\rightarrow\infty.
\end{array}
\right.
\end{equation}
The remaining jump of $g$ is $\pm 2i\pi$ and hence ${\rm e}^{-tN g(z)}$ does not have additional jumps besides the ones indicated above.
On $\Gamma_\pm$, ${\cal U}_\text{OP}=\Re \phi$ is negative, and on $\Gamma\setminus{\cal B}$ (pre-critical only) it is positive and uniformly so except near the endpoints $\beta$ and $\overline\beta$. 
At any finite distance away from ${\cal B}$, the jump matrices in (\ref{RHPA}) are exponentially close to the identity jump in any  $L^p$ norms and in $L^\infty$ norm.  Following the nonlinear steepest descent method, one first constructs a model RHP where only the jump exactly on ${\cal B}$ is retained. This is done below.

\section{The pre-critical case: \texorpdfstring{$t<t_c$}{ttc}}\label{sec-pre}

\paragraph{Outer parametrix:}
We define the following model RHP for $\Phi$ by
\begin{equation}\label{PsiRHP}
\begin{cases}
&\Phi_+=\Phi_-\pmtwo{0}{z^{-r}}{-z^{r}}{0}  \quad \mbox{on } {\cal B}\ ,
\\ &\displaystyle\Phi=\left(I+{\cal O}\left(\frac1z\right)\right)\pmtwo{z^r}{0}{0}{z^{-r}}\quad z\to\infty\ .
\end{cases}
\end{equation}
There are several ways of expressing the solution $\Phi$ and we provide the one that is most relevant to the geometry of $K$.

\begin{wrapfigure}{r}{0.3\textwidth}
\vspace{-0.cm}
\includegraphics[width=0.3\textwidth]{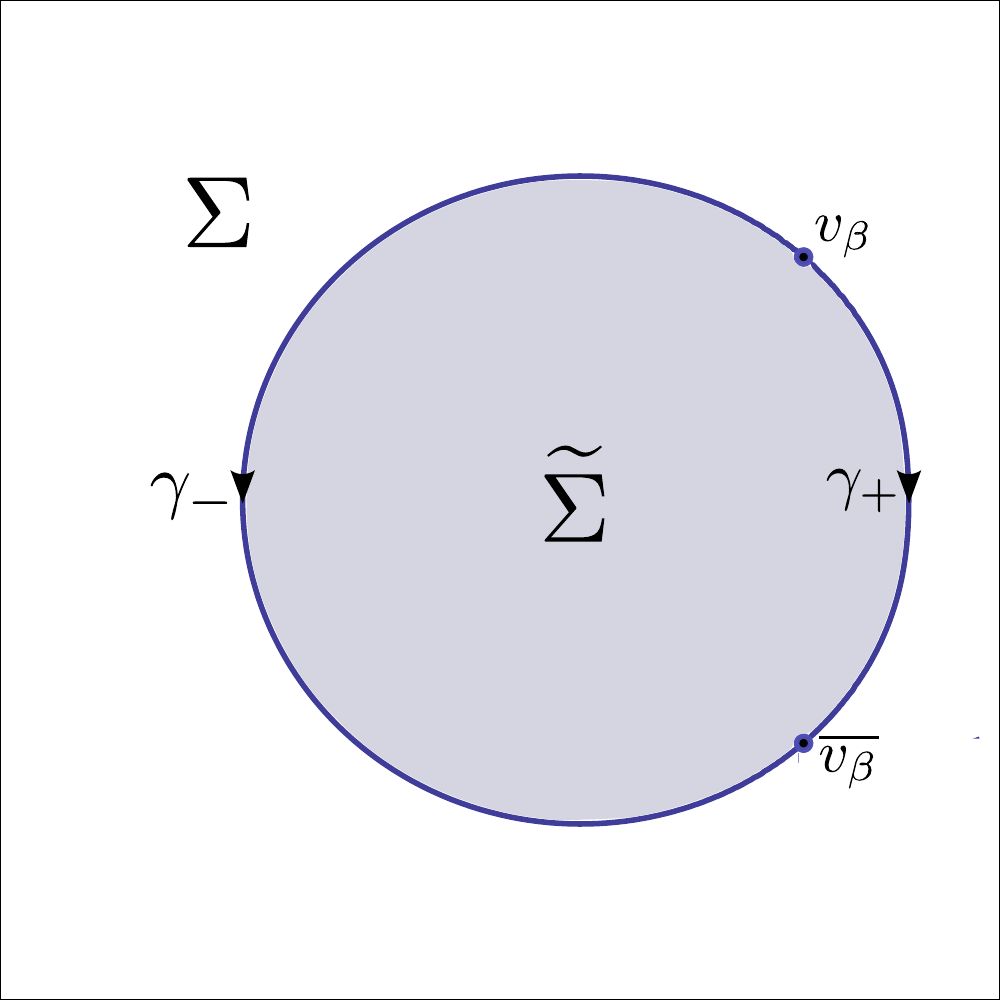}
\caption{Uniformizing plane; See the text.}\label{fig_map} 
\vspace{-1.6cm}
\end{wrapfigure}

We use $f:\widehat\C\to\widehat\C$ as the uniformization mapping of the domain with the branch cut $\C\setminus{\cal B}$. 
This map is 2-to-1 branched covering map and we have defined the inverse map $F$ at \eqref{Fz}. There are regions, $\Sigma=F(\C\setminus{\cal B})$ and $\widetilde\Sigma={\rm deck}\circ F(\C\setminus{\cal B})$, such that $\Sigma$ and $\widetilde\Sigma$ are each conformally equivalent to $\C\setminus{\cal B}$.
We define $\gamma_\pm =[F({\cal B})]_\pm$, the two preimages of ${\cal B}$ under $f$ when approached from $\pm$ sides.
See Figure \ref{fig_map}.

\begin{prop}
\label{Szegofun}  The scalar function given by
\begin{equation}\label{SzegoD}
D(z):=\sqrt{\frac{\rho}{\alpha}}\, F(z)
\end{equation}
satisfies the following properties:
\begin{enumerate}
\item $D(z)$ is analytic and nonzero in $\C\setminus \cal B$ ;
\item $D_{+}(z)D_{-}(z) = z$, $z\in \cal B$;
\item $D_\infty:=\lim_{z\to \infty} D(z)/z \in \R_+$. 
\end{enumerate}
\end{prop}
\begin{proof} The proof relies on the identity
\begin{equation}\label{prod}
f(v)=(\rho/\alpha) \,v \,{\rm deck}(v).
\end{equation}
This can be shown by
\begin{equation}\begin{split}
0&=(v-\alpha)({\rm deck}(v)-\alpha)+\kappa/\rho=v\,{\rm deck}(v)-\alpha(v+{\rm deck}(v))+\alpha^2+\kappa/\rho
\\&=v\,{\rm deck}(v)-(\alpha/\rho)\,f(v),
\end{split}
\end{equation}
where the last equality comes from $v+{\rm deck}(v)=(f(v)+|\beta|)/\rho$ at \eqref{sum} and $|\beta|=(\rho/\alpha)(\alpha^2+\kappa/\rho)$ that is shown right below \eqref{sum}. 
Using \eqref{prod} and that $F(z)|_+={\rm deck}(F(z))|_-$ on $z\in{\cal B}$, we have
\begin{equation}
D_+(z)D_-(z)=\frac{\rho}{\alpha}\,{\rm deck}(F(z))\cdot F(z)=f(F(z))=z,\quad z\in{\cal B}.
\end{equation}
The last condition is checked by $D_\infty=\lim_{z\to\infty} \sqrt{\rho/\alpha}\,F(z)/z=1/\sqrt{\rho\,\alpha}>0$.
\end{proof}
\begin{prop}
\label{solPsiRHP}
The solution to \eqref{PsiRHP} is given by
\begin{equation}\label{Phipre}
\Phi(z)= \sqrt{\rho F'(z)}\left[\begin{array}{cc} 1 &\displaystyle \frac{(\rho\alpha)^r(\kappa\rho)^{1/2}}{\rho F(z)-\rho\alpha} \\ \displaystyle\frac{-(\rho\alpha)^{-r}(\kappa\rho)^{1/2}}{\rho F(z)-\rho\alpha} &\displaystyle 1\end{array}\right]
\left( \rho F(z)\right)^{r\sigma_3}.
\end{equation}
We choose the principal branch of the square root, i.e. $\sqrt{\rho F'(z)}= 1+{\cal O}(1/z)$ as $z\to\infty$.  
\end{prop}
\begin{proof} Though the solution may be checked by a direct calculation, here we present a self-contained derivation (which is essentially the spinor construction in \cite{BertoMo}).

For the solution $\Phi$, there exists an analytic matrix valued function ${\bf H}:\Sigma\to {\rm SL}(2,\widehat\C)$ such that
\begin{equation}
{\bf H}(F(z))=(D_\infty)^{r\sigma_3}\Phi(z)\,D(z)^{-r\sigma_3}.
\end{equation}
This satisfies the RHP \eqref{PsiRHP} of $\Phi$ for $r=0$, i.e.
\begin{equation}\label{RHPH}
{\bf H}(F(z))_+={\bf H}(F(z))_-\left[\begin{array}{cc}0 & 1 \\-1 & 0\end{array}\right];\quad {\bf H}(F(z))=I+{\cal O}(1/z), ~~z\to\infty.
\end{equation}
We may freely extend ${\bf H}$ on the whole complex plane by putting
\begin{equation}\label{Hv}
{\bf H}(v):={\bf H}({\rm deck}(v)) \left[\begin{array}{cc}0 & 1 \\-1 & 0\end{array}\right],\quad v\in\widetilde\Sigma.
\end{equation}
This makes ${\bf H}$ to be continuous on $\gamma_+$, by the jump condition in \eqref{RHPH}.   On $v\in\gamma_-$ we evaluate the jump as we approach from $\Sigma$ and $\widetilde\Sigma$ by
\begin{equation}
{\bf H}(v)\big|_\Sigma={\bf H}({\rm deck}(v)) \,\left[\begin{array}{cc}0 & -1 \\1 & 0\end{array}\right]=-{\bf H}(v)\big|_{\widetilde\Sigma},\quad v\in\gamma_-.
\end{equation}
The last equality uses \eqref{Hv} with the continuity of ${\bf H}({\rm deck}(v))$ across $\gamma_-$.
Then RHP for ${\bf H}$ is given by
\begin{equation}
\begin{cases}
&{\bf H}(v)\big|_\Sigma=-{\bf H}(v)\big|_{\widetilde\Sigma}  \quad v\in\gamma_- \ ;
\vspace{0.1cm}
\\ &\displaystyle {\bf H}(v)=I+{\cal O}\left(\frac1v\right);
\vspace{0.1cm}
\\ &\displaystyle {\bf H}(v)=\left[\begin{array}{cc}0 & 1 \\-1 & 0\end{array}\right]+{\cal O}\left(v-\alpha\right)\quad v\to\alpha\ .
\end{cases}
\end{equation}
This gives a scalar RHP for each entry. It is sufficient to find the diagonal entries of ${\bf H}$ as the others are then given by \eqref{Hv}.   We obtain
\begin{equation}
{\bf H}(v)=\frac{\rho^{1/2}}{\sqrt{f'(v)}}\left[\begin{array}{cc} 1 &\displaystyle \frac1{(\kappa\rho)^{1/2}}\frac{\kappa}{v-\alpha} \\ \displaystyle\frac{1}{(\kappa\rho)^{1/2}}\frac{-\kappa}{v-\alpha} &\displaystyle 1\end{array}\right],
\end{equation}
where we take the branch cut on $\gamma_-$ such that 
\begin{equation}
\sqrt{f'(v)}=\frac{\rho^{1/2}}{v-\alpha}\sqrt{(v-v_\beta)(v-\c{v_\beta})} \approx \begin{cases}  \rho^{1/2}, \quad  &v\to\infty; \\ \kappa^{1/2}/(v-\alpha),\quad &v\to\alpha.
\end{cases}
\end{equation}
Using $f'(F(z))\,F'(z)=1$ and the formula for $D(z)$ at \eqref{SzegoD} we obtain 
\begin{equation}
\Phi(z)=(\rho\alpha)^{\frac{r}{2}\sigma_3} \sqrt{\rho F'(z)}\left[\begin{array}{cc} 1 &\displaystyle \frac1{(\kappa\rho)^{1/2}}\frac{\kappa}{F(z)-\alpha} \\ \displaystyle\frac{1}{(\kappa\rho)^{1/2}}\frac{-\kappa}{F(z)-\alpha} &\displaystyle 1\end{array}\right]
\left( \sqrt{\frac{\rho}{\alpha}}F(z)\right)^{r\sigma_3}.
\end{equation}
This proves the proposition.
\end{proof}
\paragraph{Local Parametrices:}
The procedure for constructing the local parametrices near the points $\beta, \c \beta$ is standard and the reader may refer to \cite{DKMVZ,deift_book}.   Here we only consider a neighborhood around $\c\beta$ and skip the similar construction around $\beta$.

Let $\D_{\c\beta}$ be a small disk centered at $\c\beta$ with a finite radius.
We define the local coordinate $\zeta$ by
\begin{equation}\label{Airy-coor}
\frac{4}{3}\zeta(z)^{3/2} := N\left(\phi (z)-\phi(\c\beta)\right)\ ,\quad z\in {\mathbb D}_{\bar\beta},
\end{equation}
such that $\zeta$ maps $\Gamma_{b\c\beta}\cap\D_{\c\beta}$ into $\R_+$.  Then $\zeta$ maps $\Gamma_\pm$ respectively into the rays $[0,e^{\pm 2\i\pi/3}\infty)$. 
This defines a conformal map $\zeta(z)$ from the fixed disk to a domain in the $\zeta$-plane that expands as $\mathcal O(N^{2/3})$.

We want to find the function ${\cal P}_{\c\beta}(z)$ in $\D_{\c\beta}$ (and, similarly, ${\cal P}_{\beta}(z)$ in $\D_{\beta}$) such that 
\begin{equation}
\Phi(z)z^{-(r/2)\sigma_3}{\cal P}_{\c\beta}(z) z^{(r/2)\sigma_3}\begin{cases} \text{satisfies the jump of $A$ at \eqref{RHPA};}
\\\text{converges to $\Phi(z)$ on the boundary of $\D_{\c\beta}$ as $N\to\infty$}.
\end{cases}
\end{equation}
This is satisfied by the following Riemann-Hilbert problem:
\begin{align}
&\big[\mathcal P_{\c\beta}(z)\big]_+= \pmtwo{0}{-1}{1}{0} \big[\mathcal P_{\c\beta}(z)\big]_- \pmtwo{0}{1}{-1}{0} &\mbox{on } \cal B\cap\D_{\c\beta}\ ,\\
&\big[\mathcal P_{\c\beta}(z)\big]_+= \big[\mathcal P_{\c\beta}(z)\big]_- \pmtwo{1}{0}{e^{\frac43\zeta(z)^{3/2}}}{1}  &\mbox{on } \Gamma_{\pm}\cap\D_{\c\beta}\ ,\\
&\big[\mathcal P_{\c\beta}(z)\big]_+= \big[\mathcal P_{\c\beta}(z)\big]_- \pmtwo{1}{e^{-\frac43\zeta(z)^{3/2}}}{0}{1} &\mbox{on } \Gamma_{b\c\beta}\cap\D_{\c\beta}\ ,\\
&\mathcal P_{\c\beta}(z) \sim  I+\mathcal O (\zeta(z)^{-3/2}) &\mbox { as }z \to \D_{\c\beta}\ ,\label{Airy-boundary}\\
&\mathcal P_{\c\beta}(z) \sim \mathcal O (\zeta(z)^{-1/4}) &\mbox { as } z \to \c\beta\ . 
\end{align}
The last condition is necessary later for ``the error matrix" to be analytic at $\c\beta$.

The solution to the above RHP is given in terms of the standard Airy parametrix $\mathcal A(\zeta)$ by
\begin{equation}\label{calPpre}
\mathcal P_{\c\beta}(z):=
{{\rm e}^{\frac {i\pi \sigma_3}4 }
\frac 1{\sqrt {2}} \left[
\begin{array}{cc}
1 &-1\\
1&1
\end{array}
\right] \zeta(z)^{\frac {\sigma_3} 4} 
}
\mathcal A(\zeta(z))\ ,
\end{equation}
where
\begin{equation}
\mathcal A(\zeta):= \sqrt{2\pi} {\rm e}^{-\frac {i \pi}4}
 \left\{
\begin{array}{lc}\pmtwo{y_0(\zeta)}{-y_2(\zeta)}{
y_0'(\zeta)}{- y_2'(\zeta)}{\rm e}^{\frac 23\zeta^{\frac 32}\sigma_3} & \arg \zeta\in (0,2\pi/3);\\[15pt]
\pmtwo{-y_1(\zeta)}{-y_2(\zeta)}{-y_1'(\zeta)}{- y_2'(\zeta)}{\rm e}^{\frac 23\zeta^{\frac 32}\sigma_3} & \arg \zeta\in (2\pi/3,\pi);\\[15pt]
\pmtwo{-y_2(\zeta)}{y_1(\zeta)}{-y_2'(\zeta)}{  y_1'(\zeta)}{\rm e}^{\frac 23\zeta^{\frac 32}\sigma_3} & \arg \zeta\in (\pi,5\pi/3);\\[15pt]
\pmtwo{y_0(\zeta) }{ y_1(\zeta)}{y_0'(\zeta)}{ y_1'(\zeta)}{\rm e}^{\frac 23\zeta^{\frac 32}\sigma_3} & \arg \zeta\in (5\pi/3,2\pi).
\end{array}
\right.
\end{equation}
We have used the standard Airy function to define:
\begin{equation}
 y_j(\zeta):= \omega^j {\rm Ai}(\omega^j \zeta),\ \ \ j=0,1,2,\ \ \omega = {\rm e}^{2i\pi/3}.\end{equation}

Using $\Phi$ \eqref{Phipre} and ${\cal P}_{\c\beta}$ \eqref{calPpre} we define $A^\infty$ which we show (in the following  {\bf Error analysis}) to be the leading approximation of the matrix $A(z)$ defined at \eqref{Adef}.
\begin{equation}\label{Ainf}
A^\infty(z):= \left\{ 
\begin{array}{lr}
\Phi(z),  &\!\!\!\!\!\!\!\!\!\!\!\!\!\!\!\!\!\!\!\!\!\!\!\!\!\!\!\!\!\!\mbox{for $z$ outside of the disks $\mathbb D_{\beta} \cup \mathbb D_{\c \beta}$};\\
\Phi(z) z^{-(r/2)\sigma_3}{\cal P}_{\beta}(z) z^{(r/2)\sigma_3}, &\mbox{ inside $\mathbb D_{\beta}$};\\
\Phi(z) z^{-(r/2)\sigma_3}{\cal P}_{\c\beta}(z) z^{(r/2)\sigma_3}, &\mbox{ inside $\mathbb D_{\c \beta}$}.
\end{array}
\right.
\end{equation}

\paragraph{Error analysis:} 
We define the error matrix ${\cal E}$ by
\begin{equation}\label{calEpre}
{\cal E}(z):=A^\infty(z) A^{-1}(z).
\end{equation}
As a consequence of the jumps of $A$ and $A^\infty$, $\mathcal E$ has jumps only on the boundaries of the disks, $\mathbb D_{\beta}$ and $\D_{\c\beta}$ and on the parts of $\Gamma_\pm$ and $\Gamma\setminus{\cal B}$ that lie outside of those disks.   

The jumps of ${\cal E}$ on the boundary of the lens, $\Gamma_\pm\setminus ({\mathbb D}_\beta\cup {\mathbb D}_{\overline\beta})$, is given by
\begin{align}\nonumber
[{\cal E}(z)]_+[({\cal E})^{-1}(z)]_-&=[A^\infty(z)]_+ [A^{-1}(z)]_+ [A(z)]_- [(A^\infty)^{-1}(z)]_-=\Phi(z) \left(I+{\cal O}( e^{-c N})\right)\Phi^{-1}(z)
\\&=I+{\cal O}( e^{-c N})\ , \qquad \text{for some $c>0$. } \end{align}
The jumps of ${\cal E}$ on the boundary of the disc, $\partial {\mathbb D_{\c\beta}} $, is given by (plus side is {\em inside} the disc)
\begin{equation}
[{\cal E}(z)]_+[({\cal E})^{-1}(z)]_-= \Phi(z) z^{-(r/2)\sigma_3}{\cal P}_{\c\beta}(z) z^{(r/2)\sigma_3} \Phi^{-1}(z)=I+{\cal O}(N^{-1})\ ,
\end{equation}
using the boundary behavior of \eqref{Airy-boundary} and \eqref{Airy-coor}. 
So ${\cal E}$ satisfies a small-norm RHP and we conclude that 
\begin{equation}\label{Strong}
A(z)=\left(I+{\cal O}\left(\frac 1N\right)\right)A^\infty(z)\ ,
\end{equation}
uniformly in (a compact set of) the (extended) complex plane.

\paragraph{Strong asymptotics:}
The asymptotic for the orthogonal polynomial can now be read off from our approximation \eqref{Ainf} by tracing back the transformations from $Y(z)$ (\ref{Ymatrix}) to $A(z)$ (\ref{Adef}).
Though the computation is straightforward we show a few steps for the interested readers.
\paragraph{Away from ${\cal B}$:}
The asymptotic for the orthogonal polynomial can now be read off from our approximation using that $
P_{n,N}(z)=[Y(z)]_{11}= A_{11} e^{tN g(z)}$ outside of  $\c{\Omega_\pm}$.
Using the asymptotic behavior \eqref{Strong}, we get the strong asymptotics in $\C\setminus(\D_\beta\cup\D_{\c\beta}\cup\c{\Omega_+}\cup\c{\Omega_-})$ as follows.
\begin{align}
P_{n,N}(z) &=\left[\left(I+{\cal O}(1/N)\right)A^\infty \right]_{11}e^{tN g(z)}
\\&=\left(1+{\cal O}(1/N)\right)\sqrt{\rho F'(z)}(\rho F(z))^r e^{tN g(z)}.
\end{align}
\paragraph{Near ${\cal B}$ but away from $\beta$ and $\c\beta$:}
Tracing back the transformations from $Y(z)$ (\ref{Ymatrix}) to $A(z)$ (\ref{Adef}) we find, respectively for $z\in\overline{\Omega_\pm}\setminus(\D_{\beta}\cup\D_{\c\beta})$, 
\begin{align}
P_n(z)  &= \left[ 
{\rm e}^{tN \frac {\ell}{2} \sigma_3} A(z)\,{\rm e}^{tN(g(z) - \frac \ell 2)\sigma_3}  \pmtwo{1}{0}{\pm 1/w_{n,N} }{1} 
\right]_{11} 
\\&=\left[ 
{\rm e}^{tN \frac {\ell}{2} \sigma_3} A(z)\,{\rm e}^{-\frac{N}{2}\phi(z)\sigma_3}  \pmtwo{1}{0}{\pm z^r }{1} {\rm e}^{\frac{N}{2}V(z)\sigma_3}
\right]_{11}.
\end{align}
Recalling that $A(z) = (I + \mathcal O(N^{-1}) ) A^\infty(z)$ and $A^\infty=\Phi$ in the region of our interest, we find respectively for $z\in\Omega_\pm\setminus(\D_{\beta}\cup\D_{\c\beta})$, 
\begin{align}\nonumber
P_n(z) &= {\rm e}^{\frac N2 (V(z) +t \ell) } 
\left( {\rm e}^{  -\frac {N} 2 \phi(z)  }  [\Phi(z)]_{11} (1 + \mathcal O(N^{-1}) \pm   {\rm e}^{  \frac {N} 2 \phi(z)  } z^{r}  [\Phi(z)]_{12}(1 + \mathcal O(N^{-1})  
\right)
\\
&= {\rm e}^{tN g(z) } 
\left(  [\Phi(z)]_{11}  \pm   {\rm e}^{  N \phi(z)  } z^{r}  [\Phi(z)]_{12}+ \mathcal O(N^{-1})\right)
\\
&= {\rm e}^{tN g(z) } \sqrt{\rho F'(z)}
\left(  (\rho F(z))^r  \pm   {\rm e}^{  N \phi(z)  } \left(\frac{\alpha\,z}{F(z)}\right)^{r} \frac{\sqrt{\kappa\rho}}{\rho F(z)-\rho\alpha}+ \mathcal O(N^{-1})\right). 
\end{align}
The leading term is analytic on ${\cal B}$; To check this, it is most convenient to use the expression in the first line using $[{\rm e}^{N\phi/2}]_+=[{\rm e}^{-N\phi/2}]_-$ on ${\cal B}$ and the jump relation, $\left([\Phi(z)]_{11},[\Phi(z)]_{12}\right)_+=\left(-z^r[\Phi(z)]_{12},z^{-r}[\Phi(z)]_{11}\right)_-$.

\paragraph{Near $\c\beta$ (and $\beta$):}
Inside $\D_{\c\beta}$ one can similarly obtain the strong asymptotics \eqref{PAiry} of $P_{n,N}$ from $A^\infty$ \eqref{Ainf} as in the previous calculations.  We leave the (straightforward) calculation to the interested reader, mentioning only (1) that the relation $\c v_\beta^2=\alpha\c\beta/\rho$ (which is obtained from \eqref{betapre}) can be useful and
(2) that we defined $C_{\c\beta}$ at \eqref{Cbeta} such that 
\begin{equation}
y(z)=2 C_{\c\beta} (z-\c\beta)^{1/2}\left(1+{\cal O}(z-\c\beta)\right).
\end{equation}
In \eqref{PAiry} of Theorem \ref{thm1} we express the asymptotics in the scaled coordinate $\zeta(z)\sim (C_{\c\beta}N)^{2/3}(z-\c\beta)$ taking only the leading linear term in the Taylor expansion of $\zeta(z)$ around $\c\beta$.

\section{The post-critical situation: \texorpdfstring{$t>t_c$}{ttc1}}\label{sec-post}

As $t$ increases, $\beta$ approaches $a$ as $\beta\sim a+a c/t+{\cal O}(1/t^2)$ and $b$ escapes to $\infty$. 
We define $A$ by the same equations \eqref{Adef} with the regions $\Omega_\pm$ as illustrated at the right side in Figure \ref{default}.
The RHP for $A$ is exactly the same as specified in eqs. \ref{RHPA} with the proviso that now $\Gamma \equiv {\cal B}$ and hence the fourth condition in \eqref{RHPA} is irrelevant.

In this case, the solution $\Phi$ to the RHP at \eqref{PsiRHP} is very simple.
\begin{equation}\label{Phisol}
\Phi(z)=\left\{\begin{array}{cc} z^{r\sigma_3} \qquad &z\in {\rm Ext}({\cal B})\ ,
\\ \pmtwo{0}{1}{-1}{0} &z\in {\rm Int}({\cal B})\ .
\end{array}\right.
\end{equation}

\paragraph{Local parametrix at the double point:}
We define the local coordinate $\zeta(z)$ near $\beta$ such that 
\begin{equation}
\zeta^2(z)=2N\begin{cases} \phi(z), & z\in {\mathbb D}_\beta\cap{\rm Int}({\cal B}),
\\-\phi(z), &z\in {\mathbb D}_\beta\cap{\rm Ext}({\cal B}),
\end{cases}
\end{equation}
where ${\mathbb D}_\beta$ is a sufficiently small but fixed disc around $z=\beta$ such that $\zeta(z)$ is one-to-one.
Under the mapping $\zeta$ the contour ${\cal B}$ maps to the straight rays $[0,e^{i\pi/4}\infty)\cup[0,e^{-i\pi/4}\infty)$,
and $\Gamma_+$ maps to the imaginary axis.   See Figure \ref{WeberFig} for the images of the contours under the conformal mapping $\zeta$.

We get the following expansion near $z=\beta$.


\begin{equation}\label{zbetazeta}
\frac{\zeta(z)}{N^{1/2}}=\frac{1}{\gamma_1}(z-\beta)\left(1+{\cal O}(z-\beta)\right),\quad \gamma_1:=\sqrt{\frac{\beta(\beta-a)}{a(b-\beta)}}.
\end{equation}

Inside ${\mathbb D}_\beta$ we want to find ${\cal P}$ such that $\Phi(z)\, z^{-\frac{r}{2}\sigma_3}{\cal P}(z)z^{\frac{r}{2}\sigma_3}$ satisfies the jump conditions of the RHP \eqref{RHPA} for $A$; This leads to
\begin{align}
& \left[{\cal P}(z)\right]_+=\left[{\cal P}(z)\right]_-\left[\begin{array}{cc}1&0\\{\rm e}^{\zeta(z)^2/2}&1\end{array}\right], &z\in\Gamma_+\cap\D_\beta ,
\\
& \left[{\cal P}(z)\right]_+=\left[\begin{array}{cc}0&-1\\1&0\end{array}\right]\left[{\cal P}(z)\right]_-\left[\begin{array}{cc}0&1\\-1&0\end{array}\right], &z\in{\cal B}\cap\D_\beta ,
\\
& {\cal P}(z) = I + \mathcal O\left(1/N\right) , &z\in\partial\D_\beta \ .
\end{align}

\begin{figure}[ht]
\begin{center}
\includegraphics[width=0.2\textwidth]{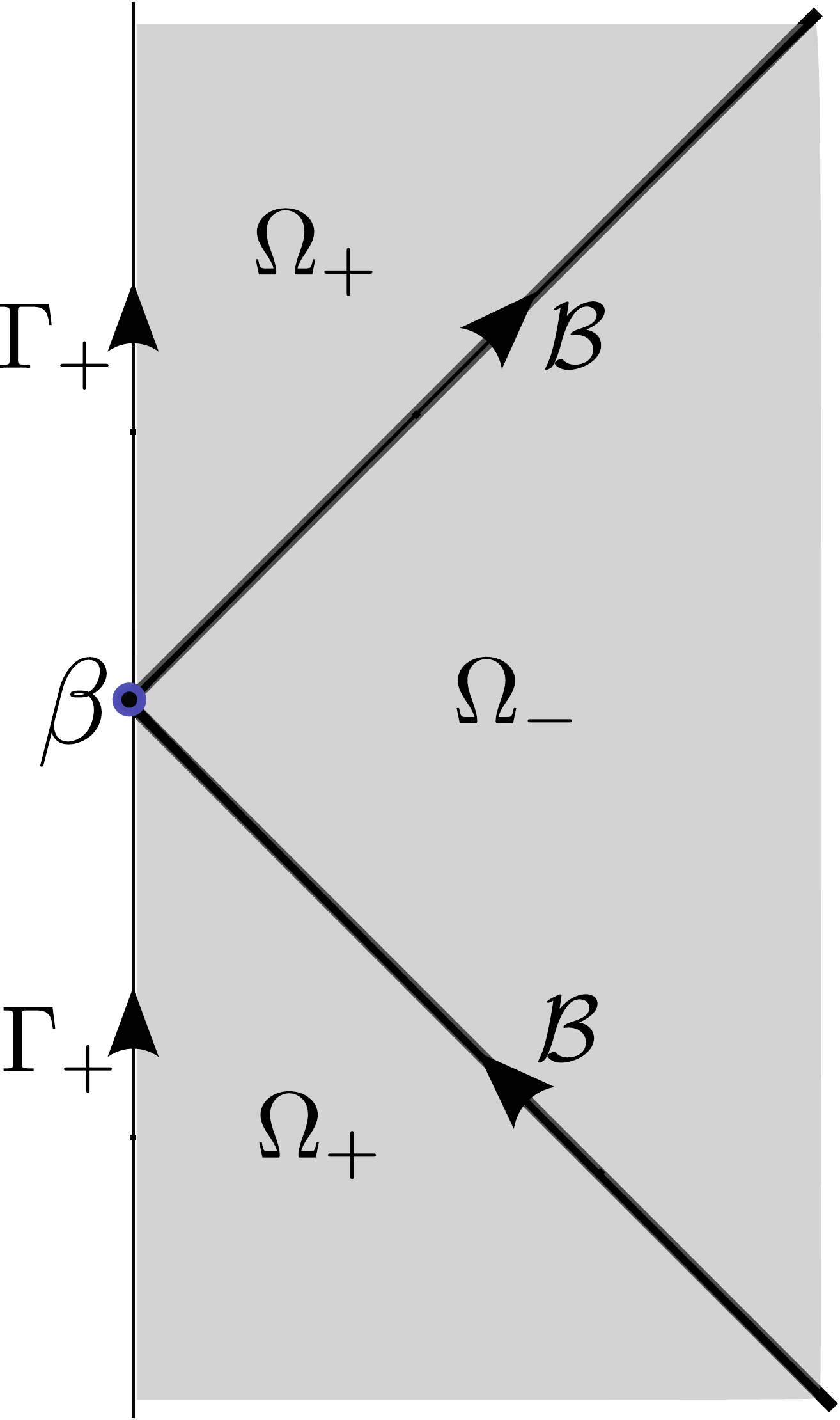}\qquad
\includegraphics[width=0.33\textwidth]{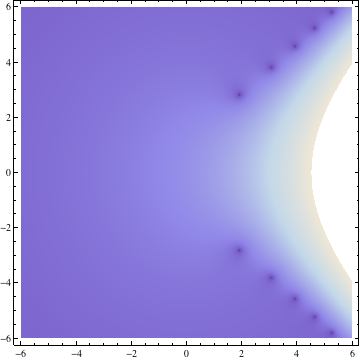}
\caption{(Left) Contour under $\zeta$-mapping; (Right) Plot of $\log|{\cal F}(\zeta)|$ s.t. the brightness increases with the value.  The dark dots are where the zeros of ${\cal F}$ are.  (${\cal F}$ is defined at \eqref{calF}.)}
\label{WeberFig}
\end{center}
 \end{figure}

A solution can be written by
\begin{equation}\label{calPpost}
{\cal P}(z):=\left\{\begin{array}{lr} H_0(z) {\bf F}(\zeta(z))\ ,\ \ \ &z\in {\rm Ext}({\cal B})\cap\D_\beta,\\
\left[\begin{array}{cc}0&-1\\1&0\end{array}\right]H_0(z) {\bf F}(\zeta(z)) \left[\begin{array}{cc}0&1\\-1&0\end{array}\right]  , \ \ \ &z\in {\rm Int}({\cal B})\cap\D_\beta,
\end{array}\right.
\end{equation}
where 
\begin{equation}\label{F}
{\bf F}(\zeta):= \left[
\begin{array}{cc}
1 & \displaystyle\frac{-1}{2i\pi} \int_{-i\infty}^{i\infty} \frac{{\rm e}^{s^2/2}}{(s-\zeta)}\, d s\\
0 &1
\end{array}
\right]
=\left(I+\frac{1}{\sqrt{2\pi}}  \frac{1}{\zeta}\left[\begin{array}{cc}0&1\\0&0\end{array}\right]+{\cal O}\left(\zeta^{-3}\right)\right) \ .
\end{equation}
For any choice of $H_{0}$ that is holomorphic and unimodular in $\D_\beta$, ${\cal P}$ defined by \eqref{calPpost} satisfies the jump conditions in $\D_\beta$.
We define the matrix $H_0$ by
\begin{equation}\label{H0}
H_0(z):=\Xi(z)\left(I-\frac{1}{\sqrt{2\pi}}  \frac{1}{\zeta(z)}\left[\begin{array}{cc}0&1\\0&0\end{array}\right]\right),
\quad 
\Xi(z):=I+\frac{1}{\sqrt{2\pi N}}  \frac{\gamma_1 \beta^{r} z^{-r}}{z-\beta}\left[\begin{array}{cc}0&1\\0&0\end{array}\right],
\end{equation}
such that it satisfies, for $z\in\partial\D_\beta$ as $N\to\infty$,
\begin{equation}\label{whyH0}
H_0(z){\bf F}(\zeta(z))\,\Xi(z)^{-1}=I+{\cal O}\left(\zeta(z)^{-3}\right).
\end{equation}
The logic of how we choose $\Xi(z)$ is as follows:
 we require that i) $\Xi (z)$ is meromorphic in the extended complex plane with poles exactly at $\beta$ and $\infty$; ii) the pole singularity of $\Xi(z)$ at $\beta$ matches the pole singularity of 
 \begin{equation}
I+\frac{1}{\sqrt{2\pi}}  \frac{1}{\zeta(z)}\left[\begin{array}{cc}0&1\\0&0\end{array}\right];
\end{equation}
and iii) $\Phi(z) z^{-(r/2)\sigma_3}\cdot\Xi (z)\cdot z^{(r/2)\sigma_3}$ has the same growth behavior as that of $A(z)$ as $z\to\infty$. 

This construction is based on the algorithm called ``partial Schlesinger transform" introduced in \cite{Colonization} to compute the higher order corrections to the Riemann-Hilbert asymptotic analysis.  
In the critical case, we repeat a similar construction. 

Now we define the strong asymptotics of $A(z)$ by
\begin{align}\nonumber
A^\infty(z)&:=\begin{cases} 
\displaystyle\Phi(z) \, z^{-\frac{r}{2}\sigma_3}\left(I+\frac{1}{\sqrt{2\pi N}}  \frac{\gamma_1\beta^r z^{-r}}{z-\beta}\left[\begin{array}{cc}0&1\\0&0\end{array}\right]\right)z^{\frac{r}{2}\sigma_3}
\vspace{0.1cm}
\\
\displaystyle\Phi(z) \left[\begin{array}{cc}0 & -z^{-r} \\z^r & 0\end{array}\right] z^{-\frac{r}{2}\sigma_3}\left(I+\frac{1}{\sqrt{2\pi N}}  \frac{\gamma_1\beta^r z^{-r}}{z-\beta}\left[\begin{array}{cc}0&1\\0&0\end{array}\right]\right)z^{\frac{r}{2}\sigma_3}
\left[\begin{array}{cc}0 & z^{-r} \\-z^r & 0\end{array}\right]
\vspace{0.1cm}
\\\displaystyle \Phi(z) \,z^{-\frac{r}{2}\sigma_3}{\cal P}(z)\,z^{\frac{r}{2}\sigma_3}
\end{cases}
\\
&=\begin{cases} 
\displaystyle\left(I +\frac{1}{\sqrt{2\pi N}}  \frac{\gamma_1\beta^r }{z-\beta}\left[\begin{array}{cc}0&1\\0&0\end{array}\right]\right) \left[\begin{array}{cc}z^r & 0 \\0 & z^{-r}\end{array}\right], &z\in {\rm Ext}({\cal B})\setminus\D_\beta,
\vspace{0.1cm}
\\
\displaystyle \left( I+\frac{1}{\sqrt{2\pi N}}  \frac{\gamma_1 \beta^{r}}{z-\beta}\left[\begin{array}{cc}0&1\\0&0\end{array}\right]\right)\left[\begin{array}{cc}0 & 1 \\-1 & 0\end{array}\right]
, &z\in {\rm Int}({\cal B})\setminus\D_\beta,
\vspace{0.1cm}
\\\displaystyle \Phi(z) \,z^{-\frac{r}{2}\sigma_3}{\cal P}(z)\,z^{\frac{r}{2}\sigma_3}, &z\in\D_\beta.
\end{cases}\label{postA}
\end{align}
\begin{remark}\label{whynot}
Note that, unlike the pre-critical case, the outer parametrix is modified from $\Phi$ by a left factor.  As we will see soon, if we used the unmodified outer parametrix $\Phi$ in the definition of $A^\infty$ we would have gotten the same strong asymptotics for $P_{n,N}$ in ${\rm Ext}({\cal B})$ with a larger error term, namely $\mathcal O(1/\sqrt{N})$.   This error term turns out to be too large for the analysis in ${\rm Int}({\cal B})$ because the leading  behavior of $P_{n,N}$ there is only of order $\mathcal O(1/\sqrt{N})$.
\end{remark}

\paragraph{Error analysis:}
We define the error matrix by
\begin{equation}\label{calEpost}
{\cal E}(z):=A^\infty(z) A^{-1}(z)\ .
\end{equation}
One can show that the jump of ${\cal E}$ is bounded such that the solution satisfies ${\cal E}(z)=I+{\cal O}(N^{-3/2})$.
Let us perform the computation on the contour, $\partial\D_\beta\cap{\rm Ext}({\cal B})$.  
Below, $+$ side is to the side of $\D_\beta$. 
\begin{align}
{\cal E}(z)_+({\cal E}(z)_-)^{-1}&=A^\infty(z)_+(A^\infty(z)_-)^{-1}
\nonumber
\\&=\Phi(z)\,z^{-\frac r2\sigma_3}{\cal P}(z)\,z^{\frac r2\sigma_3}
z^{-\frac r2\sigma_3} 
\left(I-\frac{1}{\sqrt{2\pi N}}  \frac{\gamma_1\beta^r z^{-r}}{z-\beta}\left[\begin{array}{cc}\!\!0\!\!&1\!\!\\\!\!0\!\!&0\!\!\end{array}\right]\right)
z^{\frac r2\sigma_3} \Phi(z)^{-1}
\\(\text{using \eqref{calPpost}})~~&=
z^{\frac r2\sigma_3}H_0(z) {\bf F}(\zeta(z))
\left(I-\frac{1}{\sqrt{2\pi N}}  \frac{\gamma_1\beta^r z^{-r}}{z-\beta}\left[\begin{array}{cc}\!\!0\!\!&1\!\!\\\!\!0\!\!&0\!\!\end{array}\right]\right)
z^{-\frac r2\sigma_3} 
\\(\text{using \eqref{whyH0}})~~&=
z^{\frac r2\sigma_3}
\left(I+{\cal O}\left(N^{-3/2}\right)\right)
z^{-\frac r2\sigma_3} 
\\&=
I+{\cal O}\left(N^{-3/2}\right),\quad \text{ for $z\in\partial\D_\beta\cap{\rm Ext}({\cal B})$}.
\end{align}
A similar calculation gives the same error bound on ${\rm Int}({\cal B})\cap\partial\D_c$.  The jump matrices of ${\cal E}$ on the other contours are exponentially close to the identity in large $N$ and, therefore, ${\cal E}$ satisfies a small-norm RHP and we conclude that 
\begin{equation}\label{AAinfty}
A(z)=\left(I+{\cal O}\left(N^{-3/2}\right)\right)A^\infty(z)\ ,
\end{equation}
uniformly in the whole complex plane.

As in the pre-critical case, the analysis of the asymptotic for the orthogonal polynomial can now be read off from our approximation in \eqref{postA}.     

\paragraph{Away from ${\cal B}$:} The analysis of the asymptotic for the orthogonal polynomial can now be read off our approximation using that $
P_n(z)=[Y(z)]_{11}= A_{11} e^{tN g(z)}$ outside $\Omega_\pm$.
Using the asymptotic behavior \eqref{AAinfty}, we get the strong asymptotics {\em outside the disk ${\mathbb D}_\beta$ and $\Omega_\pm$} as follows (see also (\ref{gdough}) and recall that $tN = n-r$).
\begin{align}
\label{strongdough}
P_n(z) &=\left[\left(I+{\cal O}\left(N^{-3/2}\right)\right)A^\infty(z) \right]_{11}e^{tN g(z)}
\\
&=
\begin{cases}
\displaystyle  \left(1+{\cal O}\left(N^{-3/2}\right)\right) z^re^{tN g(z)} ,
&z\in{\rm Ext}({\cal B})\ ,
\\
\displaystyle
\left( \frac{\gamma_1\beta^r}{\sqrt{2\pi N}}\frac{1}{\beta-z} +{\cal O}\left(N^{-\frac 3 2}\right)\right)\e^{tNg(z)} 
&z\in{\rm Int}({\cal B})\ ,
\end{cases}
\end{align}
where $\ell$ is given in (\ref{Robindough}).  This holds uniformly over compact subsets of $\C\setminus{\cal B}$ because $\Gamma_\pm$ can be moved as close to ${\cal B}$ as one wishes.

\paragraph{Near ${\cal B}$ but away from $\beta$:}
In order to analyze the asymptotic {\em on} ${\cal B}$ we need to trace back the transformation that $Y$ was subject to within the region $\Omega_\pm$.  We thus obtain 

\begin{figure}[ht]
\begin{center}
\includegraphics[width=0.5\textwidth]{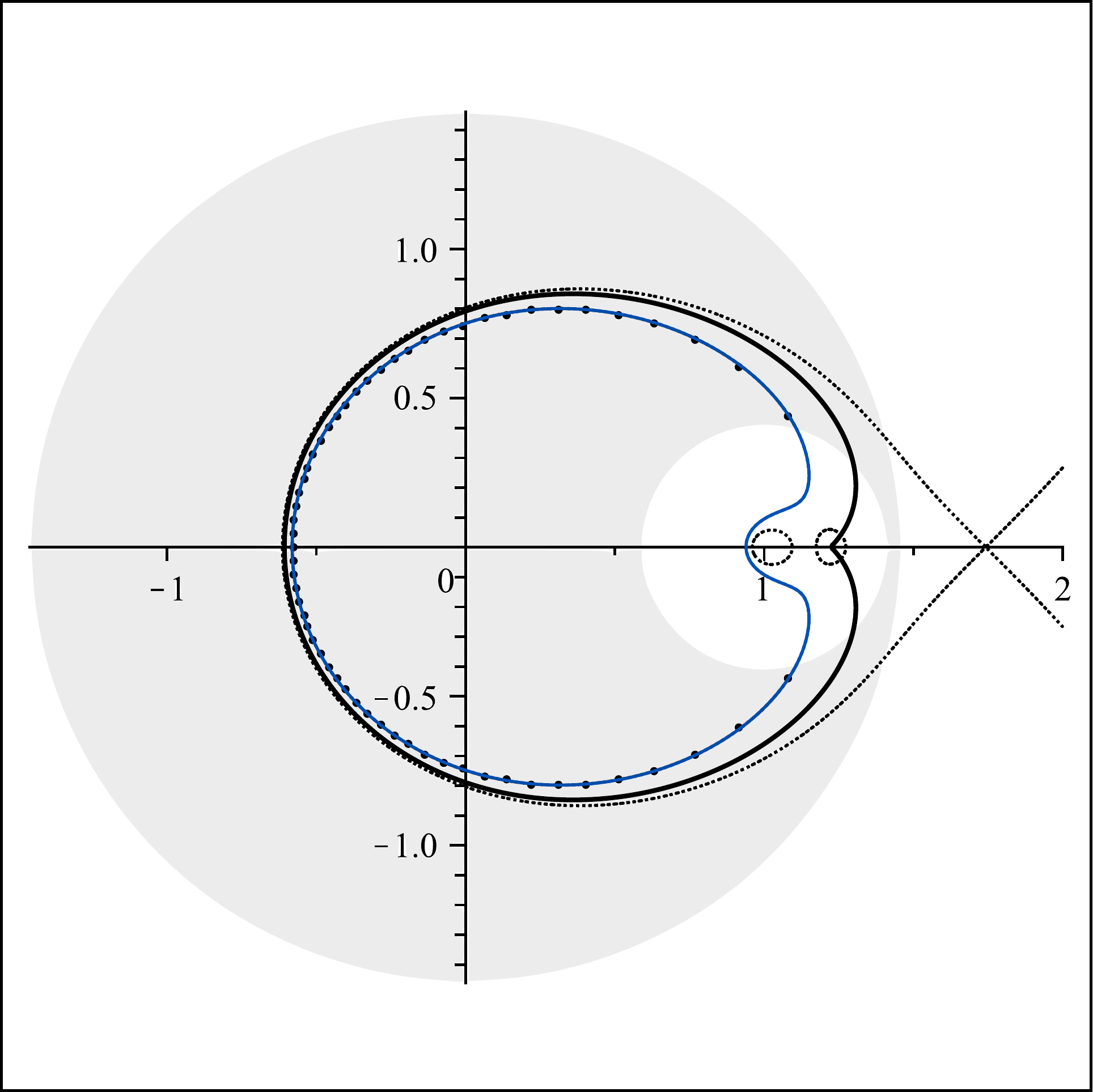}
\end{center}
\caption{The zeros of $P_{59}$ with $N=30$, $a=1$, $c=\frac 1 6$, $r=0$. Here $t = \frac n N = 1.9\c 6$ and $t_c \approx 1.8164$. The thick line is ${\cal B}$, where the zeros converge. The solid (blue) line and the dotted lines are where \eqref{zeroas} holds.  As can be surmised from the figure, the dotted lines do not approach ${\cal B}\setminus\{\beta\}$ in the limit $N\to\infty$; only the solid line approaches ${\cal B}\setminus\{\beta\}$ and the roots seem to line up along that solid line. The shaded region indicates $K$. }
\label{fig10}
\end{figure}
\begin{figure}[ht]
\begin{center}
\begin{tabular}{ccc}
\includegraphics[width=0.31\textwidth]{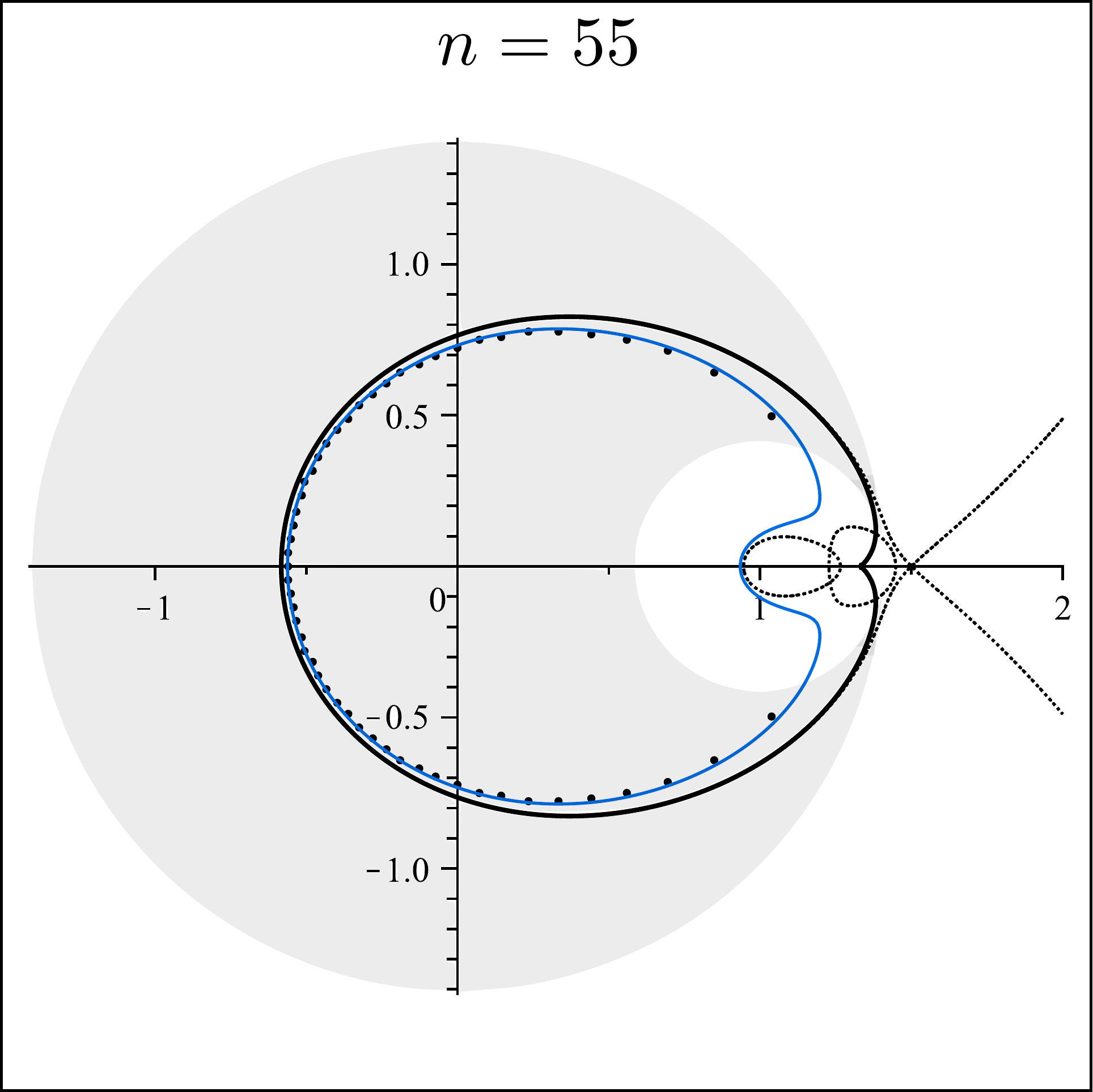} &
\includegraphics[width=0.31\textwidth]{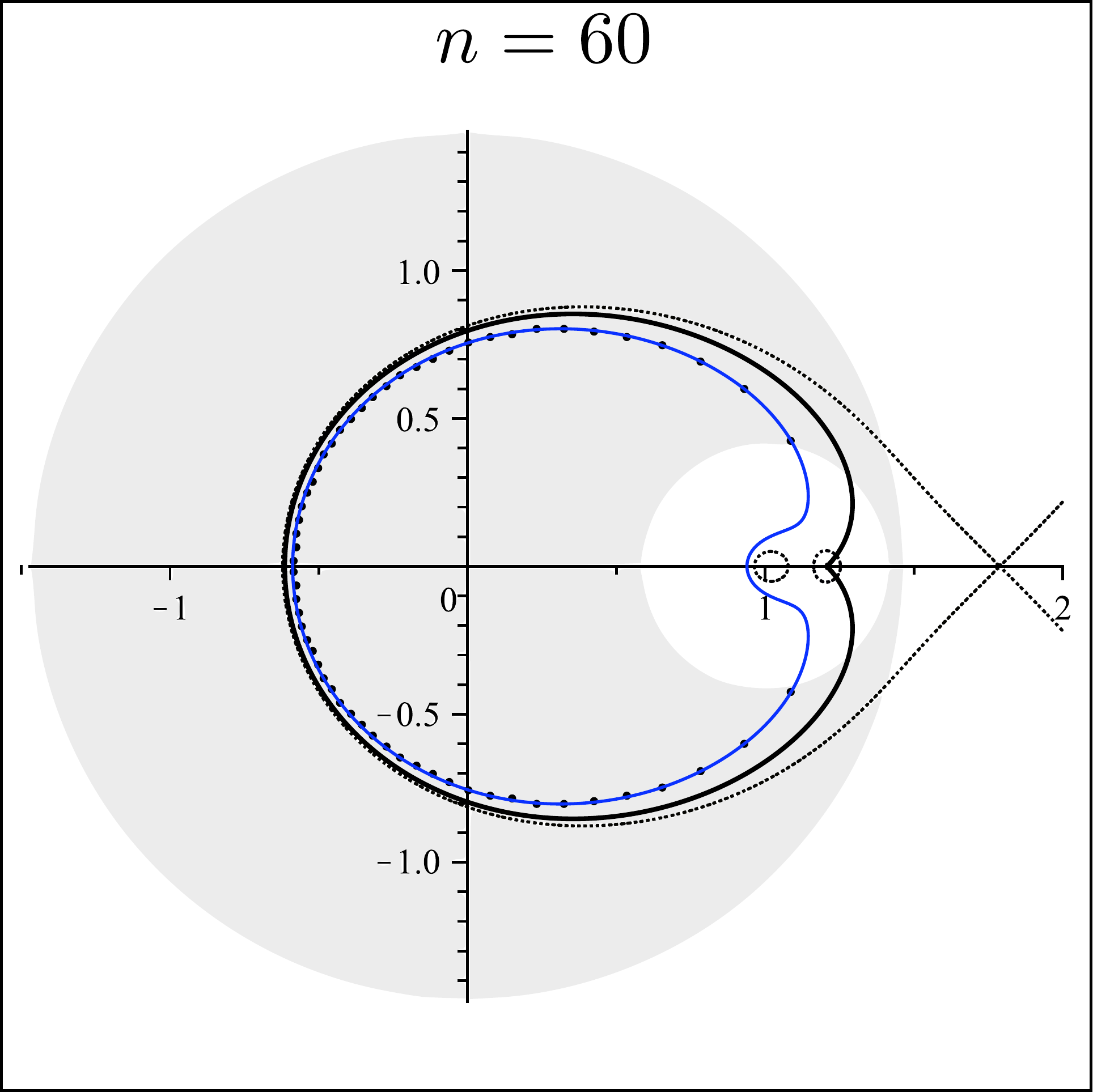} &
\includegraphics[width=0.31\textwidth]{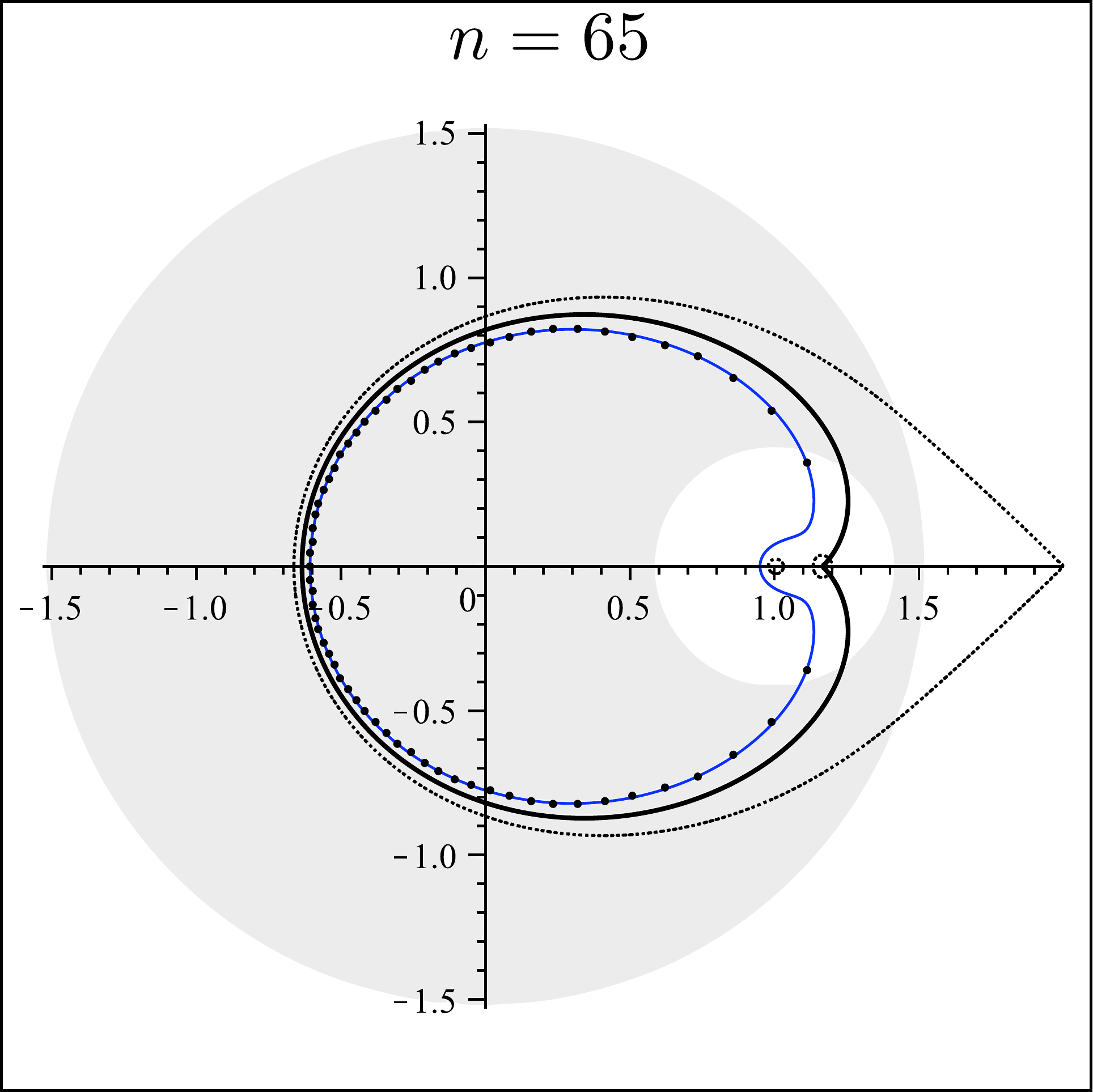} \\
\includegraphics[width=0.31\textwidth]{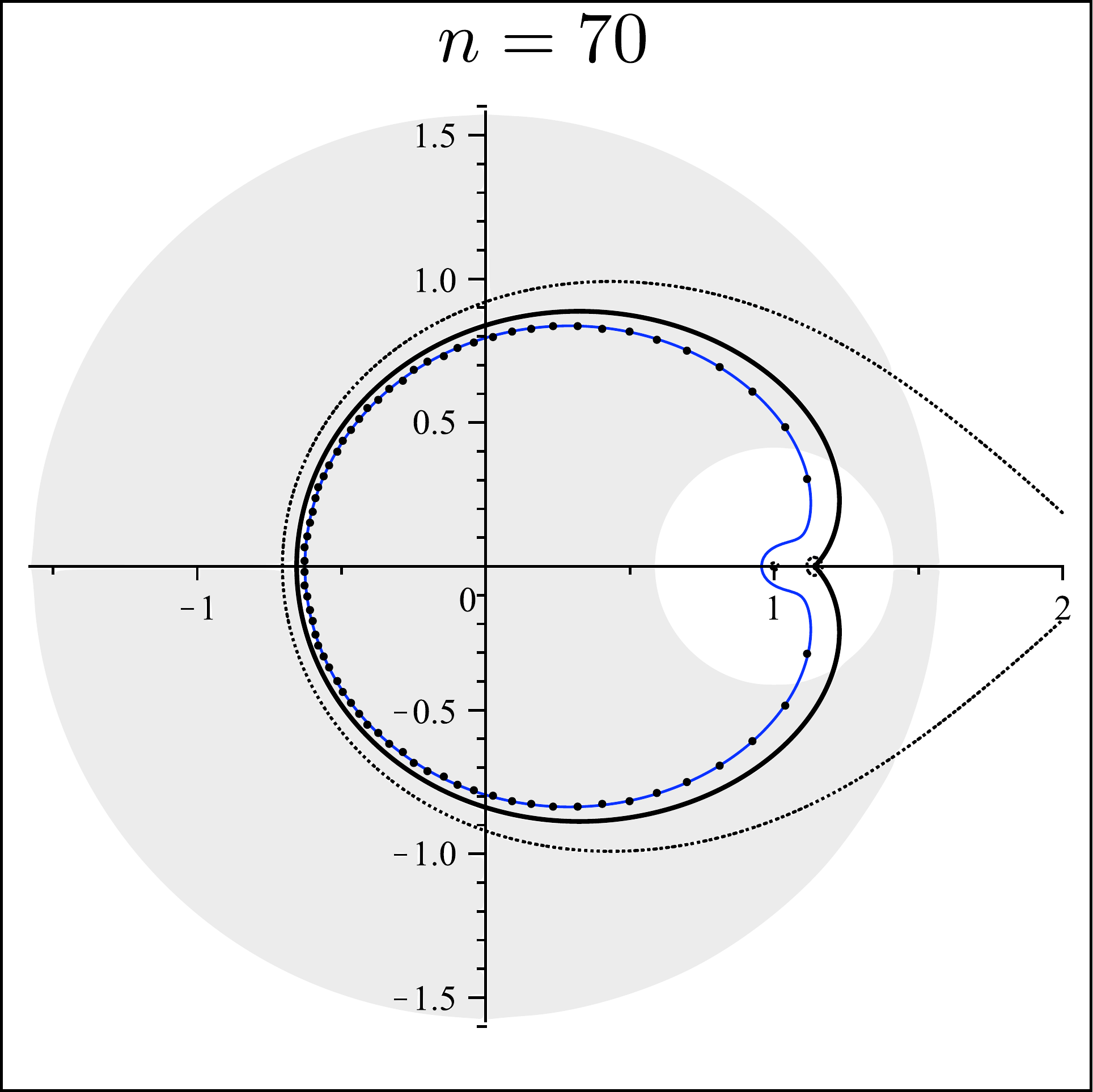} &
\includegraphics[width=0.31\textwidth]{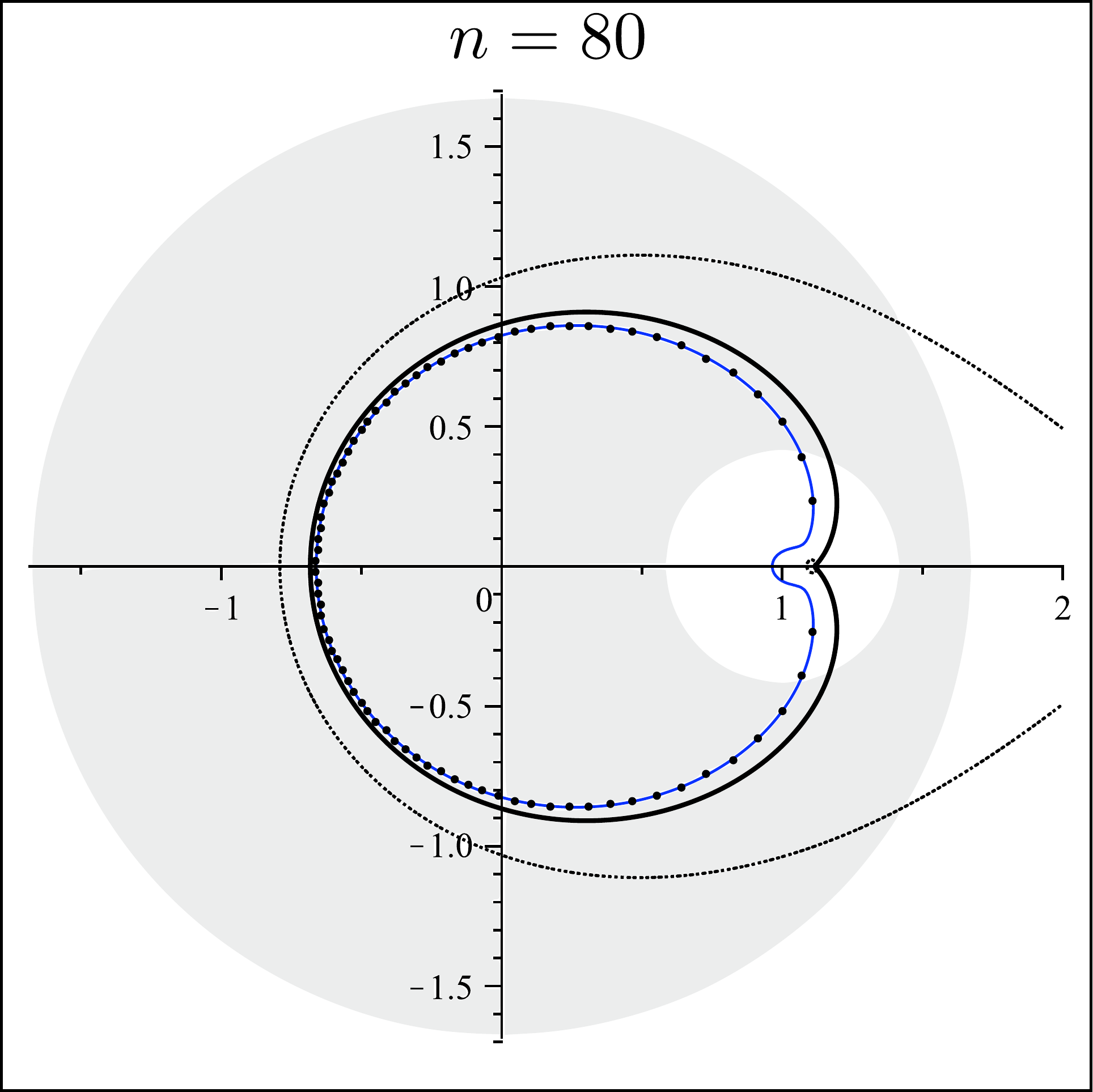} &
\includegraphics[width=0.31\textwidth]{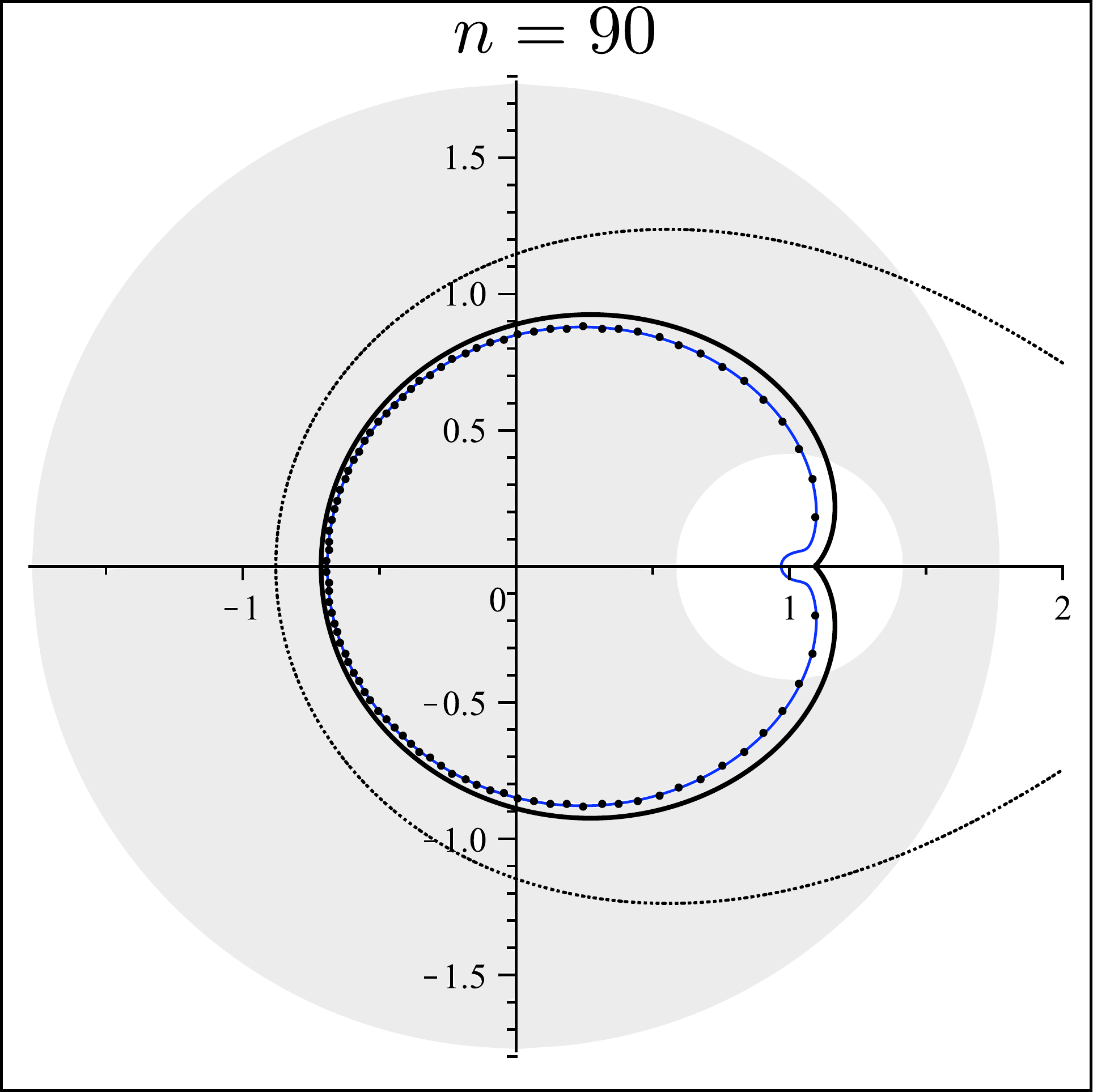} 
\end{tabular}
\caption{The zeros of $P_{n}$ with $N=30$, $a=1$, $c=\frac 1 6$, $r=0$, $n=55,60,65,\cdots$.   See the caption in Figure \ref{fig10} for more details.}
\label{fig11}
\end{center}
\end{figure}
\begin{equation}
P_n(z)  = \left[ 
{\rm e}^{\frac {tN \ell}{2} \sigma_3} A(z){\rm e}^{tN(g(z) - \frac \ell 2)\sigma_3}  \pmtwo{1}{0}{\pm \frac 1 {w_{n,N} } }{1} 
\right]_{11} ,\quad z\in\Omega_\pm,
\end{equation}
where $A = (I + \mathcal O(N^{-3/2})) A^\infty$. A straightforward computation yields 
\begin{align}
P_n(z)&=\begin{cases}\displaystyle
z^r {\rm e}^{t N g(z)}\left(1-\frac{\gamma_1\beta^r}{\sqrt{2\pi N}}\frac{{\rm e}^{N\phi(z)}}{z^r (z-\beta)}+{\cal O}\left(N^{-3/2}\right)\right)  ,\quad &z\in \Omega_-\setminus\D_\beta,
\\ \displaystyle
z^r {\rm e}^{t N g(z)}\left({\rm e}^{N\phi(z)}-\frac{\gamma_1\beta^r}{\sqrt{2\pi N}}\frac{1}{z^r (z-\beta)}+{\cal O}\left(N^{-3/2}\right)\right), & z\in \Omega_+\setminus\D_\beta.
\end{cases}\label{525}
\end{align}
The leading term is analytic on ${\cal B}$ because  $[t\,g(z)]_\pm=[t\, g(z)+\phi(z)]_\mp$. 

\paragraph{Location of zeros:} Here we explain the idea of the proof of Proposition \ref{prop-zero} concentrating on the post-critical case.
In order to locate the zeros of $P_n$ we observe that they may only lie (asymptotically)  where the following terms appearing above cancel each other:
\begin{equation}
{\rm e}^{\pm N\phi(z)}-\frac{\gamma_1\beta^r}{\sqrt{2\pi N}}\frac{1}{z^r (z-\beta)},\quad z\in\Omega_\pm.
\end{equation}
Recall now that $\gamma = \mathcal O(\sqrt{N}) $ and that $\Re \phi(z)<0$ on both sides of ${\cal B}$; this means that the zeros of $P_n$ may possibly lie only within $\Omega_+$ (i.e. slightly on the inside of ${\cal B}$) on a region that lies within $\Re \phi = \mathcal O( \ln N/N)$.
Therefore the location of the zeros of $P_n(z)$ converges to the curve given implicitly by
\begin{equation}\label{zeroas}
{\mathrm {Re}\,} \phi(z) = \frac 1 N \ln \left| \frac {\gamma_1\beta^r}{\sqrt{2\pi N}z^r(z-\beta)}\right|=-\frac {\ln N}{2N} +  \frac 1 N \ln \left| \frac {\sqrt{\beta(\beta-a)}\beta^r}{\sqrt{2\pi a(b-\beta)}z^r(z-\beta)}\right|.
\end{equation}
Their (integrated) density is asymptotically determined by the difference of the imaginary parts of the two sides in (\ref{zeroas}) along the said curve. We remark that the above carries information on the location and density of zeros only on parts of $\cal B$ that are away from the point $\beta$.

To prove the statement in Proposition \ref{prop-zero} one starts by noticing 
that there is asymptotically no zero outside a strip of width ${\cal O}(1/N)$ around the connected component of the set \eqref{zeroas} that tends to $\mathcal B$ (at the rate $\mathcal O(\ln N/N)$). This is seen by noticing that, where \eqref{zeroas} is satisfied, the two leading terms in \eqref{525} for $z\in\Omega_+$  have the same order $\mathcal O(N^{-1/2})$ and the error is of order $\mathcal O(N^{-3/2})$.  
Next one uses the asymptotics of the polynomials together with some form of the argument principle,
to show that the zeros actually lie in a strip of size ${\cal O}(1/N)$ around the curve where \eqref{zeroas} is satisfied.
The Figures \ref{fig10} and \ref{fig11} clearly demonstrate such convergence.

To conclude, we mention the convergence of the counting measure of the zeros, to the measure whose logarithmic potential is given by ${\rm Re}\,g$. We claim that, if $\mu_0$ is the probability measure supported on ${\cal B}$ such that
\begin{equation}
 	g(z)=\int_{\cal B}\log(z-w)\,\d\mu_0(w),
 \end{equation} 
 we have the weak convergence: $\mu_{n,N}\to\mu_0$ where $\mu_{n,N}$ is the normalized counting measure of the zeros of $P_{n,N}(z)$.  This can be proven by observing, from the strong asymptotics of $P_{n,N}$, that
\begin{equation}
	\lim_{n\to\infty}\frac{1}{n}\log P_{n,N}(z)= g(z),
\end{equation}
uniformly over compact subsets of $\C\setminus{\cal B}$, see \cite{saff_totik_book} (Chapter III) and \cite{Saff91} (Theorem 2.3).

\section{The critical case: \texorpdfstring{$t\sim t_c$}{ttc2}}\label{sec-crit}

For the critical case, $\Re g$ is the logarithmic potential for a {\em signed} measure and, therefore, $\Re\phi$ may not satisfy the inequality $\Re\phi<0$ near ${\cal B}$.  This does not cause any problem for us, because the region where the inequality is violated is contained inside the domain of the local parametrix where we do not require the inequality. 

Following Definition \ref{def-Omega} we define the (lens-opened) region $\Omega\pm$ by the region enclosed by ${\cal B}$ and the steepest descent lines of $\Re\phi$ from $\beta$.  We define $A$ by the same formula as in \eqref{Adef} such that it satisfies the Riemann-Hilbert problem \eqref{RHPA} except that there is no $\Gamma\setminus{\cal B}$.

Bounded away from $b_c=a+\sqrt c$ the jump matrices of $A$ are all uniformly and exponentially close to the identity jump, save for the jump on ${\cal B}$ (\ref{RHPA}). We define $\Phi$ by \eqref{Phisol} to solve \eqref{PsiRHP} as in the post-critical case.

\paragraph{Local conformal coordinate:} 

We remind the reader that $\phi$ has been defined in Definition \ref{def-phi} as a carefully defined anti-derivative of $y(z)$ (which is defined in \eqref{ypost}), and possesses a jump discontinuity across $\cal{B}$.  As described in Section \ref{sec-ypost}, the two critical points of $\phi$ ($\beta$ and $b$) approach $b_c=a+\sqrt c$ as $N\to\infty$ for the critical case.  Below we intend to find a conformal map $\cal{W}$ which captures the important local behavior of $\phi$ in a vicinity of $b$ and $\beta$.  This is accomplished by defining a cubic polynomial 
\begin{equation}
P_{3}({\cal W}) =\frac 13 ({\cal W} - {\cal W}_{\beta})^{2} ( {\cal W} + 2 {\cal W}_{\beta}) = \frac13{\cal W}(z)^3-{\cal W}_{\beta}^2 {\cal W}(z)+\frac{2}{3}{\cal W}_{\beta}^3
\end{equation}
where
\begin{equation}
{\cal W}_{\beta}:=\frac{3^{1/3}}{2^{5/3}} |\phi(b)|^{1/3}\times \begin{cases} 1 , & t\leq t_c,
\\ \i , & t>t_c.
\end{cases}
\end{equation}
This cubic polynomial satisfies the following:
\begin{equation}
P_{3}({\cal W}_{\beta}) = P_{3}'({\cal W}_{\beta})= P_{3}'(-{\cal W}_{\beta}) = 0
~~\text{ and }~~ - 8 i P_{3}(-{\cal W}_{\beta}) = [\phi(b)]_{\text{Ext}}
\end{equation}
\begin{lemma} For $t$ near $t_c$ there exists a fixed disk $\D_c$ centered at $b_c$ and the conformal map ${\cal W}:\D_c\to {\cal W}(\D_c)$ (for each $t$) that satisfies
\begin{equation}\label{Nphi0}
 - 8\ii P_3({\cal W}(z))
=\begin{cases} \phi(z) ,\ & z\in \D_c\cap{\rm Ext}({\cal B}); \\
-\phi(z),\ & z\in \D_c\cap{\rm Int}({\cal B}).\end{cases}
\end{equation}
\end{lemma}
The lemma is stated in Theorem 1 in \cite{Ursell} (with different notation: here we only note that the right hand side of \eqref{Nphi0} is an analytic function in $\mathbb D_c$ because $\phi(z)$ has a jump discontinuity across $\mathcal B$ whereby $\phi_+ = -\phi_-$).   We also refer to Section 4.3 in \cite{MillerBuckingham} for another presentation.

We define the scaled local coordinate $\xi(z)$, $\zeta_\beta$ and the scaled time coordinate $s$ by
\begin{equation}
\xi(z):=N^{1/3} {\cal W}(z),\quad \zeta_\beta:=N^{1/3}{\cal W}_\beta, \quad s:=-4 N^{2/3} {\cal W}_\beta^2=\frac{c^{1/6}N^{2/3}}{a^{1/3}b_c^{2/3}} (t-t_c)(1+{\cal O}(t-t_c)). 
\end{equation}
such that we have
\begin{equation}\label{Nphi}
 - 2\ii\left(\frac43{\xi(z)}^3+s\,\xi(z)+\frac{8}{3}\zeta_\beta^3\right)=\begin{cases} N\,\phi(z) ,&z\in{\mathbb D}_c\cap {\rm Ext}({\cal B}),
\\ -N\,\phi(z) ,& z\in{\mathbb D}_c\cap {\rm Int}({\cal B}).
\end{cases}
\end{equation}
We will consider the vicinity $t-t_c\sim {\cal O}(N^{-2/3})$ such that $s$ remains finite in the limit of large $N$.

We define $b^*_c$ to be the root of ${\cal W}(z)=0$ and a calculation (using $\phi(b)+\phi(\beta)=2\phi(b^*_c)$) shows
\begin{equation}\label{bstar}
b^*_c=b_c+(t-t_c)/(4 b_c)+{\cal O}((t-t_c)^2).
\end{equation}
Taking the derivative of \eqref{Nphi} we have $\xi'(z)|_{z=b^*_c}=(-2is)^{-1} N [y(b^*_c)]_\text{Ext({\cal B})}$, which gives
\begin{equation}\label{gammastar}
\gamma^*_c:=\frac{N^{1/3}}{i \xi'(b^*_c)}=
\gamma_c+{\cal O}\left(\sqrt{t-t_c}\right),
\quad
\gamma_{c} : = 2\frac{(b_c)^{1/3}c^{1/6}}{a^{1/3}} .
\end{equation}
Later we need that, for $z\in\overline{\D_c}$ and for a fixed $N$ and $s$, we have
\begin{equation}
\frac{\xi(z)}{N^{1/3}} = \frac{-i}{\gamma^*_{c}} 
 (z-b^*_c)\left(1+{\cal O}\left(z-b^*_c)\right)\right).
 \label{constconf}
\end{equation}
See Figure \ref{xiplane} for how the various contours map under the mapping $\xi$.

\begin{figure}[ht]
\begin{center}
\includegraphics[width=0.4\textwidth]{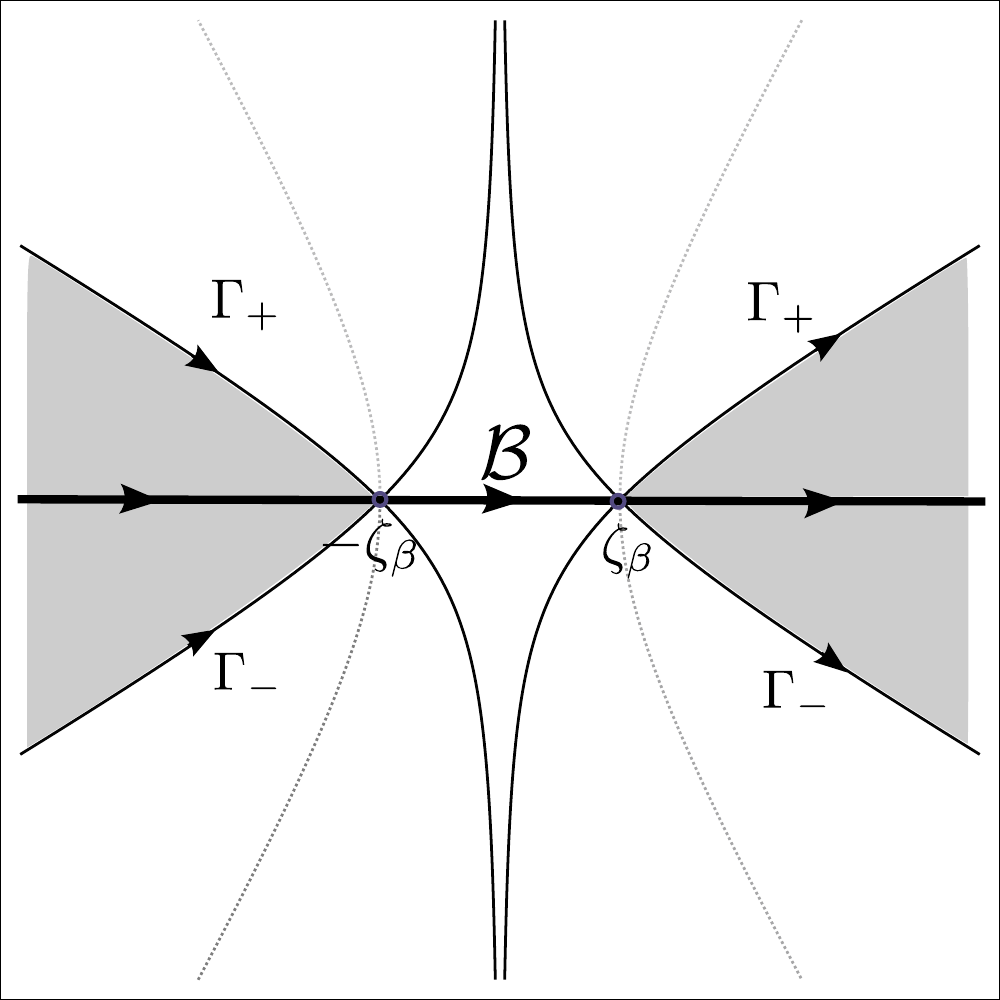}
\includegraphics[width=0.4\textwidth]{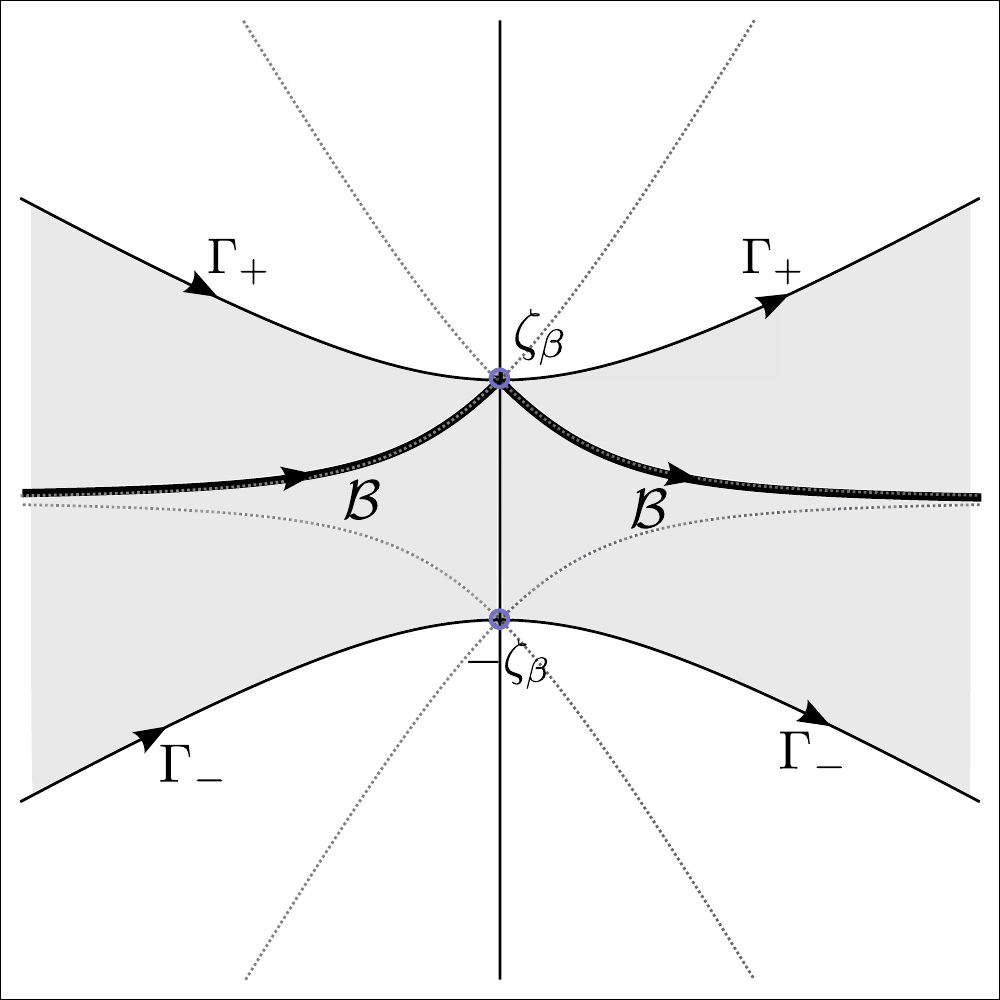}
\caption{The contours in the local scaling $\zeta$-coordinate for $t<t_c$ (left) and for $t>t_c$ (right) under the mapping $\xi$.  The branch cut (thick lines); the level lines of $\Re\left(\frac43\zeta^3+s\,\zeta+\frac{8}{3}\zeta_\beta^3\right)$ from $\zeta_\beta$ and $\c\zeta_\beta$ (dotted lines); the steepest ascent/descent lines (thin lines).   }
\label{xiplane}
\end{center}
\end{figure}

\paragraph{Local RHP:}
We require that the local parametrix solve the RHP for $A$ \eqref{RHPA} exactly on $\D_c$; i.e. we seek a matrix ${\cal P}(z)$ such that 
\begin{itemize}
\item $\Phi(z) z^{-\frac r 2\sigma_3} {\cal P}(z) z^{\frac r 2 \sigma_3}$ solves the jump conditions of $A$ within the disk $\mathbb D_c$;
\item $\Phi(z)z^{-\frac r 2\sigma_3} {\cal P}(z) z^{\frac r 2 \sigma_3}$ is bounded in $\mathbb D_c$;
\end {itemize}
These conditions are satisfied if ${\cal P}$ satisfies
\begin{equation}\label{errorwp}
\begin{cases}
{\cal P}(z)_+ = \pmtwo{0}{-1}{1}{0} {\cal P}(z)_- \pmtwo{0}{1}{-1}{0},  &z \in {\cal B};\vspace{0.1cm}\\
{\cal P}(z)_+ = {\cal P}(z)_- \pmtwo{1}{0}{{\rm e}^{2\ii\left(\frac43{\xi(z)}^3+s\,\xi(z)+\frac{8}{3}\zeta_\beta^3\right)}}{1},  &z \in \Gamma_+;\vspace{0.1cm}\\
{\cal P}(z)_+ = {\cal P}(z)_- \pmtwo{1}{0}{{\rm e}^{-2\ii\left(\frac43{\xi(z)}^3+s\,\xi(z)+\frac{8}{3}\zeta_\beta^3\right)}}{1}, &z \in \Gamma_-;
\end{cases}
\end{equation}
The solution is {\em not determined} without specifying the boundary condition at $z\in\partial\D_c$ and, therefore, the solution can be written up to an undetermined holomorphic matrix $H_c(z)$ by
\begin{equation}\label{calPcrit}
{\cal P}(z) := \begin{cases}
{\rm e}^{-i\frac {8}3 \zeta_\beta^3 \sigma_3} H_{c}(z) \Pi(\xi(z)) \,{\rm e}^{i\frac {8}3 \zeta_\beta^3 \sigma_3}, & z\in{\rm Int}({\cal B})\cap\D_c,\\
 \pmtwo{0}{1}{-1}{0}{\rm e}^{-i\frac {8}3 \zeta_\beta^3 \sigma_3} H_{c}(z)\Pi(\xi(z))\, {\rm e}^{i\frac {8}3 \zeta_\beta^3 \sigma_3}  \pmtwo{0}{-1}{1}{0}, & z\in{\rm Ext}({\cal B})\cap\D_c.
\end{cases}
\end{equation}
Note that the jump conditions \eqref{errorwp} for ${\cal P}$ are satisfied {\em for any holomorphic invertible matrix $H_c(z)$},
if the matrix $\Pi$ solves the RHP with the following (in fact, only the first three) conditions. 
\begin{align}\label{RHPP2}\begin{cases}\displaystyle
\Pi_+ (\zeta) = \Pi_-(\zeta),   &\zeta\in\cal B;
\\\displaystyle
\Pi_+ (\zeta) = \Pi_-(\zeta) \pmtwo{1}{0}{{\rm e}^{  2i\left(\frac 43 \zeta^3 + s \zeta \right)} }{1}, &\zeta\in \Gamma_+;\vspace{0.1cm}\\\displaystyle
\Pi_+ (\zeta) = \Pi_-(\zeta) \pmtwo{1}{-{\rm e}^{ - 2i\left(\frac 43 \zeta^3 + s \zeta \right)} } {0}{1}, &\zeta\in \Gamma_-;
\\ \displaystyle \Pi(\zeta)= I + \mathcal O(\zeta ^{-1}), &\zeta\to\infty;
\end{cases}
\end{align}
See Figure \ref{xiplane} for the geometry of contours in $\zeta$-coordinate.  The name of the contours in the local coordinate are given to match the corresponding contours under the mapping $\xi$; So $\zeta\in{\cal B}$ implies that $\zeta\in\xi({\cal B})$ etc.

\paragraph{Remark.}
Without the extra term $H_c(z)$ in \eqref{calPcrit} (i.e. $H_c(z)=I$), we would have eventually found that 
\begin{equation}
{\cal P}(z) =
{\rm e}^{-i\frac {8}3 \zeta_\beta^3 \sigma_3} \Pi(\xi(z)) \,{\rm e}^{i\frac {8}3 \zeta_\beta^3 \sigma_3}=I + \mathcal O(N^{-1/3}) \ ,\qquad z\in \partial \mathbb D_c.
\end{equation}
This eventually determines the error of the strong asymptotics that turns out to be too large to see the leading effect in ${\rm Int}({\cal B})$.   This is why we introduce $H_c(z)$.

The solution to the RHP \eqref{RHPP2} can be found in the literature about ordinary orthogonal polynomials \cite{BI} and is related to a particular solution called the Hastings-McLeod solution of the Painlev\'e\ II equation.
Consequently one obtains
\begin{equation}\label{Psiseries}
\Pi(\zeta)= I+\frac{1}{2\ii\zeta}\left[\begin{array}{cc}u(s)&q(s)\\-q(s)&-u(s)\end{array}\right]+{\cal O}\left(\frac{1}{\zeta^2}\right),\quad
u(s)=q'(s)^2-sq(s)^2-q(s)^4,
\end{equation}
where $q(s)$ is the so called Hastings-McLeod solution of the Painlev\'e II equation: $q''(s)=sq(s)+2q(s)^3$. We recall the asymptotic behavior of $q(s)$:
\begin{equation}\label{HMsol}
q(s)=\begin{cases}\displaystyle
 {\rm Ai}(s)+{\cal O}\left(\frac{\e^{-\frac43 s^{3/2}}}{s^{1/4}}\right)= \frac{{\rm e}^{-\frac23 s^{3/2}}}{2\sqrt\pi s^{1/4}}\left(1+{\cal O}\left(s^{-3/2}\right)\right), &s\to+\infty;
 \\\displaystyle
 \sqrt{-\frac{1}{2}s}\left(1+{\cal O}\left(\frac{1}{s^3}\right)\right), & s\to-\infty.
 \end{cases}
\end{equation}

\paragraph{Strong asymptotics and error analysis:}
We follow the same construction of $H_c$ as that of $H_0$ in the post-critical case.
We define
\begin{equation}\begin{split}
H_c(z)=I&-\frac{1}{2\ii\xi(z)}\left[\begin{array}{cc}u(s)&q(s)\\-q(s)&-u(s)\end{array}\right]
\\ &+\frac{\gamma^*_{c}}{2N^{1/3}(z-b^*_c)}\left(\frac{z}{b^*_c}\right)^{\frac r2 \sigma_3}\left[\begin{array}{cc}u(s)&q(s)\\-q(s)&-u(s)\end{array}\right]\left(\frac{z}{b^*_c}\right)^{-\frac r2 \sigma_3},
\end{split}
\end{equation}
to exactly cancel the pole singularity of the first subleading term in \eqref{Psiseries}.    This is invertible for $\D_c$ small enough.  The conjugation by $({z}/{b^*_c})^{\frac r2 \sigma_3}$ is introduced to later match the growth behavior of the outer parametric $A^\infty$ at $z=\infty$ that we define below.  Then we have
\begin{equation}\label{hmmm}
{\rm e}^{-i\frac {8}3 \zeta_\beta^3 \sigma_3}H_c(z) \Pi(\xi(z)) {\rm e}^{i\frac {8}3 \zeta_\beta^3 \sigma_3}
=z^{\frac r2 \sigma_3} S^\infty(z)\,z^{-\frac r2 \sigma_3}\left(I+{\cal O}\left(N^{-2/3}\right)\right) ,\quad z\in\partial \D_c,
\end{equation}
where we have defined the following shorthand notation for convenience.
\begin{equation}
S^\infty(z):= I+\frac{\gamma^*_{c}}{2N^{1/3}(z-b^*_c)}\left[\begin{array}{cc}u(s)&q(s) (b^*_c)^{-r} {\rm e}^{-i\frac {16}3 \zeta_\beta^3} \\-q(s) (b^*_c)^{r}{\rm e}^{i\frac {16}3 \zeta_\beta^3}&-u(s)\end{array}\right].
\end{equation}

Now we define our strong asymptotics for $A(z)$ by
\begin{align}
A^\infty(z)&:=\begin{cases}\displaystyle
\Phi(z) \,S^\infty(z) ,  &z\in{\rm Int}({\cal B})\setminus\D_c,
\\\displaystyle
\Phi(z)\, z^{-\frac r2\sigma_3}  \left[\begin{array}{cc}0 & 1 \\-1 & 0\end{array}\right]z^{\frac r2 \sigma_3}S^\infty(z)\,z^{-\frac r2 \sigma_3} \,
  \left[\begin{array}{cc}0 & -1 \\1 & 0\end{array}\right]z^{\frac r2\sigma_3}, & z\in{\rm Ext}({\cal B})\setminus\D_c,
\\
\Phi(z)\,z^{-\frac r 2\sigma_3} {\cal P}(z) \,z^{\frac r 2 \sigma_3}, &z\in\D_c,
\end{cases}
\\\label{Acrit} &=\begin{cases}\displaystyle
\left(I+\frac{\gamma^*_{c}}{2N^{1/3}(z-b^*_c)}\left[\begin{array}{cc}-u(s)& \!\!\!\! q(s) (b^*_c)^{r}{\rm e}^{i\frac {16}3 \zeta_\beta^3}\\ -q(s) (b^*_c)^{-r} {\rm e}^{-i\frac {16}3 \zeta_\beta^3} &u(s)\end{array}\right]\right)\Phi(z), &z\notin\D_c,
\\
\Phi(z)\,z^{-\frac r 2\sigma_3} {\cal P}(z) \,z^{\frac r 2 \sigma_3}, &z\in\D_c.
\end{cases}
\end{align}

As in the pre-critical case, we define the error matrix by
\begin{equation}
{\cal E}(z):=A^\infty(z)\,A^{-1}(z) .
\end{equation}
Let us evaluate the jump of the error matrix on ${\rm Int}({\cal B})\cap\partial\D_c$.  Below, $+$ side is to the side of $\D_c$. 
\begin{align}
{\cal E}(z)_+({\cal E}(z)_-)^{-1}&=A^\infty(z)_+(A^\infty(z)_-)^{-1}
\nonumber
\\&=\Phi(z)\,z^{-\frac r2\sigma_3}{\cal P}(z)\,z^{\frac r2\sigma_3}(S^\infty(z))^{-1}\Phi(z)^{-1}
\nonumber
\\(\text{using \eqref{calPcrit}})~~ &=\left[\begin{array}{cc}0\! & 1\! \\\!\!\!-1 \! & 0\!\end{array}\right] z^{-\frac r2\sigma_3}
{\rm e}^{-i\frac {8}3 \zeta_\beta^3 \sigma_3} H_{c}(z) \Pi(\xi(z)) \,{\rm e}^{i\frac {8}3 \zeta_\beta^3 \sigma_3}
z^{\frac r2\sigma_3}(S^\infty(z))^{-1}\left[\begin{array}{cc}\!0\! & \!\!\!-1 \!\! \\\!1 \!& 0\!\!\end{array}\right] 
\nonumber
\\(\text{using \eqref{hmmm}})~~ &=\left[\begin{array}{cc}0\! & 1\! \\\!\!\!-1 \! & 0\!\end{array}\right] z^{-\frac r2\sigma_3}
z^{\frac r2 \sigma_3} S^\infty(z)\,z^{-\frac r2 \sigma_3}\left(I+{\cal O}\left(N^{-2/3}\right)\right)
z^{\frac r2\sigma_3}(S^\infty(z))^{-1}\left[\begin{array}{cc}\!0\! & \!\!\!-1 \!\! \\\!1 \!& 0\!\!\end{array}\right]
\nonumber
\\&=I+{\cal O}\left(N^{-2/3}\right),\quad \text{ for $z\in{\rm Int}({\cal B})\cap\partial\D_c$.}
\end{align}
A similar calculation gives the same error bound on ${\rm Ext}({\cal B})\cap\partial\D_c$.  The jump matrices of ${\cal E}$ on the other contours are exponentially close to the identity for large $N$ and, therefore, ${\cal E}$ satisfies a small-norm RHP and we conclude that 
\begin{equation}\label{postAfinal}
A(z)=\left(I+{\cal O}\left(N^{-2/3}\right)\right)A^\infty(z)\ ,
\end{equation}
uniformly in the whole complex plane.

To find the strong asymptotics of the orthogonal polynomial we need to undo the transformations that lead us from $Y(z)$ to $A(z)$ and then rewrite it in terms of the approximation $A^\infty(z)$.   We skip the calculations as they are similar to the ones for the pre-critical and the post-critical cases.  We note that the final form of the asymptotics contains $b_c^*$ and $\gamma_c^*$ which can subsequently be replaced by $b_c$ and $\gamma_c$ using the approximations \eqref{bstar} and \eqref{gammastar}. This gives Theorem \ref{thm-crit}.

To show the Proposition \ref{prop-O1} near $\beta$ we also need to evaluate $P_{n,N}$ inside $\D_c$, that involves the Painlev\'e II parametric $\Pi$.    We only need to use that $\Pi(\zeta)$ in \eqref{RHPP2} is bounded uniformly over $\zeta$ in any compact subsets of $\C$.

%
%
%
%
%
%
%
%
%

 \section{Proof of Theorem \ref{propnorm}}\label{sec-norm}

First we relate the norming constant of $P_{n,N}$ with respect to the weight $w_{n,N}(z)\,\d z$ with the norming constant, $h_n$, with respect to $\e^{-N Q(z)}\,\d A(z)$.
Let us define the former by (recall the definition of $\nu_j$ in \eqref{nu})
\begin{equation}
\widetilde h_n:= \oint_\Gamma P_{n,N}^2(z) w_{n,N} \d z = \frac {\det[\nu_{i+j}]_{0\leq i,j\leq n}}{\det[\nu_{i+j}]_{0\leq i,j\leq n-1}}.
\end{equation}
\begin{prop}
\label{norming}
We have (note $n=Nt+r$)
\begin{align}
 h_n &= -\frac {\Gamma(Nc+n+1)}{2i\,N^{Nc+n+1}} \frac {\widetilde h_n}{ P_{n+1,N}(0)}
 \\ &=
 \left(1+ \mathcal O\left(\frac1N\right)\right)i(c+t)^r\sqrt{\frac {\pi(c+t)}{2N}}{\rm e}^{N\big((t+c)\ln (t+c)-(t+c)\big)}\frac {\widetilde h_n}{ P_{n+1,N}(0)}\ .
\end{align}
\end{prop}
The proof can be found in Appendix \ref{app-hn}.

The importance of Proposition \ref{norming} relies in the fact that we can find the asymptotics for $h_n$ in terms of the RHP already set up. Indeed it follows from eq. (\ref{Ymatrix}) that 
\begin{equation}
\widetilde h_n =- 2\pi i  \lim_{z\to \infty} z^{n+1}\, [Y(z)]_{12} \label{tildenorms}
\end{equation}
Considering the post-critical case first, using \eqref{postA}, we have
\begin{equation}
\widetilde h_n = -2\pi \ii\lim_{z\to\infty}  z^{n+1} \left[Y(z) \right]_{12} =\left(1+{\cal O}\left(N^{-1}\right)\right) \frac{1}{i}\sqrt{\frac{2\pi}{N}}\gamma_1\beta^r e^{N t\ell},
\end{equation}
and using \eqref{Pnmpost} we have
\begin{equation}
P_{n+1,N}(0)=\left(1+{\cal O}\left(N^{-1}\right)\right) \sqrt{\frac{1}{2\pi N}}\gamma_1\beta^{r} e^{N t g(0)}=\left(1+{\cal O}\left(N^{-1}\right)\right) \sqrt{\frac{1}{2\pi N}}\gamma_1\beta^{r} e^{N t \ell}.
\end{equation}
Note that the leading order behavior of both $\widetilde h_n$ and $P_{n+1,N}$ come from the subleading order of $A^\infty$ \eqref{postA} which is of order ${\cal O}(N^{-1/2})$.
This gives $\widetilde h_n/P_{n+1,N}(0)\approx -2\pi \ii $ and, from Proposition \ref{norming} we obtain
\begin{equation}
h_n =  \left(1+ \mathcal O\left(N^{-1}\right)\right)2\pi\sqrt{\frac {\pi}{2N}} (c+t)^{r+1/2}{\rm e}^{N\big[(t+c)\ln (t+c)-(t+c)\big]}.
\label{normdough}
\end{equation}
Understanding that $\rho=\sqrt{c+t}$ and $\ell_\text{2D}=(t+c)\ln (t+c)-(t+c)$ for the post-critical (and critical) case, we obtain Theorem \ref{propnorm} for the post-critical case.  

The calculation for the critical case gives exactly the same result except that the error bound gets worse to ${\cal O}(N^{-1/3})$.  The leading terms of both $\widetilde h_n$ and $P_{n+1,N}$ are written in terms of the Hasting-McLeod solution of Painlev\'e II, however, they cancel when one computes $h_n$.  Note that the leading order behavior of both $\widetilde h_n$ and $P_{n+1,N}$ come from the subleading order of $A^\infty$ \eqref{Acrit} which is of order ${\cal O}(N^{-1/3})$.  Combined with the behavior of the error matrix in \eqref{postAfinal}, one can obtain the error bound of $h_n$ for the critical case.

Let us evaluate $h_n$ for the pre-critical case.  Using $A^\infty$ in \eqref{Ainf}, we obtain
\begin{align}
\widetilde h_{n}&=- 2i  \pi\lim_{z\to\infty}z^{n+1} \left[e^{tN \frac \ell 2\sigma_3}\left(I+{\cal O}\left(N^{-1}/(1+|z|)\right)\right)A^\infty(z) \,e^{tN (g(z)-\frac \ell 2)\sigma_3} \right]_{12}
\\&=-2\pi\ii \,(\alpha\rho)^r\sqrt{\kappa\rho} \, e^{Nt\ell} \left(1+{\cal O}\left(N^{-1}\right)\right).
\end{align}
The $z$-dependence in the error ${\cal O}\left(N^{-1}/(1+|z|)\right)$ comes because the non-trivial jumps of the error matrix ${\cal E}$ \eqref{calEpre} are on a {\em bounded} set of contours.
 
Using the identities (that follow from \eqref{Fz} and the last equation in \eqref{valS})
\begin{equation}\begin{split}
&\rho F(0)=|\beta|=\alpha\rho+\kappa/\rho=(c+t)\alpha/\rho,
\\ &\rho F'(0)=(|\beta|-\Re\beta)/(2|\beta|)=\kappa/(\alpha^2\rho+\kappa)=\frac{\rho\kappa}{\alpha^2}\frac{1}{c+t},
\end{split}
\end{equation}
we have
\begin{equation}\begin{split}
P_{n+1}(0)&=(1+{\cal O}(1/N))\sqrt{\rho F'(0)} (\rho F(0))^{r+1} e^{Nt g(0)}
\\ &=(1+{\cal O}(1/N))\sqrt{\frac{\rho\kappa}{\alpha^2}\frac{1}{c+t}} \frac{(c+t)^{r+1}\alpha^{r+1}}{\rho^{r+1}} e^{Nt g(0)}.
\end{split}
\end{equation}
Collecting the above results with Proposition \ref{norming} we obtain the following asymptotic result.
 \begin{align}
 h_n &=  \left(1+ \mathcal O\left(\frac1N\right)\right)i(c+t)^r\sqrt{\frac {\pi(c+t)}{2N}}{\rm e}^{N\big[(t+c)\ln (t+c)-(t+c)\big]}\frac {\widetilde h_n}{ P_{n+1,N}(0)}
 \\ & = \left(1+ \mathcal O\left(\frac1N\right)\right) 2\pi \sqrt{\frac{\pi}{2N}}\rho^{2r+1} e^{N(t \ell-t g(0)-(t+c)+(t+c)\log (t+c))}.
 \label{hn--}
 \end{align}

We further need the following result that also proves \eqref{eq:ell2D}.
\begin{lemma}\label{lem-ell2d}  Defining $\ell_\text{2D}:=2t \Re g(z)-Q(z)$ for $z\in\partial K$.  Then we have
\begin{equation}\label{714}
\ell_\text{2D}=t(\ell-\Re g(0))-(c+t)+(c+t)\log(c+t).
\end{equation}
Explicitly the constant is computed as 
\be
\ell_\text{2D} = \left\{ 
\begin{array}{cl} 
\ds 2(c+t) \log \rho - 2c \log \a + \left(\a^2 -1 \right)\left(a^2 - \frac {2 a \rho} \a\right) - \frac {\rho^2} {\alpha^2} &  \text{for } t<t_c;\\
(c+t) \log(c+t) - (c+t) & \text {for  } t\geq t_c
\end{array}
\right.
\label{140}
\ee
\end{lemma}
\begin{proof} We take the limit $z\to 0$ in  Lemma \ref{lem-UOP2d}. As $z\to 0$ we have $S_\text{back}(z)= (c+t)/z+{\cal O}(1)\to \infty$ and $S_\text{back}(z)\, z= c+t+{\cal O}(z)$.  From ${\cal U}_\text{OP}(z)=\Re V(z)-2 t \Re g(z)+t \ell$ and using Lemma \ref{lem-UOP2d} to rewrite it differently, we have
\begin{align}
&2{\cal U}_\text{OP}(z)=-2c\log a+2(c+t)\log|z|-4 t \Re g(0)+2t\ell +{\cal O}(z),
\\& {\cal U}(z)+{\cal U}(S_\text{back}(z))-|S_\text{back}(z)-\overline z|^2=-2c\log a-2t \Re g(0)+2\ell_\text{2D} 
\\&\qquad\qquad\qquad\qquad\qquad\qquad      -2(c+t)\log((c+t)/|z|)+2(c+t)+{\cal O}(z),
\end{align}
as $z\to 0$.  Equating the above two, we obtain \eqref{714} in the limit $z\to 0$.

{ We still need to provide the explicit expressions in \eqref{140}. 
We use the definition of $\ell_{2D} = 2t \Re g(z_0) - Q(z_0)$ and $z_0\in \partial \mathbf P(K)$.\\
{\bf Post-critical case $t>t_c$}. Choosing $z_0 = \sqrt{t+c}$ and plugging straighforwardly into the expression for $g(z)$ (1.21) and $\ell$ (1.22) yields the statement after elementary simplifications.\\
{\bf Pre-critical case $t\leq t_c$}. 
 Using the definition of $\mathcal U_{2D}$ \eqref{Phi2D} and the alternative expression \eqref{calU}
(recall that $\mathcal U$ coincides with $\mathcal U_\text{2D}$ outside of $K$)
\be
\mathcal U_{2D}= Q(z) - \frac 2 \pi \int_{K} \log |z-w| \d A(w)  + \ell_{2D} = |z|^2 - |z_0|^2 - 2 \Re \int_{z_0}^z S(\zeta) \d\zeta, 
\label{qwq}
\ee
one obtains that 
\be
\ell_{2D} = \frac 2 \pi \int_{K} \log |z-w| \d A(w) - 2 \Re \int_{z_0}^z S(\zeta) \d\zeta - 2c \log \frac 1{|z-a|} - |z_0|^2.
\ee
Since this is a constant, we take the limit $z\to \infty$ and, recalling that the area of $K$ is $\pi  t$,  we have 
\be
\ell_{2D} = \lim_{z\to \infty} \left\{
2(t+c) \log |z| - |z_0|^2  - 2 \Re \int_{z_0}^z S(\zeta) \d\zeta
\right\}. \label{qwq2}
\ee
In order to compute \eqref{qwq2} we use the uniformizing map $f(v)$ and choose $z_0 = f(1)$. Recalling that $S(f(v)) = f(1/v)$ (Lemma \ref {lem-preA}), we must then evaluate
\begin{align}\begin{split}
& \lim_{v\to +\infty} \left\{
2(t+c) \log f(v) - |f(1)|^2  - 2 \Re \int_{1}^v f\left(\frac 1 w \right) f'(w) \d w \right\} \\&= \lim_{v\to +\infty} \left\{
2(t+c) \log \rho  - |f(1)|^2  - 2  \Re \int_{1}^v \left(f\left(\frac 1 w \right) f'(w) - \frac {t+c} w\right) \d w \right\}. \label{qwq3}
\end{split}
\end{align}
In evaluating \eqref{qwq3}, it is necessary to use $c+t =  {\rho(\rho+\kappa \a^{-2})}$ (see \eqref{valS}); 
the integral is straightforward and yields a convergent expression (the integrand has behavior $\mathcal O(w^{-2})$)  and \eqref{qwq3} yields
\be
2(t+c) \log \rho  - \bigg(
\rho - \frac {\kappa}{1-\a} - \frac \kappa \a
\bigg)^2   -{\frac {2\kappa
  \left( \rho (1 - \a^2)^2+\kappa{\alpha}^{2} \right) }{{\alpha}
^{2} \left(1-\a^2 \right) ^{2} }}\log\alpha
+{
\frac { 2( -\rho+\rho\,{\alpha}^{2}+\kappa ) \kappa}{\alpha
\, \left( \alpha-1 \right) ^{2} \left( \alpha+1 \right) }}.
\ee
Using the relations \eqref{valS}, elementary manipulations imply the expression \eqref{140}.}
\end{proof}

\appendix
\renewcommand{\theequation}{\Alph{section}.\arabic{equation}}
\section{Proof of Lemma \ref{lem-preA}}\label{sec-proof-lem-preA}

Using the deck transformation \eqref{deck}, the injectivity of $f$ on $|v|>1$ is equivalent to the condition: ${\rm deck}(\c{\mathbb D}^c)\cap \c{\mathbb D}^c=\emptyset$.  Since the deck transformation is a M\"obius transformation and maps the unit circle to the circle that passes through ${\rm deck}(\pm1)$, the above injectivity is equivalent to the condition: $-1\leq{\rm deck}(\pm1)\leq 1$, which gives the equation \eqref{conformal-eq}.

Suppose that $f:\c{\mathbb D}^c\to K^c$ is a conformal mapping.   (In the last step of the proof, we will check \eqref{conformal-eq} to confirm this.)
Since $f(v)\in\partial K$ when $v\in\partial{\mathbb D}$ it follows that $S(z)=S(f(v))=f(1/v)=\overline{f(v)}=\overline z$ over $z\in \partial K$.  So the item i) in Definition \ref{def-S} is satisfied.  To satisfiy ii) and iii) we note that 
\begin{equation}\text{$f(1/v)\to \infty$ only when $v\to 1/\alpha$ or $v\to 0$.}
\end{equation}  It means that the full (multi-valued) analytic continuation of $S(z)$ over the Riemann surface has two poles at $f(1/\alpha)$ and $f(0)=0$.    To satisfy ii) in Definition \ref{def-S}, since $a\neq 0$, we have $f(1/\alpha)=a$.  Therefore, the conditions ii) and iii) in Definition \ref{def-S} are equivalent to the following three equations for the the three unknowns $\rho, \kappa, \alpha$:
\begin{equation}\label{valS}
\begin{cases}
&f(1/\alpha)=\frac{\rho\,\alpha^2 - \rho +\kappa}{\alpha (\alpha^2-1)} =a,
\\
&\res_{z=a}S(z)\,\d z=\res_{v=1/\alpha}f(1/v)\, f'(v) \d v=\frac{\kappa}{\alpha^2}\left(\rho+\frac{\kappa\,\alpha^2}{(1-\alpha^2)^2}\right)=c,
\\
&-\!\res_{z=\infty}S(z)\,\d z=-\!\res_{v=\infty} f(1/v)\,f'(v)\d v=\frac{\kappa\rho}{\alpha^2} +\rho^2=c+t.
\end{cases}
\end{equation}
We show that these equations have a unique solution within the proposed parameter space.
By eliminating $\rho$ and $\kappa$ from equations \eqref{valS} we obtain the equations \eqref{rhokappa},
and  $\alpha$ must satisfy
\begin{equation}
0=P(X):=X^3-\left(\frac{a^2+4c+2t}{2a^2}\right)X^2+\frac{t^2}{2a^4},\quad X:=\alpha^2.
\end{equation}
This is a third order equation in $X$ and its discriminant in $X$ is given by $t^2\Delta/(4a^{10})$ where 
\begin{equation}\label{disc}
\Delta:=(a^2+4c+2t)^3-27t^2a^2.
\end{equation}
This is strictly positive for any $c>0$ because $\partial_c\Delta=6(a^2+4c+2t)^2\geq 0$ and $\Delta|_{c=0}= (a^2-t)^2(a^2+8t)\geq 0$.  Therefore the equation \eqref{000} has three real roots among which exactly two are positive, because the sum of the three roots is positive and the product of the three roots is negative.   

For $a^2\leq t< t_c = a(a+2\sqrt c)$ we have $c> (a^2-t)^2/(4a^2)$, and 
\begin{equation}
P(0)=t^2/(2a^4)> 0;\quad P(1)=  -2\left[c-(a^2-t)^2/(4a^2)\right]/a^2< 0.
\end{equation}
So, there is exactly one real root of the polynomial $P(X)$ in $0< X< 1$ for $a^2\leq t<t_c$.  

When $t< a^2$, there may be two roots of $P(X)$ in $0<X<1$, however there is exactly one root of $P(X)$ in $0<X<t/a^2<1$ because $P(t/a^2)=-2c\,t^2/a^6<0$. We must choose this root to have $\kappa> 0$ because, if we choose the other root such that $X\geq t/a^2$ then $\kappa$ becomes non-positive by the factor $(t-a^2\alpha^2)$ in the numerator, see \eqref{rhokappa}.     

In summary, taking the above choice of the root of $P(X)$, $\alpha=\sqrt X\in (0,1)$ is the unique solution that satisfies the system of equations \eqref{valS} as well as the conditions $\rho>0$, $\kappa>0$ and $0<\alpha<1$.

Lastly, we need to check whether our solution satisfies \eqref{conformal-eq}.   From \eqref{rhokappa} it follows that
\begin{equation}\label{inside}
\frac{\kappa}{\rho}=(1-\alpha^2)\frac{t-a^2\alpha^2}{t+a^2\alpha^2}<1-\alpha^2.
\end{equation}
To conclude, the region $K$ defined through the conformal map $f$ has the Schwarz function that satisfies Definition \ref{def-S}.

\section{Proof of Lemma \ref{lem-sy}}\label{app-lem26}
\begin{proof}
Using the definition \eqref{Sback}, $S(z)+S_\text{back}(z)$ defines a rational function because it is continuous on ${\cal B}$; Crossing ${\cal B}$ the boundary values satisfy $[F(z)]_\pm=[{\rm deck}\circ F(z)]_\mp$ by Lemma \ref{lem-F}, and therefore $[S(z)]_\pm=[S_\text{back}(z)]_\mp$.
To determine this rational function we only need to identify its pole singularities.    
\begin{itemize}
\item[--] $S(z)=f(1/F(z))$ has a pole at $p$ only if {\bf i)} $F(p)=1/\alpha$ or {\bf ii)} $F(p)=0$.  
\item[--] $S_\text{back}(z)=f(1/{\rm deck}\circ F(z))$ has a pole at $p$ only if {\bf iii)} ${\rm deck}\circ F(p)=1/\alpha$ or {\bf iv)} ${\rm deck}\circ F(p)=0$. 
\item[--] {\bf i)} and {\bf iii)} are complementary, and $p=f(1/\alpha)$ for any of them.
\item[--] {\bf ii)} and {\bf iv)} are complementary, $p=f(0)=0$ for any of them.
\item[--] Since we know that $S(z)$ has a pole at $a$ with residue $c$, {\bf i)} must hold with $p=a$.
\item[--] {\bf ii)} is impossible because, $F(0)=|\beta|/\rho\neq 0$ by using the fact that $0$ is to the right of the branch cut ${\cal B}$.
\item[--] So {\bf iv)} holds, i.e. ${\rm deck}\circ F(0)=0$.
The residue of the pole of $S_\text{back}$ at $z=0$ is given by
\begin{equation}
\res_{z=0}S_\text{back}(z)\,\d z=\res_{v=0} f(1/v)\,\d f(v)=-\res_{u=\infty} f(1/u)\,\d f(u)=-\res_{z=\infty} S(z)\,\d z=t+c.
\end{equation}
\item[--] Lastly, $S_\text{back}(\infty)= f(1/{\rm deck}\circ F(\infty))=f(1/{\rm deck}(\infty))=f(1/\alpha)=a$.
\end{itemize}
The items above uniquely determine
\begin{equation}
S(z)+S_\text{back}(z)=a+c\frac{1}{z-a}+\frac{c+t}{z}.
\end{equation}
The difference satisfies
\begin{align}\nonumber
S(z)-S_\text{back}(z)&=(F(z)-{\rm deck}\circ F(z))\cdot (\text{rational map symmetric under $F(z)\leftrightarrow{\rm deck}\circ F(z)$})
\\ &=\sqrt{(z-\beta)(z-\c\beta)}\cdot (\text{rational function}).
\end{align}
Using the pole structure of $S$ and $S_\text{back}$ at $a,0,\infty$,  the rational function is determined such that $S(z)-S_\text{back}(z)=-y(z)$ has the pole structure as proposed.
\end{proof}

\section{Proof of Lemma \ref{lem-UOP2d} and Lemma \ref{lem-steepest}}\label{app-lem-UOP2d}
The proof is based on the anti-holomorphic involution symmetry of $v\leftrightarrow 1/\c v$.
Below, we use the definition $S_\text{back}(f(v))=f(1/{\rm deck}(v))$ \eqref{Sback}, and set $v=F(z)$.
\begin{align}\nonumber
\int_\beta^z y(w)\,\d w&=\int_\beta^z \left(S_\text{back}(\zeta)-S(\zeta)\right)\d \zeta
\\\nonumber &=\int_{v_\beta}^{v} f(1/{\rm deck}(u))\,\d f({\rm deck}(u)) +\int_{v}^{v_\beta} f(1/u)\,\d f(u)
\\\nonumber \left( \atopfrac{{\rm deck}(v)\to w;}{\text{ Note }{\rm deck}(v_\beta)=v_\beta}\right)\quad &=\int_{v_\beta}^{{\rm deck}(v)} f(1/w)\,\d f(w) +\int_{v}^{v_\beta} f(1/u)\,\d f(u)
\\ &=\int_{v}^{{\rm deck}(v)} f(1/u)\,\d f(u) .
\end{align}
Choosing any point $v_0\in\partial\D$ the last equation can be written by integration by parts:
\begin{align}\label{betazy}
\int_\beta^z y(w)\,\d w&=\left(\int_{v}^{v_0}+\int_{v_0}^{{\rm deck}(v)}\right) f(1/u)\,\d f(u)
\\&=\int_{v}^{v_0}f(1/u)\,\d f(u)+ \int_{1/{\rm deck}(v)}^{1/v_0}f(1/u)\,\d f(u) +\bigg[f(1/v)\,f(v)\bigg]^{{\rm deck}(v)}_{v_0}.
\end{align}
Using $1/v_0=\overline v_0$ and $f(1/{\rm deck}(v))=S_\text{back}(z)$, we have (defining $z_0:=f(v_0)$)
\begin{align}\nonumber
2{\cal U}_\text{OP}(z)&=2\Re\int_\beta^z y(s)\,\d s=2\Re\left(\int_z^{z_0}+\int^{z_0}_{\overline{S_\text{back}(z)}}\right)S(z)\,\d  z+z\,S_\text{back}(z)+\overline z\,\overline{S_\text{back}(z)}-2|z_0|^2
\\&={\cal U}(z)+{\cal U}(\overline{S_\text{back}(z)})-(z-\overline{S_\text{back}(z)})(\overline z-S_\text{back}(z)).
\end{align}
This proves Lemma \ref{lem-UOP2d}.  

Below, we present the proof of Lemma \ref{lem-steepest}.  We remark that, if $v=F(z)$ satisfies ${\rm deck}(v)=1/\overline v$ then the left hand side of \eqref{betazy} is purely real, which can be confirmed by
\begin{align}
\int_\beta^z y(w)\,\d w&=\int_{v}^{v_0}f(1/u)\,\d f(u)+ \int_{1/{\rm deck}(v)}^{1/v_0}f(1/u)\,\d f(u) +\bigg[f(1/v)\,f(v)\bigg]^{{\rm deck}(v)}_{v_0}
\\&=\int_{v}^{v_0}f(1/u)\,\d f(u)+ \int_{\overline v}^{\overline{v_0}}f(1/u)\,\d f(u) +f(\overline v)\,f({\rm deck}(v))-f(\overline{v_0})f(v_0)
\\&=-2\Re \int_{z_0}^z S(z)\,\d z+|z|^2-|z_0|^2={\cal U}(z).
\end{align}
Since $b$ satisfies ${\rm deck}(F(b))=1/\c{F(b)}$, we have $\phi(b)={\cal U}(b)$.  Since a point on the steepest (descent or ascent) lines from $\beta$ has purely real $\phi$ value, $b$ must be on a steepest line from $\beta$.

\section{Proof of Proposition \ref{nonzero}}\label{app-nonzero}
Up to a sign due to a permutation of columns we have
\begin{equation}
\det [\nu_{i+j}]_{0\leq i,j\leq n-1} =(-1)^{n(n-1)/2}  \det \left[\oint_{\Gamma} z^{i-j} \frac {(z-a)^{N c} {\rm e}^{-N a z} \d z}{z^{N c+1}}\right]_{0\leq i,j\leq n-1}.
\end{equation}
All the determinants in this section are for $n\times n$ matrices, i.e., $0\leq i,j\leq n-1$.
Due to  Lemma \ref{lemma:moments} the second determinant is given by 
\begin{equation}
\det \left[\oint_{\Gamma} z^{i-j} \frac {(z-a)^{N c} {\rm e}^{-N a z} \d z}{z^{N c+1}}\right] = \det \left[\iint_\C z^i \c{(z-a)}^j |z-a|^{2Nc}{\rm e}^{-N |z|^2} \d A(z)\right] \prod_{k=0}^{n-1} \frac {2i N^{Nc +k +1}}{\Gamma(Nc+k+1)}.\label{deteq}
\end{equation}
Finally, the determinant on the right-hand side is strictly  positive because 
\begin{equation}
 \det \left[\iint_\C z^i \c{(z-a)}^j |z-a|^{2Nc}{\rm e}^{-N |z|^2} \d A(z)\right]  =  \det \left[\iint_\C z^i \c{z}^j |z-a|^{2Nc}{\rm e}^{-N |z|^2} \d A(z) \right] >0 
\end{equation}
where the equality is due to the fact that the columns of the two matrices are related by a unimodular triangular matrix, whereas the inequality follows from the positivity of the measure. Finally, since $\Gamma(z)$ has no zeros (and no poles since $Nc+k+1\geq 1$), the non-vanishing follows from (\ref{deteq}). 
\section{Proof of Proposition \ref{norming}}\label{app-hn}
From the determinantal expressions for $P_{n,N}$, defining $\mu_{ij}:= \iint z^i \c z^j  |z-a|^{2Nc}{\rm e}^{-N|z|^2}\d^2 z$, we have
\begin{equation}
h_n = \frac {\det[\mu_{ij}]_{0\leq i,j\leq n}}{\det[\mu_{ij}]_{0\leq i,j\leq n-1}}.
\end{equation}
The denominator already appeared in  the proof of Proposition \ref{nonzero} and is given by
\begin{equation}
\det[\mu_{ij}]_{0\leq i,j\leq n-1} = (-1)^{n(n-1)/2}\det[\nu_{i+j}]_{0\leq i,j\leq n-1} \prod_{k=0}^{n-1} \frac{\Gamma(Nc+k+1)}{2iN^{Nc+k+1}}\ .
\label{uno}
\end{equation}
For the numerator in $\widetilde h_n$ we have 
\begin{align}\nonumber
\det[\nu_{i+j} ]_{0\leq i,j\leq n} &= \det\left[\oint_{\Gamma} z^{i+1 + (j-n)} \frac {(z-a)^{Nc}{\rm e}^{-Naz}}{z^{Nc+1}}\d z \right]_{0\leq i,j\leq n}  \\\nonumber
&= (-1)^{n(n+1)/2}\prod_{\ell=0}^{n} \frac{2iN^{Nc+\ell+1}}{\Gamma(Nc+\ell+1)}  \det\left[\int z^{i+1} (\c{z-a})^{k}  |z-a|^{2Nc}{\rm e}^{-N|z|^2} \d A(z) \right]_{  0\leq i,k \leq n} \\\nonumber
&=(-1)^{n(n+1)/2}\prod_{\ell=0}^{n} \frac{2iN^{Nc+\ell+1}}{\Gamma(Nc+\ell+1)} \det[\mu_{ik}]_{\atopfrac{1\leq i\leq n+1}{0\leq k\leq n}} 
\\&=(-1)^{n(n+1)/2}
\prod_{\ell=0}^{n} \frac{2iN^{Nc+\ell+1}}{\Gamma(Nc+\ell+1)} (-1)^{n+1} P_{n+1,N}(0) \det[\mu_{ik}]_{0\leq i,k\leq n}
\end{align}
Dividing both sides by $\det[\nu_{i+j}]_{0\leq i,j\leq n-1}$ and using eq. (\ref{uno}) we prove the first equality in Proposition \ref{norming}.
\begin{equation}
\widetilde h_n  =\frac{\det[\nu_{i+j} ]_{0\leq i,j\leq n} }{\det[\nu_{i+j} ]_{0\leq i,j\leq n-1} }  =- \frac {2i N^{Nc+n+1}}{\Gamma(Nc+n+1)}  P_{n+1,N}(0)\, h_n
\end{equation}
The last equality
comes from the Stirling approximation formula
\begin{equation}
\Gamma(z) = \sqrt{\frac {2\pi}{z}} z^z {\rm e}^{-z} \left(1 + \mathcal O\left(z^{-1}\right)\right).
\end{equation}
We find (recall that $n=tN+r$)
\begin{align}\nonumber
\frac { \Gamma(N(c+t)+r + 1)}{  N^{N(c+t) +r+1 }    }    
 &=\left(1+ \mathcal O\left(\frac1N\right)\right)(c+t)^{r+1}\frac {\Gamma(N(c+t))}{  N^{N(c+t)}    } 
 \\&= \left(1+ \mathcal O\left(\frac1N\right)\right)(c+t)^{r}\sqrt{\frac {2\pi(c+t)}{N}}{\rm e}^{N\big[(t+c)\ln (t+c)-(t+c)\big]}.
\end{align}

\bibliographystyle{abbrv}
\bibliography{bratwurst_ref}


\end{document}